\documentclass[english]{article}
\usepackage{babel}

\usepackage{authblk}
\usepackage{graphicx,mathrsfs,enumerate,epstopdf}
\usepackage[utf8]{inputenc}
\usepackage{csquotes}
\usepackage{mathtools}
\usepackage{amsfonts,amstext}
\usepackage{amsmath,amsthm,amssymb,tensor,diagbox,multirow}
\usepackage[hidelinks]{hyperref}\hypersetup{pdfpagemode=UseNone,pdfstartview=}
\usepackage[left=2cm,right=2cm]{geometry}
\usepackage{xcolor}
\usepackage{url}
\usepackage{times}
\usepackage{soul}
\usepackage{dsfont}
\usepackage[ruled,vlined,linesnumbered]{algorithm2e}
\usepackage[sorting=none, date=year, style=nature, doi=false,  isbn=false, url=false]{biblatex} 
\bibliography{notas} 

\makeatletter

\AtEveryBibitem{\clearfield{month}}
\AtEveryBibitem{\clearfield{day}}
\AtEveryBibitem{\clearfield{issn}}

\newbibmacro{string+doi}[1]{%
  \iffieldundef{doi}{#1}{\href{http://dx.doi.org/\thefield{doi}}{#1}}}
\DeclareFieldFormat{title}{\usebibmacro{string+doi}{\mkbibemph{#1}}}
\DeclareFieldFormat[article]{title}{\usebibmacro{string+doi}{\mkbibquote{#1}}}

\newtheorem{dfn}{Definition}
\newtheorem{thm}{Theorem}

\newtheorem{lem}[thm]{Lemma}

\newtheorem{obs}[thm]{Observation}

\newcommand{\bra}[1]{\langle #1|}
\newcommand{\ket}[1]{|#1\rangle}
\newcommand{\braket}[2]{\langle #1|#2\rangle}
\newcommand{\ketbra}[2]{| #1 \rangle \langle #2 |}

\newcommand{\tr}{ {\rm tr} }

\usepackage[caption=false]{subfig}

\makeatother

\usepackage{babel}
\begin{document}
\title{Fragmented imaginary-time evolution for {early-stage} quantum signal processors}

\author[1,2,*]{Thais L. Silva}
\author[2,3]{M\'arcio M. Taddei}
\author[1,4]{Stefano Carrazza}
\author[1,2]{Leandro Aolita}

\affil[1]{Quantum Research Centre, Technology Innovation Institute, Abu Dhabi, UAE} 
\affil[2]{Federal University of Rio de Janeiro, Caixa Postal 68528, Rio de Janeiro, RJ 21941-972, Brazil}
\affil[3]{ICFO - Institut de Ciencies Fot\`oniques, The Barcelona Institute of Science and Technology, 08860, Castelldefels, Barcelona, Spain}
\affil[4]{TIF  Lab,  Dipartimento  di  Fisica,  Universit\`a  degli  Studi  di  Milano  and  INFN  Sezione  di  Milano,  Milan, Italy}
\affil[*]{thaisdelimasilva@gmail.com}
\date{}
%%%%%%%%%%%%%%%%%%%%%%%%%%%%%%%%%%%%%%%%%%%%%%%%%%%%%%%%
\maketitle

\begin{abstract}
Simulating quantum imaginary-time evolution (QITE) is a significant promise of quantum computation. However, the known algorithms are either probabilistic (repeat until success) with unpractically small success probabilities or coherent (quantum amplitude amplification)  with circuit depths and ancillary-qubit numbers unrealistically large in the mid-term. Our main contribution is a new generation of deterministic, high-precision QITE algorithms that are significantly more amenable experimentally. A surprisingly simple idea is behind them: partitioning the evolution into a sequence of fragments that are run probabilistically. It causes a considerable reduction in wasted circuit depth every time a run fails. Remarkably, the resulting overall runtime is asymptotically better than in coherent approaches, and the hardware requirements are even milder than in probabilistic ones. Our findings are especially relevant for the early fault-tolerance stages of quantum hardware.
\end{abstract}
%%%%%%%%%%%%%%%%%%%%%%%%%%%%%%%%%%%%%%%%%%%%%%%%%%%%%%%%

%%%%%%%%%%%%%%%%%%%%%%%%%%%%%%%%%%%%%%%%%%%%%%%%%%%%%%%%
%%%%%%%%%%%%%%%%%%%%%%%%%%%%%%%%%%%%%%%%%%%%%%%%%%%%%%%%
\section{Introduction}
\label{sec:intro}  

{Given a Hamiltonian $H$ and an inverse temperature $\beta\geq0$, QITE is the task of evolving quantum states according to the non-unitary propagator $e^{-\beta H}$.} QITE is central not only to ground-state optimisations \cite{VariationalQuITE19,motta2020,Gomes20,Sunetal21,Nishi21} but also to partition-function estimation and quantum Gibbs-state sampling \ \cite{PW09, BB10,temme2011,YA-G12,KB16,chowdhury2017,QSDP17,brandao2017,vanApeldoorn2020quantumsdpsolvers,gilyen2019,variationalGibbssampler20,KKB20,CSS21}, i.e. the task of preparing thermal quantum states at tunable inverse temperature $\beta$. This is both fundamentally relevant and useful for notable algorithmic applications. For instance, even though approximating ground states of generic Hamiltonians is not expected to be efficient  even on a quantum computer -- as it can solve QMA-complete problems \cite{QMAcomp06} --, significant speed-ups over classical simulations are possible. This has motivated several ground-state cooling algorithms (with and without QITE), especially for combinatorial optimisations \cite{AdibaticQC00,QAOA14, Qbacktracking15,motta2020,QAOAGoogleExp21} or molecular electronic structures \cite{VQE14,IonVQE18,adapativeVQE19,VariationalQuITE19}. On the other hand, Gibbs-state samplers are used as main sub-routines for quantum semi-definite program solvers \cite{brandao2017,QSDP17,vanApeldoorn2020quantumsdpsolvers} or for training \cite{KW17,WW19,variationalQRBM20} quantum machine-learning models {\cite{Biamonte17, QBoltzmann18}, e.g}. Moreover, QITE also enables quantizations \cite{motta2020} of the METTS or Lanczos algorithms, which directly simulate certain thermal properties without Gibbs-state sampling.

Quantum Gibbs states can be approximated by quantum Metropolis Markov-chains \cite{temme2011,YA-G12} or by variational circuits trained to minimis{e the free energy \cite{variationalGibbssampler20}, e.g.}
However, the former involve deep and complex circuits, whereas the latter are highly limited by the variational Ansatz. In turn, heuristic QITE algorithms for ground-state optimisations exist \cite{VariationalQuITE19,motta2020,Gomes20,Sunetal21,Nishi21,benedetti2021,PRXQuantum.2.010342,cao2021quantum}. There, one simulates pure-state QITE  with a unitary circuit that depends on the input state, the Hamitonian, and $\beta$. For small-$\beta$ steps, one can determine the circuits by measurements on the input state at each step and classical post-processing. One possibility is to optimise a variational circuit on the measured data \cite{VariationalQuITE19}, but this is again limited by the expressivity of the Ansatz. Another possibility is to invert a linear system generated from the measurements \cite{motta2020,Gomes20,Sunetal21,Nishi21}, but the size of such system (as well as the number of measurements required) is exponential in the number of qubits, unless restrictive locality assumptions are made. 

%%%%%%%%%%%%%%%%%%
\begin{figure*}[t!]
\centering{}\includegraphics[width=1.0\textwidth]{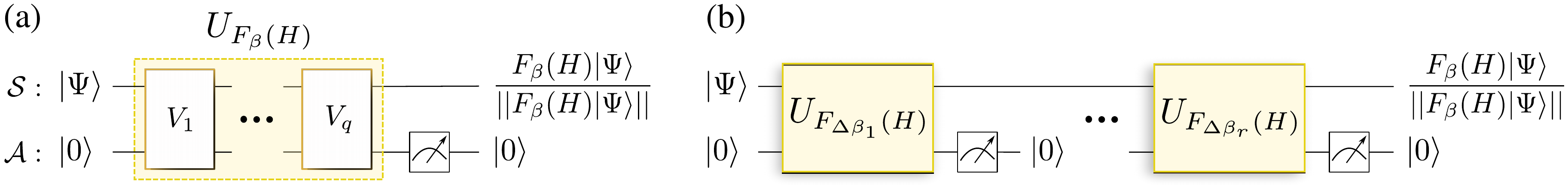}
\caption{{\bf High-level schematics of our algorithms}. a) QITE primitives: 
A system register $\mathcal{S}$ carries the input state $\ket{\Psi}$, whereas an ancillary register $\mathcal{A}$ is initialised in a computational-basis state $\ket{0}$. A unitary transformation $U_{F_{\beta}(H)}$, composed of a sequence $\{V_k\}_{\in[q]}$ of $q$ gates, with $[q]\coloneqq\{1, \hdots q\}$, is applied and then the ancillas are measured. Each gate makes one query to the Hamiltonian oracle (not shown). The specific choice of gates in the sequence is such that,
%determined by a powerful formalism for operator-function design, known as quantum signal processing, which allows us to synthesise the QITE propagator $F_{\beta}(H)\coloneqq e^{-\beta (H-\lambda_{\rm min})}$ efficiently.More precisely, 
conditioned on detecting $\ket{0}$ on the ancillas, the desired state $\frac{F_{\beta}(H)\ket{\Psi}}{\left\Vert F_{\beta}(H)\ket{\Psi}\right\Vert}$ is output up to controllable error. We refer to the circuit generating $U_{F_{\beta}(H)}$ as a QITE primitive.
%We develop two QITE primitives, one for each oracle model. Both primitives feature excellent scalings of $q$ with $\beta$ and the error. 
b) Master QITE algorithms: The post-selection probability -- given approximately by $p_\Psi(\beta)=\left\Vert F_{\beta}(H)\ket{\Psi}\right\Vert^{2}$ -- can decrease with $\beta$ very fast. Hence, for high $\beta$, probabilistic approaches based on repeat-until-success fail for the vast majority of trials. %, thereby producing an enormous waste in runtime.
In turn, coherent approaches based on quantum amplitude amplification provide a close-to-quadratic runtime speed-up, but at the expense of enormous circuit depths. %, among other disadvantages.
In contrast, we introduce a master algorithm that concatenates $r$ QITE fragments of inverse temperatures $\{\Delta\beta_l\}_{l\in [r]}$, with $\sum_{l\in [r]}\Delta\beta_l=\beta$ and $\beta_l<\beta$ for all $l\in [r]$. Each fragment is successively run probabilistically and has both a success probability significantly higher and a query complexity significantly lower than that of the entire evolution run at once. 
%{Indeed, the ancillary measurements in b) detect potential errors in the operator synthesis after each fragment, instead of just once at the end in a).} 
This ends up yielding an enormous saving in overall runtime (even beating coherent approaches for high $\beta$) while at the same time preserving all the practical advantages of probabilistic approaches for {experimental} implementations.
}
\label{fig:general_mindset}
\end{figure*}
%%%%%%%%%%%%%%%%%%%%%%%%%%%%%%%%%%%%%%%%%%
%%%%%%%%%%%%%%%%%%%%%%%%%%%%%%%%%%%%%%%%%%
 The most general, guaranteed-precision QITE algorithms are based on unitary circuits followed by ancillary-qubit post-selection {\cite{PW09, chowdhury2017, vanApeldoorn2020quantumsdpsolvers, gilyen2019, CSS21}}. These circuits -- {to which we refer as} \emph{QITE primitives} -- are efficient in $\beta$ as well as in the target precision. However, due to the {intrinsically} probabilistic post-selection, they must be applied multiple times  {-- by what we refer to as \emph{master QITE algorithms} -- t}o obtain a deterministic output. 
Repeat-until-success master algorithms apply the primitive in parallel (i.e. in independent probabilistic runs), thereby not inducing any increase in circuit depth. However, their overall complexity is inversely proportional to the post-selection probability. Instead, coherent master algorithms \cite{PW09, chowdhury2017, vanApeldoorn2020quantumsdpsolvers, gilyen2019}, based on amplitude amplification \cite{Brassard2002}, have close-to-quadratically smaller overall complexity. However, they require enormous circuit depths and significantly more ancillas. In addition, no fundamental efficiency limit for generic QITE algorithms
is known.

\subsection{Overview}

Here, we introduce two efficient QITE primitives based on {the} quantum signal processing (QSP) {framework} \cite{low2017optimal, Low2019hamiltonian, gilyen2019,vanApeldoorn2020quantumsdpsolvers} as well as a practical master QITE algorithm (see Fig.\ \ref{fig:general_mindset}); and prove a universal lower bound for the complexity of QITE primitives that can be seen as an imaginary-time counterpart of th{e n}o fast-forwarding theorem for RTE \cite{Berry2007,Berry2014,Berry2015a}. {The} first primitive is designed for Hamiltonians given in the well-known block-encoding oracle model, whereas the second one for a simplified model of real-time evolution oracles involving a single time. Both primitives feature excellent query complexity (number of oracle calls) and ancillary-qubit overhead. In fact, for the first primitive the complexity is sub-additive in $\beta$ and $\log(\varepsilon^{-1})$, with $\varepsilon$ the tolerated error. This scaling saturates our universal bound  when $\beta\ll\log(\varepsilon^{-1})$. Hence Primitive 1 is optimal in that regime, which, interestingly, turns out crucial for our master algorithm. In contrast, Primitive 2's complexity is multiplicative in $\beta$ and $\log(\varepsilon^{-1})$, but it requires a single ancilla throughout and its oracle significantly fewer gates. This is appealing for intermediate-scale quantum hardware. 
In turn, our master QITE algorithm breaks the evolution into small-$\beta$ fragments and runs each fragment's primitive probabilistically. Surprisingly, this yields an overall runtime competitive with -- and, in the {relevant} regime of high $\beta$, even better than -- that of coherent approaches while, at the same time, preserving all the advantages of probabilistic ones for {experimental} feasibility.

Finally, the complexity of our maste{r a}lgorithm depend{s o}n the \emph{fragmentation schedule}, i.e. number $r$ of fragments and their relative sizes. {On one hand, for Primitive 1, we rigorously prove that, from a critical inverse temperature $\beta_{\rm c}= \mathcal{O}(2^{N/2}\, N)$ on, the runtime is lower than that with coherent QITE. This is shown by explicitly constructing schedules with only $r=2$ fragments that do the job, remarkably. On the other hand, that fragmented QITE outperforms coherent QITE is also observed for both primitives through extensive numerical evidence. More precisely, we study the overall runtime as a function of $\beta$ and $\varepsilon$, up to $N=15$ qubits, and for numerically-optimized schedules. These experiments involve random instances of Hamiltonians encoding four computationally} hard classes of problems: Ising models associated to the $i$) MaxCut and $ii$) weighted MaxCut problems \cite{AdibaticQC00,QAOA14, Qbacktracking15}; $iii$) restricted quantum Boltzmann machines (transverse-field Ising models)  \cite{Biamonte17, QBoltzmann18}; and $iv$) a quantum generalization (fully-connected Heisenberg models) of the Sherrington-Kirkpatrick model \cite{SK75,Pachenco12} for spin glasses. We see a clear trend whereby, fro{m $\beta_{\rm c}= \mathcal{O}\big(2^{N/2}\big)$ on,  fragmentation outperforms coherent QITE for both primitives, for an optimal number of fragments $r\lesssim 6$.  
The obtained values for $\beta_{\rm c}$ imply that our algorithm outperforms coherent QITE in the computationally hardest range of $\beta$, particularly relevant for Hamiltonians with an exponentially small spectral gap  \cite{Exp_gap_NP_comp,Gap_Anderson_loc}}. Moreover, {impressively, such advantages are attained at no cost in circuit depth or number of ancillas, which are identical to those of  probabilistic QITE}. It is worth noting that, although we prove that fragmented QITE can outperform the coherent algorithm, it does not mean that its scaling is better (see the Supplementary Material \cite[Sec. VII]{TLS2022Sup}).  

\section{Results}

 We consider an $N$-qubit system $\mathcal{S}$, of Hilbert space $\mathbb H_\mathcal{S}$.  
We denote by $\mathbb H_\mathcal{A}$ the Hilbert space of an ancillary register $\mathcal{A}$.
We first discuss the primitives, then the universal complexity lower bound, and the master algorithm at last. Formal definitions and proofs of theorems are found in Methods.
%%%%%%%%%%%%%%%%%%%%%%%%%%%%%%%%%%%%%%%%%%%%%%%%%%%%%%%%
%%%%%%%%%%%%%%%%%%%%%%%%%%%%%%%%%%%%%%%%%%%%%%%%%%%%%%%%
\subsection{Quantum imaginary-time evolution primitives}
\label{sec:primitives}
 {We use the notation $(\beta,\varepsilon',\alpha)$-QITE-primitive to refer to a circuit that implements a block-encoding of the QITE propagator, i.e.  a unitary $U_{F_\beta(H)}$ acting on $\mathcal{S}$ and $\mathcal{A}$ containing an $\varepsilon'$-approximation of $\alpha F_\beta (H)$ as one of its matrix blocks, with $0\leq\alpha\leq1$ a subnormalization factor, $F_\beta (H):=e^{-\beta(H-\lambda_{\text{min}})}$, and $\lambda_{\text{min}}$ the minimal eigenvalue of $H$.} When applied to a state $\ket{\Psi}_\mathcal{S} \ket{0}_\mathcal{A}$, the primitive (approximatelly) produces the target state $\frac{F_{\beta}(H)\ket{\Psi}}{\left\Vert F_{\beta}(H)\ket{\Psi}\right\Vert}$ on the system after postselecting the ancillas in $\ket{0}_\mathcal{A}$. The postselection success probability is given by $p_{\Psi}(\beta, \alpha)=\alpha^2\, \left\Vert F_{\beta}(H)\ket{\Psi}\right\Vert^2$. The trace-distance error in the output-state is $\mathcal{O}(\varepsilon)$ if the spectral error in the primitive is $\varepsilon^{\prime}\leq\,\varepsilon\, \sqrt{p_{\Psi}(\beta, \alpha)}/2$ \cite[Sec. II]{TLS2022Sup}.
 
 We introduce two QITE primitives.
Both of them possess the basic structure shown in Fig.\ \ref{fig:general_mindset} a), where a sequence of  gates  $\{V_k\}_{\in[q]}$, with $[q]\coloneqq\{1, 2, \hdots q\}$, generates an approximate block-encoding of $F_\beta (H)$. The circuit acts on the system,  block-encoding ancillas and at most one extra qubit ancilla. The approximation consists of truncating an expansion of the exponential function at finite order $q$. Each gate $V_k$ makes one call to the oracle of $H$ (or its inverse) and contains $\mathcal{O}(1)$ parameterized single qubit rotations.  The parameters of this gates are determined by the function expansion using quantum signal processing \cite{low2017optimal, Low2019hamiltonian, gilyen2019,vanApeldoorn2020quantumsdpsolvers}. Conceptually, the two primitives differ in the kind of expansion and the type of oracle. Their circuit descriptions are given in the Methods, especially in Fig. \ref{fig:generalf}. 

{The first primitive implements a Chebyshev expansion using a block-encoding oracle $O_1$, i.e a unitary that has $H$ as one of its blocks. We denoted by $|\mathcal{A}_{O_1}|$  the ancillary-register size and by $g_{O_1}$ the gate complexity of $O_1$. In Methods, we prove the following.}   
\begin{thm}\label{teoQuITEbe} (QITE primitive using Chebyshev approximation and block-encoding
oracles). {Given $0<\varepsilon'<1$ and $\beta>0$, there is a circuit $P_{1}$ that is a $(\beta,\varepsilon',1)$-QITE-primitive using}
\begin{equation}
\label{eq:BEquery}
q_{{1}}(\beta,\varepsilon') = \mathcal{O}\left(\frac{e\,\beta}{2}+\frac{\ln(1/\varepsilon')}{\ln\left[e+2\ln(1/\varepsilon')/(e\,\beta)\right]}\right)
\end{equation}
queries to $O_1$ and $O_1^\dagger$, 
 $|\mathcal{A}_{1}|=|\mathcal{A}_{O_1}|+1$ total ancillary qubits, and gate complexity  $g_{P_1}=\mathcal{O}(g_{O_1}+|\mathcal{A}_{O_1}|)$\ per query. Moreover, the classical run-time to calculate the gates of $P_1$ is $\mathcal{O}({\rm poly}\big(q_1(\beta,\varepsilon')\big)$.
\end{thm}

{A nice feature of Eq.\ \eqref{eq:BEquery} is its sub-additivity in $\beta$ and $\ln(1/\varepsilon')$}. We note that a QITE primitive {was obtained in \cite{gilyen2019} that} works for the same oracle model 
%(using approximation-theory bounds \cite{SV14}) 
and has complexity {upper-bounded by} $\mathcal{O}\big(\sqrt{2\,\max[e^2\,\beta , \ln(2/\varepsilon')]\,\ln(4/\varepsilon')}\big)$. 
%(and a similar scaling had also obtained in \cite{chowdhury2017} but with stronger assumptions on the oracle)
{This} is asymptotically better in $\beta$ tha{n E}q.\ \eqref{eq:BEquery}, {but} it underperforms it for all $\beta\lesssim 8\ln(4/\varepsilon')$. 
%On the other hand,
%its dependence on $\log(1/\varepsilon')$ is multiplicative instead of additive.  As a result, 
{In particular, while Eq.\ \eqref{eq:BEquery} tends to zero for $\beta\to0$, the bound from Ref. \cite{gilyen2019} tends to $\mathcal{O}\big(\ln(1/\varepsilon')\big)$. Interestingly, the strict upper bound that we obtain in Methods is the expression within $\mathcal{O}( )$ in Eq.\ \eqref{eq:BEquery} up to a modest factor: 8. Moreover, in \cite[Sec. V]{TLS2022Sup}  we numerically verify that that expression is itself a valid bound (no extra factor), even for low $\beta$. Most importantly}, in Sec.\ \ref{sec:cooling_rate_bound} we show that {it approaches the optimal scaling} as $\beta$ decreases relative to {$\ln(1/\varepsilon')$. We stress that t}he latter regime is crucial for the master algorithm of Sec.\ \ref{sec:master_alg}, whose first fragments require{, precisely,} low inverse temperatures and high precisions.  {In turn, in the opposite regime of high $\beta$, preliminary numerical observations {\cite{Lucas2022}} suggest that the asymptotic scaling of the exact value of $q_1$ could actually be as good as $q_1(\beta,\varepsilon')=\mathcal{O}\big(\sqrt{\beta\,\ln(1/\varepsilon')}\big)$, i.e. similar to that from \cite{gilyen2019}.}

{The second primitive implements a Fourier expansion assuming access to a unitary oracle $O_2$, with gate complexity $g_{O_2}$, that contains the time evolution $e^{-iHt}$ at time $t=\frac{\pi}{2}\big(1+\frac{\gamma}{\beta}\big)^{-1}$.  }
In Methods, we prove the following.
\begin{thm}\label{teoQuITEberealtime} (QITE primitive using Fourier approximation nd single real-time
evolution oracles). {Given $0<\varepsilon'<1$ and $\beta>0$, there is a    
 $(\beta,\varepsilon',\alpha)$-QITE-primitive $P_{2}$ with $\alpha=e^{-\beta(1+\lambda_{\rm min})-\gamma}$, it uses  
\begin{equation}
\label{eq:real_time_or}
q_{{2}}(\beta,\varepsilon',\alpha) = \mathcal{O}\big((\beta/\gamma+1)\ln(4/\varepsilon')\big),
\end{equation}
queries to $O_2$ and $O_2^\dagger$, $|\mathcal{A}_{2}|=1$ ancilla, and $g_{P_2}=g_{O_2}+\mathcal{O}(1)$ gates per query}. Moreover, the gates of $P_2$ are obtained in classical runtime $\mathcal{O}({\rm poly}\big(q_2(\beta,\varepsilon',\alpha)\big)$. 
\end{thm} 

{As shown in Methods, the ``$\mathcal{O}( \cdot)$" in Eq. \eqref{eq:real_time_or} also hides only a  modest  global factor: 4.} In contrast to Eq.\ \eqref{eq:BEquery}, the relation between $\beta$ and $\ln(1/\varepsilon')$ in Eq.\ \eqref{eq:real_time_or} is multiplicative. However, in return, $P_2$ requires $|\mathcal{A}_2|=1$ ancillary qubit throughout, remarkably. This is a drastic reduction  relative to block-encoded oracle algorithms, and also to other algorithms based on real-time evolution. The latter is due to the use of a single real-time instead of an error-dependent number of them \cite{PW09, vanApeldoorn2020quantumsdpsolvers}. In fact, 
$|\mathcal{A}_2|=1$ is the minimum possible, because, since $F_\beta (H)$ is non-unitary, at least 1 ancilla is needed to block-encode it. {Moreover, the scaling of $g_{P_2}$ is optimal too. Since it is  based on real-time evolution oracles, it requires no qubitization \cite{Low2019hamiltonian}. Consequently, it adds only a small, constant number of gates per query to the intrinsic gate complexity $g_{O_2}$ of the oracle.}  These features make $P_2$ specially well-suited for near-term devices. Importantly, rather than a peculiarity of $P_2$, the favourable scalings of $|\mathcal{A}_{2}|$ and $g_{P_2}$ are generic features of the type of operator-function design behind it: An optimised Fourier-approximation algorithm for arbitrary analytical real functions of Hermitian operators \cite{TLS2022}.

Our algorithms support any $\lambda_{\rm min}\in[-1,1]$. For $P_2$, this is reflected by the sub-normalization factor $e^{-\beta(1+\lambda_{\rm min})}$, which decreases as $\lambda_{\rm min}$ departs from $-1$. In turn, the other factor, $e^{-\gamma}$, arises from the Gibbs phenomenon of Fourier series. The theorem holds for all $\gamma\geq 0$, allowing one to trade success probability for query complexity. {For $\varepsilon^{\prime}\ll1$, the optimal value of $\gamma$ depends only on $\beta$ for both coherent and probabilistic algorithms \cite[Sec. III]{TLS2022Sup}. }

Finally, Theos. \ref{teoQuITEbe} and \ref{teoQuITEberealtime} can be straightforwardly extended to the realistic case of approximate oracles: In \cite[Sec. I]{TLS2022Sup}, %App.\ \ref{sec:aprooximate_oracles},
we show (for generic analytical operator functions) that it suffices to take the oracle error (deviation from an ideal oracle) as $\varepsilon'_O=O(\varepsilon'/q)$ to keep the primitive's error in $O(\varepsilon')$. 

%%%%%%%%%%%%%%%%%%%%%%%%%%%%%%%%%%%%%%%%%%%%%%%%%%%%%
\subsection{Cooling-speed limits for oracle-based QITE algorithms}
\label{sec:cooling_rate_bound}
{T}he most challenging applications of QITE involve small post-selection probabilities, decreasing exponentially in $N$ in the worst cases. In an effort to reduce the overall complexity [see Eq.\ \eqref{eq:Qcomp}], this has fueled a long race \cite{PW09, BB10,chowdhury2017,QSDP17,vanApeldoorn2020quantumsdpsolvers,gilyen2019} to improve $q(\beta,\varepsilon')$, going from the seminal $\mathcal{O}\big(\beta\,{\rm poly}(1/\varepsilon')\big)$ of \cite{PW09} to the recent {$\mathcal{O}\big(\sqrt{2\max[e^2\,\beta, \ln(2/\varepsilon')]\,\ln(4/\varepsilon')}\big)$} of \cite{gilyen2019} or the additive scaling of Eq.\ \eqref{eq:BEquery}. However, to our knowledge, no runtime limit for QITE simulations {has been} established. This contrasts with real-time evolution (RTE), where fundamental runtime lower bounds are given by the ``no-fast-forwarding theorem"  \cite{Berry2007,Berry2014,Berry2015a}. Thes{e a}re saturated by optimal {RTE} algorithms \cite{low2017optimal,Low2019hamiltonian,gilyen2019}. 
Here we derive an analogous bound for imaginary time, {which we call cooling-speed limit in allusion to the use of QITE  to cool systems down to their ground state. }

More precisely,  we prove a universal efficiency limit for QITE primitives based on block-encoded oracles. This is convenient as it directly applies to our primitive with lowest query complexity, i.e. $P_1$. 

\begin{thm}[Imaginary-time no-fast-forwarding theorem]
Let $\beta>0$ and $0<\varepsilon'<\alpha/2$. Then, any $(\beta,\varepsilon',\alpha)$-QITE-primitive querying block-encoding Hamiltonian oracles has query complexity at least $q_{\min}(\beta,\varepsilon',\alpha)\geq \tilde q$, where $\tilde q\in\mathbb{R}_{>0}$ is the unique solution to the  equation
\begin{equation}\label{eq:tildeq}
\left|\frac{1-e^{-\frac{\beta}{4\tilde q}}}{2}\right|^{2\,\tilde q} = \frac{2\,\varepsilon'}{\alpha} \ . 
\end{equation}
%Then, , i.e.\ the nearest integer to $\tilde q$ (``round half up'').
\label{th:lowerbound}\end{thm}

Even though the bound is only given implicitly, interesting conclusions can readily be drawn. First, for any fixed $\beta$, the left-hand-side of Eq.\ \eqref{eq:tildeq} decreases monotonically with  $\tilde q$ (therefore the uniqueness of the solution). Second, for any fixed $\varepsilon'$ and $\alpha$, $\tilde q$ grows monotonically with $\beta$. Third, and most important, Eq.\ \eqref{eq:tildeq} is approximated by $\big(\frac{\beta}{{8}\,\tilde q}\big)^{2\tilde q}=2\,\varepsilon'/\alpha$ for $\beta\ll \tilde q$, as a Taylor expansion shows. The latter equation has a known explicit solution \cite{gilyen2019}, which, for $\alpha=1$, is given by Eq.\ \eqref{eq:BEquery}. %\mmt{Thais, esse fator 8 (em vez de 2) muda tudo de $\varepsilon'\to\varepsilon'^4$, mas como $\varepsilon'$ tá no Teo 1 só como $\log\varepsilon'$, então isso não altera o scaling.}  
Hence, for $\beta/\tilde q\to0$, Eq.\ \eqref{eq:BEquery} tends to the optimal scaling. 
Note that $\beta\ll \tilde q$ is equivalent to the first term Eq.\ \eqref{eq:BEquery} being much smaller than its second term, which in turn implies that $\varepsilon'$ should be  exponentially small in $\beta$. Thus, $P_1$ is close to optimal for small inverse temperatures or high precisions. Interestingly, this is the regime at which the first fragments of our master algorithm operate, as we see next.

%%%%%%%%%%%%%%%%%%%%%%%%%%%%%%%%%%%%%%%%%%%%%%%%%%%%%%%%%
%%%%%%%%%%%%%%%%%%%%%%%%%%%%%%%%%%%%%%%%%%%%%%%%%%%%%%%%%
\subsection{Fragmented master QITE algorithm} 
\label{sec:master_alg}

We call master QITE algorithm a procedure which incorporates the primitives to attain deterministic QITE. It means that these algorithms deterministically produce the state $\frac{F_{\beta}(H)\ket{\Psi}}{\left\Vert F_{\beta}(H)\ket{\Psi}\right\Vert}$, up to trace-distance error $\varepsilon$, if they are given an input state $\ket{\Psi}\in\mathbb H_\mathcal{S}$. 

Until now, two variants of master QITE algorithms had been reported, probabilistic and coherent (see Fig. \ref{fig:master}). The former leverage repeat-until-success: apply  $P_{\beta,\varepsilon^{\prime},\alpha}$ (on independent systems) until getting the desired output. Every time the postselection on the ancillas is not successful the resulting system state is discarded and system and ancillas are reinitialized for a new trial. The average number of trials until one gets one success is given as $\mathcal{O}(1/p_{\Psi}(\beta, \alpha))$.
In contrast, the latter are based on quantum amplitude amplification \cite{Brassard2002}. There, $P_{\beta,\varepsilon^{\prime},\alpha}$ is incorporated into a unitary amplification engine that is sequentially applied (on the same system) $\mathcal{O}\big(\sqrt{1/p_{\Psi}(\beta, \alpha)}\big)$ times. 
Hence, the overall query complexity of both variants is given by the unified expression 
\begin{equation}
Q_\kappa(\beta,\varepsilon, \alpha)=\mathcal{O}\bigg(\frac{1}{{\big(p_{\Psi}(\beta, \alpha)}\big)^{ {\mu_{\kappa}}}}\,q(\beta,\varepsilon^{\prime}, \alpha)\bigg),
\label{eq:Qcomp}
\end{equation}
where  {$\kappa=\,$prob/coh for probabilistic or coherent schemes, respectively, $\mu_{\rm prob}=1$, $\mu_{\rm coh}=1/2$},  
and $\varepsilon^{\prime}\coloneqq\varepsilon \sqrt{p_{\Psi}(\beta, \alpha)}/2$.
Since $p_{\Psi}(\beta, \alpha)$ can decrease with $N$ exponentially, the quadratic advantage in $1/p_{\Psi}(\beta, \alpha)$ of coherent approaches is highly significant. However,  coherent algorithms have a circuit depth $\mathcal{O}(\sqrt{1/p_{\Psi}(\beta, \alpha)})$ times greater than in probabilistic ones and require $\mathcal{O}(N)$ extra ancillas. This makes coherent schemes {impractical for} intermediate-scale quantum devices.

Our master algorithm relies on the basic identity $F_\beta(H)=\prod_{l=1}^{r} F_{\Delta\beta_l}(H)$ to partition the evolution into $r\in\mathbb{N}$ fragments of inverse temperatures $S_r\coloneqq\{\Delta\beta_l>0\}_{l\in [r]}$, such that $\sum_{l\in [r]}\Delta\beta_l=\beta$. We refer to $S_r$ as the \emph{fragmentation schedule}. For each $l$,  
the algorithm repeats until success a $(\Delta\beta_l,\varepsilon^{\prime}_l,\alpha_l)$-QITE-primitive $P_{\Delta\beta_l,\varepsilon^{\prime}_l,\alpha_l}$ on the output state $\ket{\Psi_{l-1}}$ of the $(l-1)$-th step, with $\varepsilon^{\prime}_l$ given in Eq \eqref{eq:F_errors}.  
That is, if the ancillas $\mathcal{A}$ are successfully post-selected in state $\ket{0}$, the system's output state $\ket{\Psi_{l}}$ is input into the $(l+1)$-th fragment. Else, the algorithm starts all over from the first fragment on $\ket{\Psi_{0}}\coloneqq\ket{\Psi}$, until $\ket{\Psi_{l-1}}$ is prepared and the $l$-th fragment can be run again. 
 {Alternatively, the measurement on $\mathcal{A}$ after each fragment can be seen as monitoring that the correct block of $U_{F_{\Delta\beta_l}(H)}$ is applied on each $\ket{\Psi_{l-1}}$, in contrast to the single error detection after $U_{F_{\beta}(H)}$ in the probabilistic master algorithm (see Fig. \ref{fig:general_mindset}).
  Note that the total number of trials (i.e. preparations of $\ket{\Psi}$) coincides with the number of repetitions of the first fragment. We also note that our method resembles the discrete formulation of the Zeno effect applied in the quantization of the Metropolis-Hastings walk for classical Hamiltonians \cite{Lemieux2020efficientquantum}. However, here we cannot apply the rewind technique, i.e iterate between two consecutive steps of inverse temperature instead of rebooting in case of a failure in the postselection \cite{PhysRevA.103.052408}. Rewind applied to fragmented QITE would not produce the right output state.       The following pseudocode summarizes all the algorithm}:

%%%%%%%%%%%% MASTER QUITE ALGORITHM %%%%%%%%%%%%%%%%%%%
\begin{center}
{\setlength{\fboxsep}{1pt}

%\framebox{
\begin{minipage}[t]{0.9\columnwidth}
\centering
\begin{algorithm}[H]\label{alg:Fragmented-QuITE-algorithm}
\SetAlgorithmName{Algorithm}{Algorithm}{Algorithm}
\SetAlgoLined
\SetKwInOut{Input}{input}\SetKwInOut{Output}{output}
\SetKwData{Even}{even}

\Input{$\ket{\Psi}$, $\beta\geq0$, $\varepsilon\geq0$, $S_r$, and QITE primitives $\{P_{\Delta\beta_l,\varepsilon^{\prime}_l,\alpha_l}\}_{l\in [r]}$ querying an oracle for $H$ (for $\Delta\beta_l\in S_r$, $\varepsilon^{\prime}_l$ given in Eq \eqref{eq:F_errors}, and $\alpha_l>0$).}
\Output{ $F_\beta(H)\ket{\Psi}/||F_\beta(H)\ket{\Psi}||$ up to trace-distance error $\mathcal{O}(\varepsilon)$.} 
\BlankLine
%build the list $0=\beta_{0}<\beta_{1}<\beta_{2}<\cdots<\beta_{r}=\beta$\;
%build the list $\varepsilon^\prime_{1},\, \cdots\,,\varepsilon^\prime_{r}$\;
Set $l=1$, initialize $\mathcal{S}$ in the input state $\ket{\Psi}$ and $\mathcal{A}$ in the reference state $\ket{0}$\;
\While{$l\leq r$}{apply %unitary $U_{F_{\Delta\beta_l}(H)}$ generated by 
the circuit $P_{\Delta\beta_l,\varepsilon^{\prime}_l,\alpha_l}$ on $\mathcal{S}$ and $\mathcal{A}$\;
measure $\mathcal{A}$ in computational basis\;
   \lIf{the outcome is $\ket{0}$}{$l\to l+1$}
   \lElse{break loop and go back to line 1}
}
  \caption{Fragmented QITE}
\end{algorithm}
\end{minipage}
}%}
\end{center}
%%%%%%%%%%%%%%%%%%%%%%%%%%%%%%%%%%%%%%%%%%%%%%%%%%

The correctness and complexity of Algorithm \ref{alg:Fragmented-QuITE-algorithm} are established by the following theorem, proven in \cite[Sec. IV]{TLS2022Sup}. %{App.\ \ref{app:theo_fragmented}}.
 \begin{thm}\label{teo_correct_master} (Fragmented master QITE algorithm). If 
 \begin{equation}
\label{eq:F_errors}
\varepsilon'_l\leq \left\{
     \begin{array}{lr}
       \frac{\varepsilon\,\prod_{k=1}^r\alpha_k}{2\times 4^{r-1}}\sqrt{p_{\Psi}(\beta)} & \text{ if } l=1,\\
       \frac{\varepsilon\,\prod_{k=l}^r\alpha_k}{4^{r-l+1}}\sqrt{\frac{p_{\Psi}(\beta)}{p_{\Psi}(\beta_{l-1})}} & \text{ if } l>1,\\
     \end{array}
   \right.
\end{equation}
for all $l\in[r]$, Algorithm \ref{alg:Fragmented-QuITE-algorithm} is a  master QITE algorithm for $H$ on $\ket{\Psi}$ with error $\mathcal{O}(\varepsilon)$ and average query complexity
\begin{equation}
\label{eq:F_overall_query}
{Q_{S_r}(\beta,\varepsilon) =\sum_{l=1}^{r} n_l\, q(\Delta\beta_l,\varepsilon^{\prime}_l, \alpha_l)},
\end{equation}
where $n_l\coloneqq\frac{p_{\Psi}(\beta_{l-1})}{p_{\Psi}(\beta)\prod_{k=l}^r\alpha^2_k}$ is the average number of times that $P_{\Delta\beta_l,\varepsilon^{\prime}_l,\alpha_l}$ is run, with $\beta_0\coloneqq0$, $\beta_l\coloneqq\sum_{k=1}^{l}\Delta\beta_k$ for all $l\in[r]$, and $p_{\Psi}(\tilde{\beta})\coloneqq\|F_{\tilde{\beta}}(H)\ket{\Psi}\|^2$ for any $\tilde{\beta}$.
\end{thm}

%%%%%%%%%%%%%%%%%%%%%%%%%%%%%%%%%%%%%%%%%%%%%%%%%%%%%%%%%%%%%%%%%%%%%%%
\begin{figure}[t!]
  \centering{}\includegraphics[width=0.6\columnwidth]{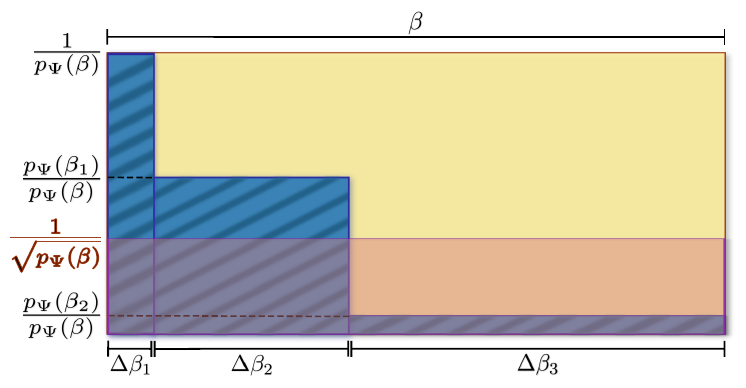}
  \caption{{{\bf Intuition behind the complexity reduction by fragmentation}. The overall complexity of the probabilistic master algorithm is dominated by the area of the yellow rectangle. In contrast, the corresponding complexity of the fragmented algorithm (here, for the exemplary case of $r=3$ fragments) is dominated by the area of the blue-shaded rectangles. Up to logarithmic corrections in the precision, the cumulative width of the blue-shaded rectangles coincides with the width of the yellow one, of order $\beta$. In contrast, while the height of the yellow rectangle is of order $1/p_{\Psi}(\beta)$, the height of the blue-shaded ones decreases from order $1/p_{\Psi}(\beta)$ till order $p_{\Psi}(\beta_{r-1})/p_{\Psi}(\beta)$, making the blue-shaded area smaller than the yellow one. For high-enough $\beta$, the reduction can be so strong that the complexity of the fragmented algorithm can reach even that of the coherent algorithm which is represented by the area of the pink rectangle with height $1/\sqrt{p_{\Psi}(\beta)}$. This intuition is rigorously proven for Primitive 1 (in Theorem \ref{thm:beta_c}) and numerically verified to exhaustion for both Primitives 1 and 2 (in Sec. \ref{sec:GS_sampling}).}
  \label{fig:Frag_complexity_schematics}}
\end{figure}
%%%%%%%%%%%%%%%%%%%%%%%%%%%%%%%%%%%%%%%%%%%%%%%%%%%%%%%%%%%%%%%%%%%%%%%%%%%%%%%%%%

{We note that, for Primitive 1, the average total number of trials coincides with that of the probabilistic algorithm: $n_1=1/p_{\Psi}(\beta)=:n_{\rm prob}$ (see Methods). This is important because the probabilistic algorithm consumes $q_1(\beta,\varepsilon^{\prime})$ queries per trial, successful or not. In contrast, the fragmented one consumes per trial $q_1(\Delta\beta_1,\varepsilon^{\prime}_1)$ queries, plus $q_1(\Delta\beta_2,\varepsilon^{\prime}_2)$ queries only if the first post-selection succeeds, plus $q_1(\Delta\beta_3,\varepsilon^{\prime}_3)$ queries only if the second one succeeds too, and so on. Hence, 
the total waste in queries is lower with fragmentation (see Fig. \ref{fig:Frag_complexity_schematics}).
 The strength of the reduction depends on how fast $p_{\Psi}(\beta_l)$ (and so $n_l$) decreases with $l$; but, in any case, it gets more drastic as $\beta$ increases. That is, the largest reductions are expected at the hardest regime of $p_{\Psi}(\beta)\ll1$. To maximize the effect, one wishes $q_1(\Delta\beta_l,\varepsilon^{\prime}_l)$ to decrease with $l$ as fast as possible. Note that Eq.\ \eqref{eq:F_errors} implies $\varepsilon^{\prime}_l< \varepsilon^{\prime}_{l+1}$, which plays against the latter wish. However, fortunately, $q_1(\Delta\beta_l,\varepsilon^{\prime}_l)$ grows approximately linearly in $\Delta\beta_l$ but sub-logarithmically in $1/\varepsilon^{\prime}_l$. Hence,  for sufficiently high $\beta$, one can make $q_1(\Delta\beta_l,\varepsilon^{\prime}_l)$ arbitrarily smaller than $q_1(\Delta\beta_{l+1},\varepsilon^{\prime}_{l+1})$ by choosing $\Delta\beta_l$ sufficiently smaller than $\Delta\beta_{l+1}$.}

{Based on these heuristics, we next prove for Primitive 1 that Alg. \ref{alg:Fragmented-QuITE-algorithm} can not only outperform the probabilistic  algorithm but also -- for sufficiently high $\beta$ -- even the coherent one, surprisingly. The proof is constructive: we devise suitable schedules that give the desired advantage for fragmentation. Remarkably, it is enough to consider only $r=2$ fragments. The result is valid for any $\ket{\Psi}$ and $H$, under only mild assumptions on the success probability $p_{\Psi}$ as a function of $\beta$. We denote the inverse function of $p_{\Psi}$ by $p^{-1}_{\Psi}$.
For simplicity, we state the theorem explicitly for the restricted case of $H$ non-degenerate, with a unique ground state $\ket{\lambda_{\rm min}}$ of overlap ${o}^2\coloneqq|\braket{\lambda_{\rm min}}{\Psi}|^2$ with $\ket{\Psi}$. However, it can be straightforwardly generalized to the degenerate case by redefining ${o}^2$ as the overlap with the lowest-energy subspace.

\begin{thm}[Fragmented QITE outperforms coherent QITE]
\label{thm:beta_c} Let $\ket{\lambda_{\rm min}}\in\mathbb{H}_\mathcal{S}$ be the unique ground state of $H$ and $\ket{\Psi}\in\mathbb{H}_\mathcal{S}$  such that $0<{o}\leq1/2.2$. Define the critical inverse temperature $\beta_c=\frac{2}{{o}}\left[\frac{2}{e}\ln\left(\frac{8}{{o}\,\varepsilon}\right)+p^{-1}_{\Psi}(\frac{{o}}{2.2})\right]$. Then, if ($H$ and $\ket{\Psi}$ are such that) $p_{\Psi}(\beta_c){\leq1/4}$, there exists a two-fragment schedule $S_2$ for which, for $P_1$, it holds that $Q_{S_2}(\beta,\varepsilon)<Q_{\text{coh}}(\beta,\varepsilon)$ for all $\beta\geq\beta_c$ and $0<\varepsilon<1$. In particular, $S_2=\big\{\Delta\beta_1={p^{-1}_\Psi\big(\frac{{o}}{2}\frac{1}{{\ln[e+2\ln(2/{o}\,\varepsilon)/e\beta]}}\big)},\Delta\beta_2=\beta-\Delta\beta_1\big\}$ is a valid choice of such schedules.
\end{thm}

The proof is given in the Supplementary Information \cite[Sec. VI]{TLS2022Sup}.
The schedules constructed there have the sole purpose of proving the existence of $\beta_c$ in general and are therefore not necessarily optimal for each specific $H$ and $\ket{\Psi}$. For instance,  in \cite[Sec. VIII]{TLS2022Sup}  
we study Gibbs-state sampling (i.e. for the maximally-mixed state as input, with ${o}=2^{-N/2}$) for $H$ describing non-interacting particles, where a closed-form expression for $p_{\Psi}(\beta)$ can be obtained. For this simple case, the theorem yields $\beta_c=\mathcal{O}\big(2^{N/2}\, N\big)$. However, in Sec.\ \ref{sec:GS_sampling} we numerically optimize the schedules and obtain $\beta_c=\mathcal{O}\big(2^{N/2}\big)$ for hard-to-simulate, interacting systems.
%For Primitive 2, the fact that $\alpha_l<1$ gives an extra exponential dependance of $n_l$ on $r$ as compared to $P_1$, which worsens the performance of fragmentation. 
The proof %in \cite[Sec. VII]{TLS2022Sup} %App.\ \ref{app:analytical_evidence}
exploits the additive dependence of $q_{{1}}$ on $\beta$ and the logarithmic term in Eq. \eqref{eq:BEquery}. Its extension to the multiplicative case of $q_{2}$ is left for future work. Nevertheless, here, we do consistently observe an advantage of fragmented QITE over coherent one for $P_2$. More precisely, in Sec.\ \ref{sec:GS_sampling}, we numerically find that also for $P_2$ does fragmentation outperform coherent-QITE at Gibbs-state sampling, with $\beta_c$ scaling with $N$ as in $P_1$ but with a somewhat larger pre-factor (which is expectable, as $\alpha_l<1$ gives an exponential dependance of $n_l$ on $r$ that worsens the performance). Either way, that fragmentation can outperform quantum amplitude amplification at all is remarkable, since the latter requires circuits $\mathcal{O}(\sqrt{1/p_{\Psi}(\beta)})$ times deeper and $\mathcal{O}(N)$ more ancillas than the former.}

Our findings would have little practical relevance if $\beta_{\rm c}$ was unphysically high. Fortunately, $\beta_{\rm c}=\mathcal{O}\big(2^{N/2}\big)$ is in an intermediate regime useful for important applications: E.g., Ground-state cooling (or, more generally, Gibbs-state sampling at low temperatures) requires $\beta$ scaling inversely proportionally to the spectral gap, which can be exponentially small in $N$ even for relatively simple Hamiltonians such as transverse-field Ising models \cite{Exp_gap_NP_comp,Gap_Anderson_loc}. 
{In fact, in Sec. \ref{sec:GS_sampling} we compare $\beta_{\rm c}$ with the inverse temperatures $\beta_{0.9}$ needed for a modest ground-state fidelity $0.9$. We systematically observe that $\beta_{\rm c}$ is either greater than or close to $\beta_{0.9}$, evidencing the relevance of the regime of advantage of fragmented over coherent QITE.
Finally, as m}entioned, $P_1$ is particularly well-suited for  fragmentation. On the one hand, it displays $\alpha_l=1$ for all $l\in[r]$. On the other hand, and most importantly, $q_1$ becomes optimal as $\beta$ decreases relative to $\ln(1/\varepsilon')$. {This is convenient to minimize Eq. \eqref{eq:F_overall_query}, because the first fragments (specially the first one) operate precisely at low $\Delta\beta_l$ and $\varepsilon^{\prime}_l$, close to that optimality regime. The latter is verified both analytically for the non-interacting case of \cite[Sec. VIII]{TLS2022Sup} %App. \ref{app:non_interacting}
and numerically for the examples of Sec.\ \ref{sec:GS_sampling} in \cite[Sec. IX]{TLS2022Sup}, %App.\ \ref{app:histograms},
where we consistently observe that $\beta_1$ is typically only a tinny fraction of $\ln(1/\varepsilon_1')$. Colloquially speaking, the widths of the first blue-shaded rectangles in Fig. \ref{fig:Frag_complexity_schematics}  can be reduced more with $P_1$ than with other primitives}.

%%%%%%%%%%%%%%%%%%%% FIG 2 %%%%%%%%%%%%%%%%%%%%%%%%%%%%%
%%%%%%%%%%%%%%%%%%%%%%%%%%%%%%%%%%%%%%%%%%%%%%%%%%%%%%%%
\begin{figure*}[t!]
\centering{}\includegraphics[width=1\textwidth]{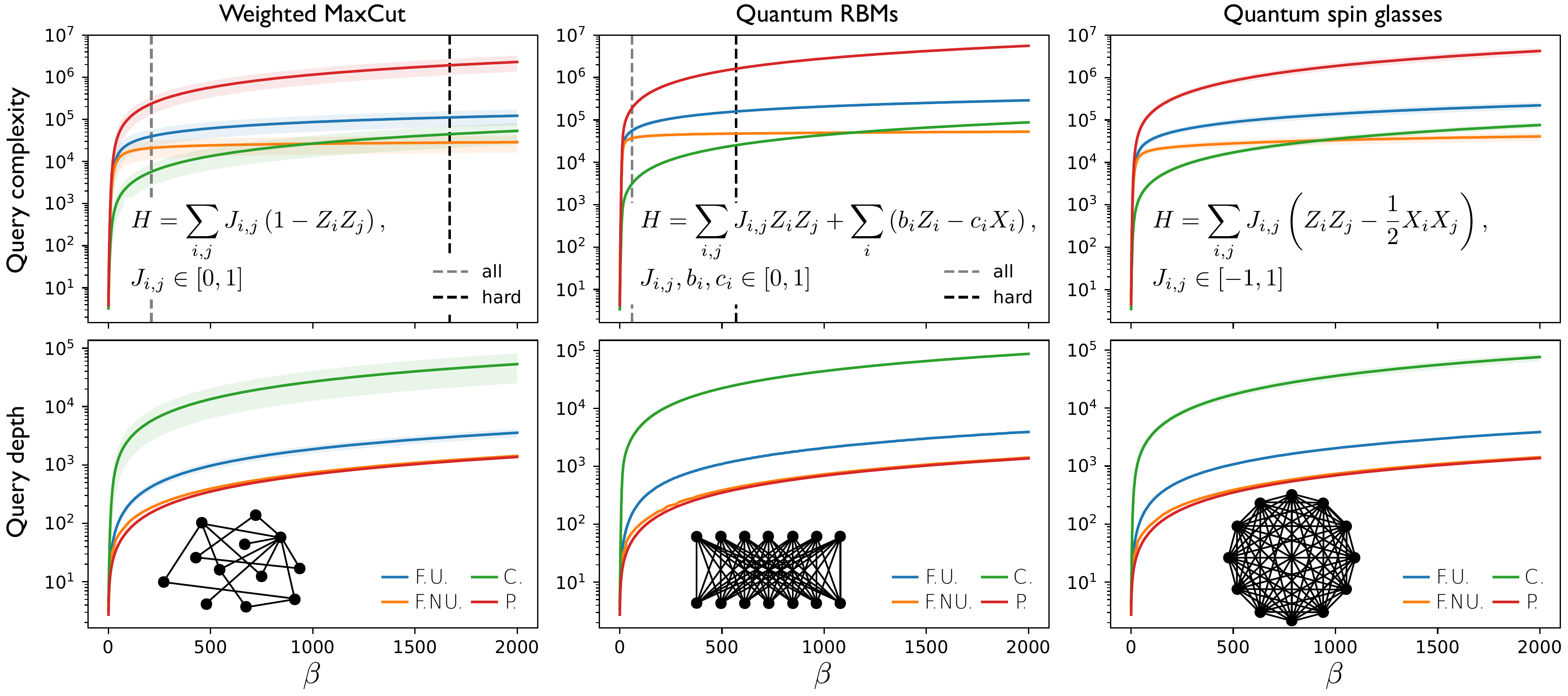}
\caption{{\bf Runtimes and circuit depths of quantum Gibbs-state samplers running on Primitive 1 versus inverse temperature}. Red corresponds to the probabilistic master QITE algorithm (P), green to the coherent one (C), blue to the fragmented one with uniform schedule $S_{r}$ for the best $r$ (F.U. for fragmented uniform), and orange to the fragmented one with a schedule $S_{r,a}$ as in Eq.\ \eqref{eq:schedule_ansatz} for the best $r$ and $a$ (F.NU. for fragmented non-uniform) (see also Fig.\ \ref{fig:Opt_schedules}).  {Three classes of Hamiltonians are shown (expressions in upper panels and lattice geometries in lower ones).} %correspond to the weighted MaxCut problem, the quantum RBM, and a quantum version of the Sherrington-Kirkpatrick spin glass model.
Solid curves represent the means over 1000 random instances from each class, wherea{s %(the thicknesses of) 
s}haded areas are the corresponding standard deviations. %Also, note the logarithmic vertical scale.
The examples shown correspond to $N=12$ qubits and a tolerated error of $\varepsilon=10^{-3}$, but qualitatively identical behaviors are observed for all $N$ between 2 and 15 as well as for $\varepsilon=10^{-2}$ and $\varepsilon=10^{-1}$. %Besides, qualitatively similar behaviors are observed for Primitive 2 (see App.\ \ref{app:QuITEprimitive2}). 
Upper panels: average overall query complexity. Both fragmented algorithms {comfortably} outperform the probabilistic one already at small $\beta$. In addition, fragmentation with non-uniform schedule outperforms even coherent QITE at a critical inverse temperature $\beta_c$. 
%Moreover, fragmentation with uniform schedule is in the same order of magnitude as coherent QITE at $\beta=\beta_c$ and eventually also outperforms it (not shown) at larger $\beta$ (between 5500 and 10000 depending on the Hamiltonian class). %The dependence of $\beta_c$ with $N$ is shown in Fig.\ \ref{fig:Beta_crit}. 
{The black and gray vertical dashed lines mark respectively the values $\beta^{(\text{aver})}_{0.9}$ and $\beta^{(\text{hard})}_{0.9}$ at which the average fidelity with the ground state (over all instances and over the $10\%$ of them with the smallest gaps) reaches a modest value of $0.9$ {(not shown in the third panel because they lie beyond the range of $\beta$ shown; see \cite[Sec. X]{TLS2022Sup}.}% App. \ref{app:QuITEprimitive1_hard}). }
%Those lines do not appear in the third panel because $\beta^{(\text{aver})}_{0.9}$ and  $\beta^{(\text{hard})}_{0.9}$ lie beyond the range of $\beta$ shown, and are hence much greater than $\beta_c$.  
Both in the first and second panels, $\beta^{(\text{aver})}_{0.9}$ is smaller than $\beta_\text{c}$, but the complexity of fragmented QITE at $\beta^{(\text{aver})}_{0.9}$ is already significantly smaller than that of probabilistic QITE. 
These considerations imply that fragmented QITE is either competitive or directly superior to coherent QITE for ranges of $\beta$ that are highly relevant for ground state preparation, e.g. The advantage of fragmentation becomes more evident when we compare the average query depths in the lower panels.} Defined as the maximum number of queries per circuit run (i.e., not taking into account independent trials), the query depth quantifies the circuit depth (relative to the depth per query) required by {one successful run. }
%As expected, coherent QITE lies orders of magnitude above probabilistic QITE; but, surprisingly, fragmented QITE with non-uniform and uniform schedules are respectively almost identical to and of the same order of the latter.
} 
\label{fig:Frag_QITE_P1}
\end{figure*}
%%%%%%%%%%%%%%%%%%%%%%%%%%%%%%%%%%%%%%%%%%%%%%%%%%%%%%%%
%%%%%%%%%%%%%%%%%%%%%%%%%%%%%%%%%%%%%%%%%%%%%%%%%%%%%%%%
%%%%%%%%%%%%%%%%%%%%%%%%%%%%%%%%%%%%%%%%%%%%%%%%%%%%%%%%

%%%%%%%%%%%%%%%%%%%%%%%%%%%%%%%%%%%%%%%%%%%%%%%%%%%%%%%%%
%%%%%%%%%%%%%%%%%%%%%%%%%%%%%%%%%%%%%%%%%%%%%%%%%%%%%%%%%
\subsection{Fragmented quantum Gibbs-state samplers} 
\label{sec:GS_sampling}
{We benchmark the performance of Alg. \ref{alg:Fragmented-QuITE-algorithm} at quantum Gibbs-state sampling by comparing Eqs.  \eqref{eq:F_overall_query} and \eqref{eq:Qcomp} for four classes of spin-$1/2$ systems: Ising models associated to the $i$) MaxCut and $ii$) weighted MaxCut problems \cite{AdibaticQC00,QAOA14, Qbacktracking15}; $iii$) transverse-field Ising interactions on the restricted-Boltzmann-machine (RBM) geometry \cite{Biamonte17, QBoltzmann18}; and $iv$) Heisenberg all-to-all interactions, corresponding to a quantum generalization of the Sherrington-Kirkpatrick model \cite{SK75,Pachenco12} for spin glasses. All four classes feature long-range frustation;
%and therefore cannot be reliably tackled with the techniques of \cite{motta2020,Gomes20,Sunetal21,Nishi21}. In addition, 
and classically simulating their Gibbs states (for random instances) is a computationally-hard task \cite{Karp72,TSSW00,LS10,Montanari19,FGGZ19}.

The Gibbs state $\varrho_{\beta}\coloneqq\frac{e^{-\beta (H-\lambda_{\rm min})}}{Z_{\beta}}$ of $H$ at $\beta$, with $Z_{\beta}\coloneqq\textrm{Tr}\left[e^{-\beta (H-\lambda_{\rm min})}\right]$ its partition function, can be prepared by QITE at $\beta/2$ on the maximally-mixed state $\varrho_0=\frac{\mathds{1}}{Z_{0}}$, where $Z_{0}=2^{N}$. 
Hence, the post-selection probability is $p_{\Psi}(\beta/2, \alpha)=\alpha^2\, \frac{Z_{\beta}}{Z_{0}}$, where $\alpha=1$ for $P_1$  and $e^{-\beta(1+\lambda_{\rm min})-\gamma}$ for $P_2$. This, together with Eqs. \eqref{eq:BEquery} and \eqref{eq:real_time_or}, determine the overall query complexities, with respect to $P_1$ and $P_2$, respectively, for the three master algorithms: probabilistic [Eq.\ \eqref{eq:Qcomp} for $\kappa=$ prob], coherent [Eq.\ \eqref{eq:Qcomp} for $\kappa=$ coh], and fragmented [Eq.\ \eqref{eq:F_overall_query}]. More technically, rather than Eqs. \eqref{eq:BEquery} or \eqref{eq:real_time_or} we use their ceiling functions, to guarantee that each fragment's query complexity is integer. }

For $N$ up to 15 qubits, we draw 1000 random {$H$'s} within each class. For fair comparison, we re-scale all {$H$'s} so that $\lambda_{\rm min}=-1$ and $\lambda_{\rm max}=1$.
For each of them, we calculate the complexities for $\beta$ between 0 and 10000 and $\varepsilon=0.1$, $0.01$, or $0.001$. Partition functions are evaluated by exact diagonalization of $H$. 
{Evaluating Eq.\ \eqref{eq:F_overall_query} requires in addition a choice of schedule.  We propose 
\begin{equation}
\label{eq:schedule_ansatz}
S_{r,a}\coloneqq\left\{\left[\left(\frac{l}{r}\right)^a-\left(\frac{l-1}{r}\right)^a\right]{\beta/2}\right\}_{l\in[r]},
\end{equation}
 for $a>1$, so that $\beta_l=\big(\frac{l}{r}\big)^a{\beta/2}$ for all $l\in[r]$. This guarantees that $\Delta\beta_1<\Delta\beta_{2} \hdots <\Delta\beta_{r}$ and allows us to control the strength of the inequalities by varying $a$.
%We optimise $r$ and $a$ by brute-force search: for each problem instance ($N$, $H$, and $\beta$), we sweep $r$ and $a$ until minimizing $Q_{S_{r,a}}(\beta,\varepsilon)$ \cite{thais_l_silva_2021_5595705}. 
For each problem instance ($N$, $H$, and $\beta$), we sweep $r$ and, for each value of $r$, we find the optimal $a$ through {the Broyden–Fletcher–Goldfarb–Shanno (BFGS)} algorithm until minimizing $Q_{S_{r,a}}({\beta/2},\varepsilon)$ \cite{thais_l_silva_2021_5595705}.}

%%%%%%%%%%%%%%%%%%%% FIG 3 %%%%%%%%%%%%%%%%%%%%%%%%%%%%%
%%%%%%%%%%%%%%%%%%%%%%%%%%%%%%%%%%%%%%%%%%%%%%%%%%%%%%%%
\begin{figure}[t!]
\centering{}\includegraphics[width=0.6\columnwidth]{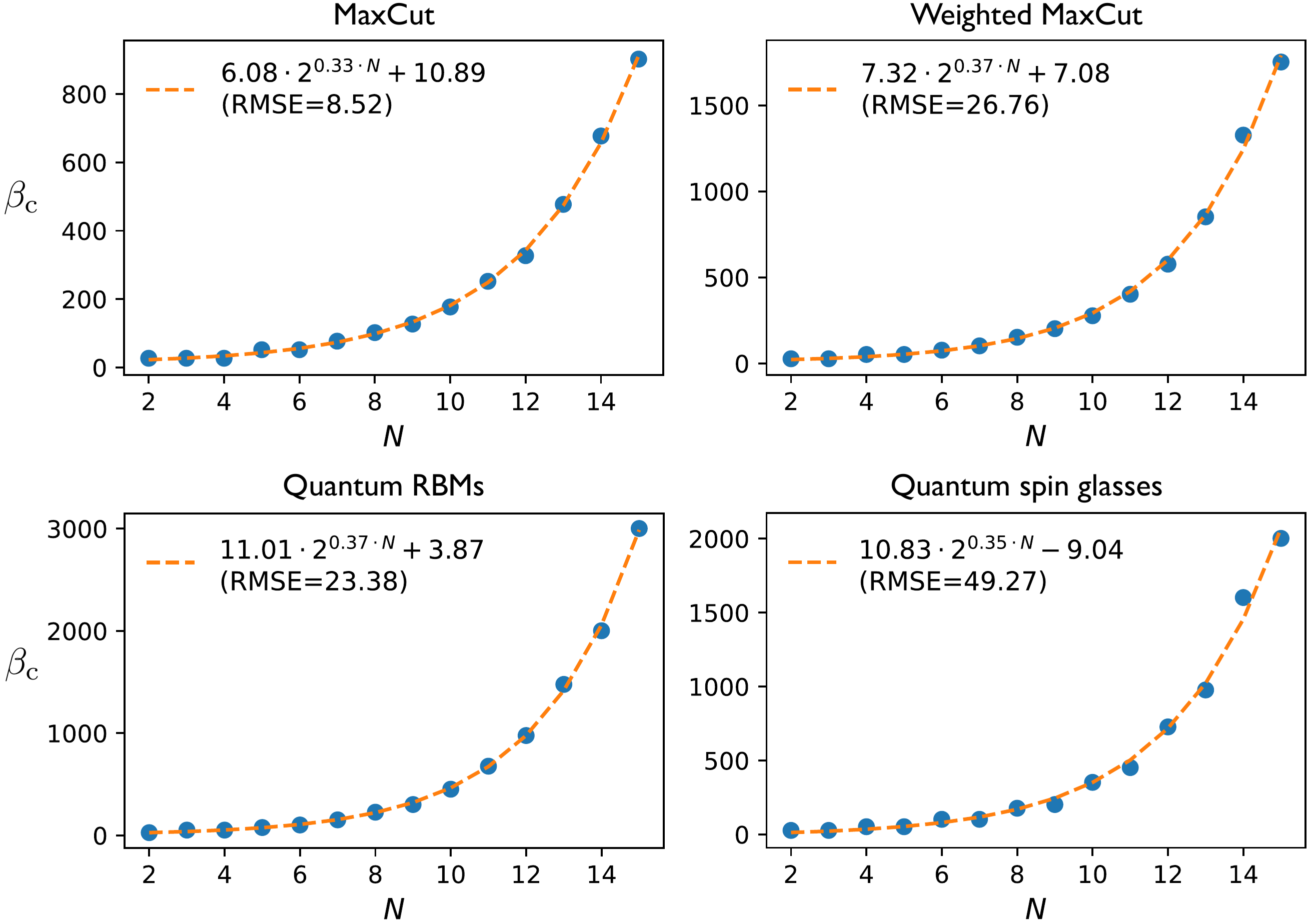}
\caption{{\bf Critical inverse temperatures for $P_1$ versus number of qubits}. The error and Hamiltonian classes are the same as in Fig.\ \ref{fig:Frag_QITE_P1}, except for  MaxCut, defined as weighted MaxCut but with random $J_{i,j}\in\{0,1\}$ for all $(i,j)$. Blue dots represent the means over 1000 instances from each class, whereas dashed orange curves their fits over the Ansatz $\beta_{\rm c}(N)=A\, 2^{\eta\, N}+ B$, with $A, B, \eta\in\mathbb{R}$. The fit results, together with their root-mean-square deviations (RMSDs), are shown in the insets. Similar scalings with $N$ are observed for $P_2$ \cite[Sec. XI]{TLS2022Sup}. %(see App.\ \ref{app:QuITEprimitive2}).
In all cases, $\beta_{\rm c}=\mathcal{O}\big(2^{N/2}\big)$ is satisfied.
\label{fig:Beta_crit}}
\end{figure}
%%%%%%%%%%%%%%%%%%%%%%%%%%%%%%%%%%%%%%%%%%%%%%%%%%%%%%%%
%%%%%%%%%%%%%%%%%%%%%%%%%%%%%%%%%%%%%%%%%%%%%%%%%%%%%%%%

%%%%%%%%%%%%%%%%%%%% FIG 3 %%%%%%%%%%%%%%%%%%%%%%%%%%%%%
%%%%%%%%%%%%%%%%%%%%%%%%%%%%%%%%%%%%%%%%%%%%%%%%%%%%%%%%
\begin{figure}[t!]
\centering{}\includegraphics[width=0.6\columnwidth]{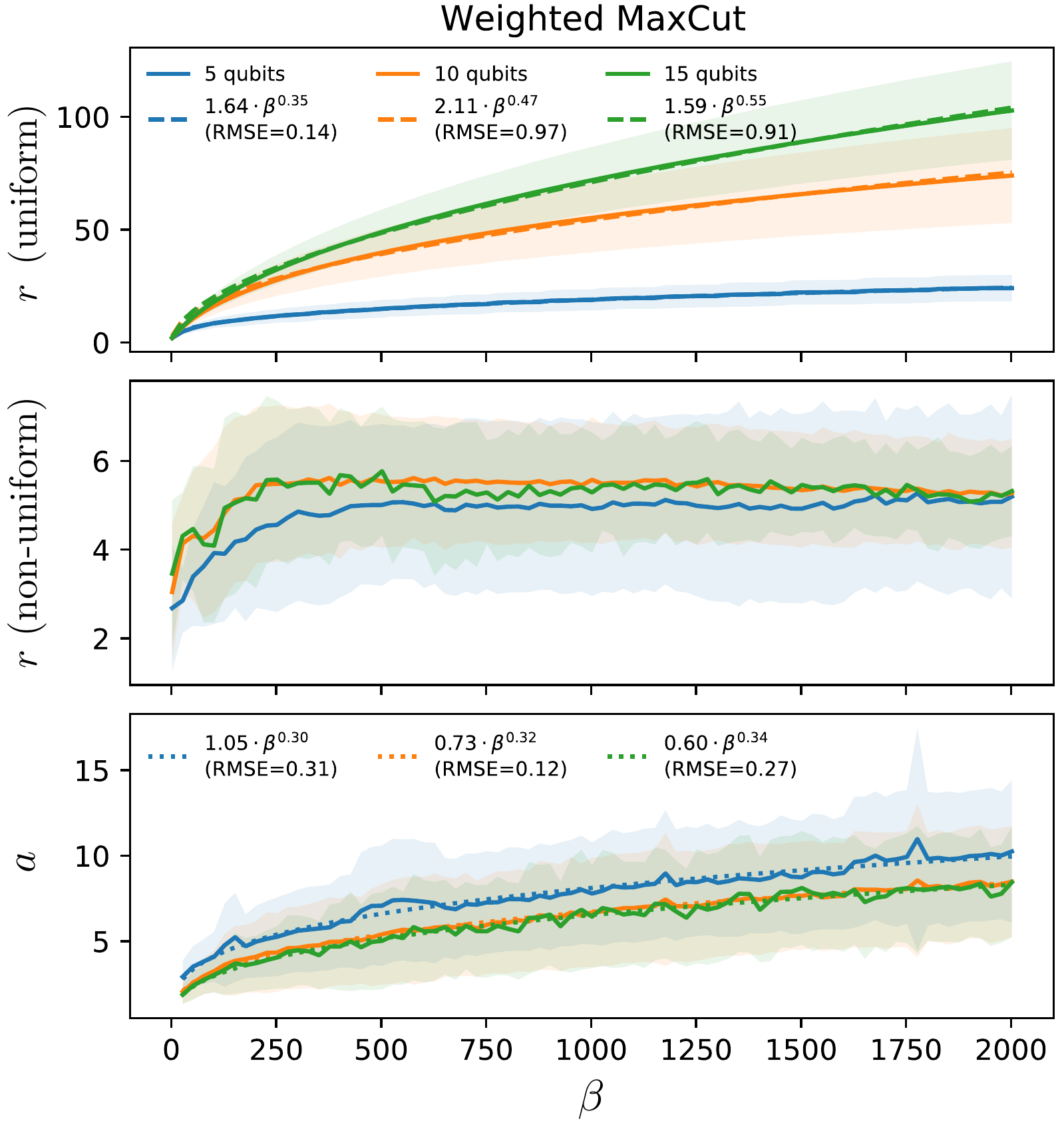}
\caption{{\bf Optimal fragmentation schedules for Primitive 1 versus inverse temperature}. System sizes are $N=5$ (blue), $N=10$ (orange), and $N=15$ (green). Solid curves represent the means over 1000 random weighted-MaxCut Hamiltonians, whereas (the thicknesses of) shaded curves are the standard deviations. The tolerated error is $\varepsilon=10^{-3}$. Qualitatively identical behaviors are observed for all $N$ between 2 and 15 as well as for $\varepsilon=10^{-2}$ and $\varepsilon=10^{-1}$; and the same holds for the other Hamiltonian classes in Fig.\ \ref{fig:Beta_crit}. The upper panel shows the optimal number of fragments $r$ for uniform schedules $S_{r,1}$. The central and lower panels respectively show the optimal $r$ and $a$ for non-uniform schedules $S_{r,a}$. The dashed and dotted curves in the upper and lower panels respectively represent fits over the ans\"{a}tze $r(\beta)=A\, \beta^{\eta}$ and $a(\beta)=A\, \beta^{\eta}$, with $A$ and $\eta\in\mathbb{R}$. The fit results are shown in the  insets. {Remarkably, for non-uniform schedules, the observed scaling for $r$ is constant not only with $\beta$ but also with $N$.}} 
\label{fig:Opt_schedules}
\end{figure}
%%%%%%%%%%%%%%%%%%%%%%%%%%%%%%%%%%%%%%%%%%%%%%%%%%%%%%%%
%%%%%%%%%%%%%%%%%%%%%%%%%%%%%%%%%%%%%%%%%%%%%%%%%%%%%%%%

{The overall complexities and circuit depths obtained (together with those for uniform schedules, i.e. with fixed $a=1$) are shown in Fig.\ \ref{fig:Frag_QITE_P1} for $P_1$; and the scalings with $N$ of $\beta_{\rm c}$ in Fig.\ \ref{fig:Beta_crit}. Similar scalings for the critical inverse temperature are obtained for $P_2$ but with somewhat higher constant pre-factors \cite[Sec. XI]{TLS2022Sup}, %(see App.\ \ref{app:QuITEprimitive2}), 
which is expectable due to the non-unit sub-normalization factors $\alpha_k$ in $n_l$. Summarizing, our numerical experiments support the following observation.

\begin{obs} (Gibbs-state sampling with fragmented QITE). Let the primitives be of fixed type, either $P_1$ or $P_2$. Then, for every $H$ and $\varepsilon>0$ studied, there exists $\beta_{\rm c}=\mathcal{O}\big(2^{N/2}\big)$ such that, for all $\beta\geq\beta_{\rm c}$, there is a schedule $S_r=S_r(\beta)$ that makes $Q_{S_r}({\beta/2},\varepsilon)\leq Q_{\rm coh}({\beta/2},\varepsilon,\alpha)$. Moreover, the maximal circuit depth required by fragmentation is asymptotically the same as that of probabilistic QITE. 
\label{obs:fragmented}
\end{obs}
%%%%%%%%%%%%%%%%%%%%%%%%%%%%%%%%%%%%%%%%%%%%%%%%%%%%%%%%
Apart from the notable fact that fragmentation outperforms coherent QITE for both primitives, it is also remarkable that, long before $Q_{S_r}$ reaches $Q_{\rm coh}$, at $\beta_{\rm c}$, $Q_{S_r}$ is already much smaller than $Q_{\rm prob}$. Crucially, these advantages of fragmented QITE come at no cost in circuit depth, since the query depth of fragmentation, $\sum_{l=1}^{r}\, q(\Delta\beta_l,\varepsilon^{\prime}_l, \alpha_l)$, is observed to almost coincide with that of repeat until success, $q({\beta/2},\varepsilon^{\prime}, \alpha)$, specially for high $\beta$. Note that the latter needs not be the case: strictly speaking, neither $q_1$ nor $q_2$ are additive in $\Delta\beta_l$ due to the non-linear dependance of $\varepsilon^{\prime}_l$ on $\Delta\beta_l$.}

Of course, the optimal schedules as functions of $\beta$ are a priori unknown. Nevertheless, the trends we observe for the schedule proposals in Eq. \eqref{eq:schedule_ansatz} are so compelling that they provide a sound basis for educated guesses in general: 
%%%%%%%%%%%%%%%%%%%%%%%%%%%%%%%%%%%%%%%%%%%%%%%%%%%%%%%%
\begin{obs} (Optimal schedules). For uniform and non-uniform schedules $S_{r,1}$ and $S_{r,a}$, given by Eq.\ \eqref{eq:schedule_ansatz}, the overall  complexity for $P_1$ is respectively minimised by $r=\mathcal{O}(\beta^{1/2})$ and $(r,a)=\big(\mathcal{O}(1),\mathcal{O}(\beta^{1/3})\big)$ (see Fig.\ \ref{fig:Opt_schedules}); whereas for $P_2$ by $r=2$ and $(r,a)=\big(\mathcal{O}(1),\mathcal{O}(\beta^{1/6})\big)$ \cite[Sec. XI]{TLS2022Sup}. % (see App.\ \ref{app:QuITEprimitive2}).
 \label{obs:opt_schedule}
\end{obs}
%%%%%%%%%%%%%%%%%%%%%%%%%%%%%%%%%%%%%%%%%%%%%%%%%%%%%%%%
As expected from the exponential dependence on $r$ in Eq.\ \eqref{eq:F_errors}, a slow growth of $r$ with $\beta$ is observed for each $N$ to minimize $Q_{S_r}({\beta/2},\varepsilon)$. This is indeed seen for $P_1$ with uniform schedules (Fig.\ \ref{fig:Opt_schedules}, upper panel). On the other hand, for $P_2$ with uniform schedules, $r=2$ is observed \cite[Sec. XI]{TLS2022Sup} % (App.\ \ref{app:QuITEprimitive2})
to minimize $Q_{S_r}$ but the resulting complexity does not reach $Q_{\rm coh}$ over the scanned domain ($0\leq \beta\leq 10000$).
However, for both $P_1$ (Fig.\ \ref{fig:Opt_schedules}, central panel) and $P_2$ with non-uniform schedules (where fragmentation does outperform amplitude amplification), the observed scaling of $r$ is constant with both $\beta$ and $N$, remarkably. 
In turn, that $a$ grows with $\beta$ implies that each $\Delta\beta_l$ decreases relative to $\Delta\beta_{l+1}$ as $\beta$ grows. This is consistent with the intuition from Sec.\ \ref{sec:master_alg} that each $\Delta\beta_l$ should be smaller than $\Delta\beta_{l+1}$.  {In addition, we consistently observe that, for the obtained optimal schedules, $\Delta\beta_1$ is only a tinny fraction (around 0.1\% to 2\%) of $8\ln(4/\varepsilon_1')$ \cite[Sec. IX]{TLS2022Sup}.} In fact, for both primitives, inserting the obtained $a(\beta)$ into Eq.\ \eqref{eq:schedule_ansatz}, one sees that all $\Delta\beta_l$'s (except the last one, $\Delta\beta_r$) also decrease in absolute terms as $\beta$ grows. Yet, that $a$ grows slowly with $\beta$ guarantees that the $\Delta\beta_l$'s do not decrease too much. More precisely,  comparing with Eq.\ \eqref{eq:F_errors}, we see that $\Delta\beta_l>\varepsilon'_l$ for all $l\in[r]$. This is an important sanity check, because if $\Delta\beta_l<\varepsilon'_l$, the identity operator would readily provide an $(\varepsilon'_l,1)$-block-encoding of $F_{\Delta\beta_l}(H)$, hence rendering the obtained scaling for $a(\beta)$ meaningless.
%%%%%%%%%%%%%%%%%%%%%%%%%%%%%%%%%%%%%%%%%%%%%%%%%%%%%%%%
%%%%%%%%%%%%%%%%%%%%%%%%%%%%%%%%%%%%%%%%%%%%%%%%%%%%%%%%
\section{Discussion}
\label{sec:conclusions}
We have presented two QITE primitives {and a master QITE algorithm.} The first primitive {is designed for block-encoding Hamiltonian oracles and} has query complexity (number of oracle calls) sub-additive in the inverse-temperature $\beta$ and $\ln(\varepsilon^{-1})$, with $\varepsilon$ the error. This scaling is better than all previously-known bounds \cite{chowdhury2017,gilyen2019} for $\beta\lesssim 8\ln(4\,\varepsilon^{-1})$  and becomes provably optimal for $\beta\ll\ln(\varepsilon^{-1})$. Optimality is proven by showing {saturation of} a universal cooling-speed limit that is an imaginary-time counterpart of the celebrated no fast-forwarding theorem for real-time simulations \cite{Berry2007,Berry2014,Berry2015a}. {It is an open question what the optimal scaling is away from the saturation regime. Coincidentally, the first steps of our master algorithm operate precisely in that regime}. On the other hand, the second primitive is designed for a simplified model of real-time evolution oracles involving a single time. It{s complexity is} multiplicative in $\beta$ and $\ln(\varepsilon^{-1})$, but it requires a single ancillary qubit throughout and its oracle is {experimentally-friendlier than in previous QITE primitives. Interestingly, preliminary numerical analysis \cite{Lucas2022} suggests that the asymptotic scaling with $\beta$ of both primitives' complexities could actually be significantly better than in the analytical bounds above, for $P_1$ even reaching levels as good as $q_1(\beta,\varepsilon')=\mathcal{O}\big(\sqrt{\beta\,\ln(1/\varepsilon')}\big)$.}

Our primitives are based on two technical contributions to quantum signal processing (QSP) \cite{low2017optimal, Low2019hamiltonian, gilyen2019,vanApeldoorn2020quantumsdpsolvers} relevant on their own. The first one is a bound on the approximation error of Hermitian-operator functions by their truncated Chebyshev series, for any analytical real function. The second one is a novel, Fourier-based QSP variant for real-time evolution oracles superior to previous ones \cite{vanApeldoorn2020quantumsdpsolvers} in that it requires a single real time (and therefore a single ancilla), instead of multiple ones. Moreover, it is also experimentally friendly in that it requires no qubitization \cite{Low2019hamiltonian}.

Primitive technicalities aside, the main conceptual contribution of this work is the master QITE algorithm, which is conceptually simple, yet surprisingly powerful. It is based on breaking the evolution into small-$\beta$ fragments. This gives a large reduction in wasted queries and circuit depth, yielding an overall runtime competitive with (and for high $\beta$ even better than) that of coherent approaches based on quantum amplitude amplification (QAA). This is remarkable since the latter requires in general $N$ extra ancillary qubits and circuits $\mathcal{O}\big(1/\sqrt{p_{\Psi}(\beta, \alpha)}\big)$ times deeper than the former. To put this in perspective, it is illustrative to compare with quantum amplitude estimation (QAE). In its standard form, QAE has similar hardware requirements as QAA \cite{Brassard2002}. However, recently, interesting algorithms have appeared \cite{Low_depth_QAE20,Exp_low_depthQAE21} that perform partial QAE with circuit  depths that can interpolate between the probabilistic and coherent cases. In contrast, here, we beat full QAA using circuit depths for most runs much lower than in the bare probabilistic approach. 

{That fragmented QITE outperforms coherent QITE is proven rigorously for Primitive 1 and also supported by exhaustive numerical evidence for both primitives. Namely, our numerical experiments address random instances of Ising, transverse-field Ising, and Heisenberg-like Hamiltonians encoding computationally hard problems relevant for combinatorial optimisations, generative machine learning, and statistical physics, e.g. We emphasize that our analysis of is based on the analytical upper bounds on the query complexity we obtained, instead of the complexities themselves. The corresponding analysis for the actual (numerically obtained) query complexities requires re-optimizing the fragmentation schedules. Preliminary observations \cite{Lucas2022} in that direction are again promising, indicating that the actual overall complexities may be orders of magnitude lower than in Fig. \ref{fig:Frag_QITE_P1}, e.g.
In any case, qualitatively similar interplays between fragmentation and QAA are expected even for other types of primitives (beyond QITE) whose complexity and post-selection probability have similar scalings. 
%In fact, even though Theorem \ref{teo_correct_master} is stated explicitly in terms of QITE, it actually holds in general for any operator function $F$ that factorizes as $F(H)=\prod_{l=1}^{r} F_{l}(H)$ in terms of suitable factors $F_{l}$ (see App.\ \ref{app:theo_fragmented}). 
All these exciting prospects are being explored for future work.}

Our findings open a new research direction towards {mid}-term high-precision quantum algorithms. 
%A particularly interesting possibility is the extension of fragmentation to operator-function synthesis with functions other than the exponential. 
In particular, the presented primitives, cooling-speed limit, QSP methods, and master algorithm constitute a {powerful} toolbox for quantum signal processors {specially relevant for the transition from NISQ to early prototypes of fault-tolerant} hardware.

%%%%%%%%%%%%%%%%%%%%%%%%%%%%%%%%%%%%%%%%%%%%%%%%%%%%%%%%
%%%%%%%%%%%%%%%%%%%%%%%%%%%%%%%%%%%%%%%%%%%%%%%%%%%%%%%%

\section{Methods}

\subsection{Preliminaries}
\label{sec:prelim}
We consider an $N$-qubit system $\mathcal{S}$, of Hilbert space $\mathbb H_\mathcal{S}$. QITE with respect to a Hamiltonian $H$ on $\mathbb H_\mathcal{S}$ and over an imaginary time $-i\,\beta$ is represented by the non-unitary operator $e^{-\beta H}$. 
This can be simulated via post-selection with a unitary operator $U$ that encodes $e^{-\beta H}$ in one of its matrix blocks {\cite{PW09, chowdhury2017, vanApeldoorn2020quantumsdpsolvers, gilyen2019, CSS21}}. 
We denote by $\mathbb H_\mathcal{A}$ the Hilbert space of an ancillary register $\mathcal{A}$, by $\mathbb H_\mathcal{SA}\coloneqq \mathbb H_\mathcal{S}\otimes\mathbb H_\mathcal{A}$ the joint Hilbert space of $\mathcal{S}$ and $\mathcal{A}$, and by $\left\Vert A \right\Vert$ the spectral norm of an operator $A$.
The following formalizes the encoding.

%%%%%%%%%%%%%%%%%%%%%%%%%%%%%%%%%%%%%%%%%%%%%%%%%%%%%%%%
\begin{dfn} (Block encodings). For sub-normalization $0\leq\alpha\leq 1$ and tolerated error $\varepsilon> 0$, a unitary operator $U_A$ on $\mathbb H_\mathcal{SA}$
is an $(\varepsilon,\alpha)$-block-encoding of a linear operator $A$ on $\mathbb H_\mathcal{S}$ if $\left\Vert \alpha\,A-\bra{0}\,U_A\,\ket{0}\right\Vert \leq\,\varepsilon$,
for some $\ket{0}\in\mathbb{H}_\mathcal{A}$. 
For $\varepsilon=0$ and $(\varepsilon,\alpha)=(0,1)$ we use the short-hand terms perfect $\alpha$-block-encoding  and perfect block-encoding, respectively. 
\label{def:block_enc}
\end{dfn}
\noindent E.g., if $U_A$ is a perfect $\alpha$-block-encoding of $A$, measuring $\ket{0}\in\mathbb{H}_\mathcal{A}$ on $U_A\ket{\Psi}\ket{0}\in\mathbb H_\mathcal{SA}$, for any $\ket{\Psi}\in\mathbb H_\mathcal{S}$, leaves $\mathcal{S}$ in the state $\frac{A\ket{\Psi}}{\left\Vert A\ket{\Psi}\right\Vert}$. The probability of that outcome is $\alpha^2\left\Vert A\ket{\Psi}\right\Vert ^{2}$. Note that, since $\left\Vert U_A \right\Vert=1$,  a perfect $\alpha$-block-encoding is possible only if $\alpha\left\Vert A\right\Vert \leq1$. Hence, $\alpha$ allows one to encode matrices even if their norm is greater than $1$. Typically, however, one wishes $\alpha$ as high as possible, to avoid unnecessary reductions in post-selection probability.

Our algorithms admit two types of oracle as input. The first one is based on perfect block-encodings of $H$ and therefore requires $\left\Vert H\right\Vert \leq1$. If $\left\Vert H\right\Vert >1$, however, the required normalisation can be enforced by a simple spectrum rescaling. More precisely, for $\lambda_{-}$ and $\lambda_{+}$ arbitrary lower and upper bounds, respectively, to the minimal and maximal eigenvalues of $H$, $\lambda_{\rm min}$ and $\lambda_{\rm max}$, the rescaled Hamiltonian $H^{\prime}\coloneqq\frac{H-\bar{\lambda}\mathds{1}}{\Delta\lambda}$ fulfils $\left\Vert H^{\prime}\right\Vert\leq1$ by construction, with the short-hand notation $\bar{\lambda}\coloneqq\frac{\lambda_{+}+\lambda_{-}}{2}$ and $\Delta\lambda\coloneqq\frac{\lambda_{+}-\lambda_{-}}{2}$. Then, by correspondingly rescaling the inverse temperature as $\beta^{\prime}\coloneqq\Delta\lambda\,\beta$, one obtains the propagator $e^{-\beta^{\prime} H^{\prime}}$, which induces the same physical transformation as $e^{-\beta H}$. Hence, from now on, without loss of generality we assume throughout that $\left\Vert H\right\Vert\leq1$, i.e. that $-1\leq\lambda_{\rm min}\leq\lambda_{\rm max}\leq1$.

We are now in a good position to define our first oracle, $O_1$, {which is the basis of our first primitive, $P_1$. We denote by $\mathcal{A}_{1}$ the entire ancillary register needed for $P_1$ and by $\mathcal{A}_{O_1}\subset\mathcal{A}_{1}$ the specific ancillary qubits required to implement $O_1$}. 
%%%%%%%%%%%%%%%%%%%%%%%%%%%%%%%%%%%%%%%%%%%%%%%%%%%%%%%%
\begin{dfn}
\label{def:block_enc_or} (Block-encoding Hamiltonian oracles). We refer as a block-encoding oracle for a Hamiltonian $H$ on $\mathbb H_\mathcal{S}$ to a controlled unitary operator $O_1$ on $\mathbb H_{\mathcal{SA}_{O_1}}$ of the form $ {O_1=U_H\otimes\ket{0}\bra{0}+\mathds{1}\otimes\ket{1}\bra{1}}$, where $\mathds{1}$ is the identity operator on $\mathbb H_{\mathcal{S}}$, $\{\ket{0},\ket{1}\}$ a computational basis for the control qubit, and $U_H$ a perfect block encoding of $H$.
\end{dfn}

\noindent This is a powerful oracle paradigm used both in QITE {\cite{chowdhury2017, vanApeldoorn2020quantumsdpsolvers, gilyen2019,CSS21}} and real-time evolution \cite{BCCKS15,low2017optimal, Low2019hamiltonian, gilyen2019}. It encompasses, e.g., Hamiltonians given by linear combinations of unitaries, $d$-sparse Hamiltonians (i.e. with at most $d$ non-null matrix entries per row), and Hamiltonians given by states \cite{Low2019hamiltonian}. Its complexity depends on $H$, but highly efficient implementations are known. E.g., for $H$ a linear combination of $m$ unitaries, each one requiring at most $c$ two-qubit gates, {$O_1$ can be implemented with $|\mathcal{A}_{O_1}|=\mathcal{O}(\log_{2} m)$ ancillary qubits and gate complexity (i.e. total number of two-qubit gates) $g_{O_1}=\mathcal{O}\big(m(c+\log_{2} m)\big)$} \cite{BCCKS15,Low2019hamiltonian}.  

The second oracle model that we consider encodes $H$ through the real-time 
unitary evolution it generates. 
%%%%%%%%%%%%%%%%%%%%%%%%%%%%%%%%%%%%%%%%%%%%%%%%%%%%%%%%
\begin{dfn}
\label{def:real_t_or} (Real-time evolution Hamiltonian oracle). We refer as a real-time evolution oracle for a Hamiltonian $H$ on $\mathbb H_\mathcal{S}$ at a time $t\in\mathbb{R}$ to a controlled-$e^{-itH}$ gate  {$O_2=\mathds{1}\otimes\ket{0}\bra{0}+e^{-itH}\otimes\ket{1}\bra{1}$}.
\end{dfn}
\noindent This is a simplified version of the models of \cite{PW09, vanApeldoorn2020quantumsdpsolvers}, e.g. There, controlled real-time evolutions at multiple times are required, thus involving multiple ancillas. In contrast, $O_2$ involves a single real time, so the ancillary register $\mathcal{A}_{O_2}$ consists of $|\mathcal{A}_{O_2}|=1$ single qubit (the control). In fact, we show below that no other ancilla is needed for our second primitive, $P_2$, i.e. $\mathcal{A}_{2}=\mathcal{A}_{O_2}$. This is advantageous for near-term implementations. {There, one may for instance apply product formulae \cite{campbell2019, COS19} to implement $O_2$ with gate complexities $g_{O_2}$ that, for intermediate-scale systems, can be considerably smaller than for $O_1$. Furthermore, this oracle is also relevant to hybrid analogue-digital platforms, for which QSP schemes have already been studied \cite{HamiltonianQSP21}}.

QITE algorithms based on post-selection rely on a unitary quantum circuit %that queries (i.e. applies) a Hamiltonian oracle $O$ 
to simulate a block encoding of the QITE propagator. We refer to such circuits as QITE primitives. 
%%%%%%%%%%%%%%%%%%%%%%%%%%%%%%%%%%%%%%%%%%%%%%%%%%%%%%%%
\begin{dfn}\label{def:primitives} (QITE primitives). Let $\beta\geq0$, $\varepsilon^{\prime}\geq0$, and $\alpha\leq 1$. A $(\beta,\varepsilon^{\prime},\alpha)$-QITE-primitive of query complexity $q(\beta,\varepsilon^{\prime},\alpha)$ is a circuit $P$, with $q(\beta,\varepsilon^{\prime},\alpha)$ calls to an oracle $O$ for $H$ or its inverse $O^{\dagger}$, that generates an $(\varepsilon^{\prime},\alpha)$-block-encoding $U_{F_{\beta}(H)}$ of $F_{\beta}(H)\coloneqq e^{-\beta (H-\lambda_{\rm min})}$, for all $H$. 
\end{dfn}
\noindent 
Note that $P$ is Hamiltonian agnostic, i.e. it admits any $H$ provided it is properly encoded in the corresponding oracle.
The factor $e^{-\beta\lambda_{\rm min}}$ implies that $\|F_{\beta}(H)\|=1$, thus maximizing the post-selection probability. However, if $\lambda_{\rm min}$ is unknown, one can replace it by a suitable lower bound $\lambda_{-}\geq -1$% (recall that $\lambda_{\rm min}\geq -1$)
. This  introduces only a constant sub-normalisation.
In turn, the query complexity is the gold-standard figure of merit for efficiency of oracle-based algorithms. It quantifies the runtime of $P$ relative to that of an oracle query. In fact, $P$ is time-efficient if its query complexity and {gate complexity per query $g_P$ are both} in $\mathcal{O}\big({\rm poly}(N, \beta,1/\varepsilon^{\prime},\alpha)\big)$.

Importantly, normalisation causes the post-selection probability $p_{\Psi}(\beta,\varepsilon^{\prime}, \alpha)$ of $P$ (on an input state $\ket{\Psi}$) to propagate onto the error $\varepsilon$ in the output state, making the latter in general greater than $\varepsilon^{\prime}$. The exact dependence of $\varepsilon$ on $\varepsilon^{\prime}$ is dictated by $p_{\Psi}(\beta,\varepsilon^{\prime}, \alpha)$. However, if $\varepsilon^{\prime}\leq\,\varepsilon\, \sqrt{p_{\Psi}(\beta, \alpha)}/2$, with $p_{\Psi}(\beta, \alpha)\coloneqq p_{\Psi}(\beta,0, \alpha)=\alpha^2\, \left\Vert F_{\beta}(H)\ket{\Psi}\right\Vert^2$, the output-state error i{s  $\mathcal{O}(\varepsilon)$ \cite[Sec. II]{TLS2022Sup}, %(see App.\ \ref{app:post_prob_prop}),
with $``\mathcal{O}(\cdot)"$ standing for ``asymptotically upper-bounded by".} In turn, the primitives must be incorporated into master algorithms which we formaly define below. 
%%%%%%%%%%%%%%%%%%%%%%%%%%%%%%%%%%%%%%%%%%%%%%%%%%%%%%%%
\begin{dfn}\label{def:masters} (Master QITE algorithms). Given $\varepsilon\geq 0$, $\beta\geq0$, $\ket{\Psi}\in\mathbb H_\mathcal{S}$, and $(\beta^{\prime},\varepsilon^{\prime},\alpha^{\prime})$-QITE-primitives $P_{\beta^{\prime},\varepsilon^{\prime},\alpha^{\prime}}$ querying oracles for a Hamiltonian $H$, a $(\beta,\varepsilon)$-master-QITE-algorithm for $H$ on $\ket{\Psi}$ is a procedure that outputs the state $\frac{F_{\beta}(H)\ket{\Psi}}{\left\Vert F_{\beta}(H)\ket{\Psi}\right\Vert}$ up to trace-distance error $\varepsilon$ with unit probability. Its overall query complexity $Q(\beta,\varepsilon)$  is the sum over the query complexities of each $P_{\beta^{\prime},\varepsilon^{\prime},\alpha^{\prime}}$ applied.
 \end{dfn}

\subsection{Quantum signal processing  \label{sec:methods}}

Quantum signal processing (QSP) is a powerful method to obtain 
an $\varepsilon^\prime$-approximate block encoding of an operator function
$f(H)\coloneqq\sum_{\lambda}f(\lambda)\ket{\lambda}\bra{\lambda}$, where $\{\ket{\lambda}\in\mathbb H_\mathcal{S}\}$ are the eigenvectors and $\{\lambda\}$ the eigenvalues of a Hamiltonian $H$, from queries to an oracle for $H$ \cite{low2017optimal}. We note that QSP can also be extended to non-Hermitian operators \cite{gilyen2019}, but here we restrict to the Hermitian case for simplicity.
%QSP maps the problem of function design to $SU(2)$ rotations. In this section,
We present two QSP methods for general functions one for each
oracle model in Defs. 2 and  3.  
%However, our main novelty is to show that QSP
%can be used to produce Hamiltonian functions from a real-time
%evolution Hamiltonian oracle. 
%The results presented here are later particularized for the {QI}TE propagator $F_{\beta}(H)=e^{-\beta(H-\lambda_{\rm min})}$.
Our QITE primitives are then obtained by particularizing these methods to the case $f(H)=F_{\beta}(H)$, with $F_{\beta}(H)=e^{-\beta(H-\lambda_{\rm min})}$.
%%%%%%%%%%%%%%%%%%%%%%%%%%%%%%%%%%%%%%%%%%%%%%%%
%%%%%%%%%%%%%%%%%%%%%%%%%%%%%%%%%%%%%%%%%%%%%%%%
\subsubsection{Real-variable function design with single-qubit rotations}
\label{sec:Real_f_design}
We start by reviewing  how to approximate functions of one real variable with single-qubit pulses. 
%%%%%%%%%%%%%%%%%%%%%%%%%%%%%%%%%%%%%%%%%%%%%%%%

\textbf{Single-qubit QSP method 1.} 
\label{subsec:Real_f_design1}
Consider the single qubit rotation
$R_1(\theta,\phi)\coloneqq e^{i\theta{X}}e^{i\phi{Z}}$, where ${X}$
and ${Z}$ are the first and third Pauli matrices, respectively,
and $\phi\in[0,2\pi]$. The angle $\theta\in[-\pi,\pi]$
is the signal to be processed and the rotation $e^{i\theta{X}}$
is called the iterate. One can show \cite{Low2016PRX} that, given $q\in \mathbb{N}_{\text{even}}$ and a sequence of angles ${\boldsymbol{\Phi}_1}=\big(\phi_{1},\cdots,\phi_{q+1}\big)\in\mathbb{R}^{q+1}$, the sequence of rotations $\mathcal{R}_1\left(\theta,{\boldsymbol{\Phi}_1}\right)\coloneqq e^{i\phi_{q+1}{Z}}\prod_{k=1}^{q/2}R_1(-\theta,\phi_{2k})R_1(\theta,\phi_{2k-1})$
 has matrix representation in the computational basis
\begin{equation}
\mathcal{R}_1\left(\theta,{\boldsymbol{\Phi}_1}\right)=\left(\begin{array}{cc}
B(\cos\theta) & i\,\sin\theta\,D(\cos\theta)\\
i\,\sin\theta\,D^*(\cos\theta) & B^{*}(\cos\theta)
\end{array}\right),\label{eq:poly}
\end{equation}
where  $B$ and $D$ are polynomials in $\cos\theta$ with complex coefficients
determined by $\boldsymbol{\Phi}_1$. 

For target real polynomials $\mathscr{B}(\cos\theta )$ and $\mathscr{D}(\cos\theta)$, we wish to find ${\boldsymbol{\Phi}_1}$ that generates $B(\cos\theta)$ and $D(\cos\theta)$ with $\mathscr{B}(\cos\theta )$ and $\mathscr{D}(\cos\theta)$ as either their real or imaginary parts, respectively.
%In general, one is interested in the converse problem: given certain
%polynomials $B(\cos\theta)$ and $D(\cos\theta)$, is there a set
%of angles ${\boldsymbol{\Phi}}$ such that these polynomials are obtained as
%the entries of an operator $\mathcal{R}\left(\theta,{\boldsymbol{\Phi}}\right)$? A less restrictive problem is to find when a pair of real polynomials $\mathscr{B}(\cos\theta )$, $\mathscr{D}(\cos\theta)$ can be obtained as the real/imaginary parts of an achievable pair $B(\cos\theta)$, $D(\cos\theta)$, leaving the imaginary/real part of the QSP polynomials free. 
This can be done iff they satisfy \cite[Sec. XII]{TLS2022Sup} %(App.\ \ref{app:qsp})
\begin{equation}
\label{achievabilityCond1}
 \mathscr{B}^{2}(\cos\theta)+\sin^2\theta \,\mathscr{D}^{2}(\cos\theta)\leq1
\end{equation}
for all $\theta$, and have the form 
\begin{equation} \label{achievabilityCond2}
\begin{split}
  \mathscr{B}(\cos\theta)& =\sum_{k=0}^{q/2}b_k\cos(2k\theta)\\
  \sin\theta\,\mathscr{D}(\cos\theta)& =\sum_{k=1}^{q/2} d_{k}\sin\left(2k\theta\right),
  \end{split}
\end{equation}
with $b_k\in\mathbb{R}$ and $d_k\in\mathbb{R}$.
Alternatively, Eq.\ \eqref{achievabilityCond2} can also be expressed in terms of Chebyshev polynomials of first $T_{k}(\cos\theta)\coloneqq\cos(k\theta)$ and second $U_{k}(\cos\theta)\coloneqq\sin\left((k+1)\theta\right)/\sin\theta$ kinds.  This can be used to obtain either Chebyshev or Fourier series of target operator functions.
If the target expansion satisfies Eqs. \eqref{achievabilityCond1} and \eqref{achievabilityCond2}, the angles ${\boldsymbol{\Phi}_1}$ can be computed classically in time
$\mathcal{O}\left(\text{poly}(q)\right)$ \cite{Low2016PRX,Haah2019product,chao2020finding,dong2020efficient}.
% (see App.\ \ref{app:QSPpulses}).

%%%%%%%%%%%%%%%%%%%%%%%%%%%%%%%%%%%%%%%%%%%%%%%%
\textbf{Single-qubit QSP method 2.} 
\label{subsec:Real_f_design2}
This method is inspired by a construction in Ref. \cite{PerezSalinas2021} and shown in detail in a companying paper \cite{TLS2022}. The fundamental gate is $R_2(x,\omega,\zeta,\eta,\varphi,\kappa)=e^{i\frac{\zeta+\eta}{2} Z}e^{-i\varphi Y}e^{i\frac{\zeta-\eta}{2} Z}e^{i\omega x Z}e^{-i\kappa Y}$, which has five adjustable parameters $\{\omega,\boldsymbol{\xi}\}\in\mathbb{R}^5$, where $\boldsymbol{\xi}\coloneqq\{\zeta,\eta,\varphi,\kappa\}$. Here, $x\in\mathbb{R}$ will play the role of the signal and $e^{i \omega x Z}$ that of the iterate. In Ref. \cite{PerezSalinas2021}, it was observed that the gate sequence $\mathcal{R}_2(x,\boldsymbol{\omega},\boldsymbol{\Phi}_2)\coloneqq\prod_{k=0}^{q}R_2(x,\omega_k,\boldsymbol{\xi}_k)$, with $\boldsymbol{\omega}\coloneqq\{\omega_0,\cdots, \omega_q\}\in\mathbb{R}^{q+1}$ and $\boldsymbol{\Phi}_2=\{\boldsymbol{\xi}_0,\cdots,\boldsymbol{\xi}_{q}\}\in\mathbb{R}^{4(q+1)}$, can encode certain finite Fourier series into its matrix components. In \cite{TLS2022}, not only it is formally proven that for any target series a unitary operator  can be built with it as one of its matrix elements but also we provide an explicit, efficient recipe for finding the adequate choice of pulses $\boldsymbol{\Phi}_2$. This is the content of the following lemma.
\begin{lem}\label{lem:qsp2}[Single-qubit Fourier series synthesis]
 Given $\tilde{g}_q(x)=\sum_{m=-q/2}^{q/2} c_m\, e^{im x}$, with $q\in\mathbb{N}$ even, there exist $\boldsymbol{\omega}$ and  $\boldsymbol{\Phi}_2$ such that $\bra{0}\,\mathcal{R}_2(x,\boldsymbol{\omega},\boldsymbol{\Phi}_2)\,\ket{0}=\tilde{g}_q(x)$ for all $|x|\leq \pi$ iff $|\tilde{g}_q(x)|\leq 1$ for all $|x|\leq\pi$. Moreover,  $\omega_0=0$ and $\omega_k=(-1)^k/2$, for all $1\leq k\leq q$, and $\boldsymbol{\Phi}_2$ can be calculated classically from $\{c_m\}_{m}$ in time $\mathcal{O}\left(\text{poly}(q)\right)$.
\end{lem}

%%%%%%%%%%%%%%%%%%%%%%%%%%%%%%%%%%%%%%%%%%%%%%%%%%%%%%%%%
%%%%%%%%%%%%%%%%%%%%%%%%%%%%%%%%%%%%%%%%%%%%%%%%%%%%%%%%%
\subsubsection{Operator-function design from block-encoded oracles}
\label{sec:QSP_BE}
Here, we synthesize an ($\varepsilon'$,1)-block-encoding of $f(H)$ from queries to an oracle for $H$ as in Def.\,\ref{def:block_enc_or}. The algorithm can be seen as a variant of the single-ancilla method from Ref.\,\cite{Low2019hamiltonian} with slightly different pulses.
The basic idea is to design a circuit, $P_1$, that generates a perfect block-encoding $V_{\boldsymbol{\Phi}_1}$ of a target Chebyshev expansion $\tilde{f}_q(H)\coloneqq\sum^{q/2}_{k=0} b_k \,T_k(H)$ that $\varepsilon_\text{tr}$-approximates $f(H)$, for some $0\leq \varepsilon_\text{tr} \leq \varepsilon'$. This can be done by adjusting $\boldsymbol{\Phi}_1$ as in Sec.\ \ref{subsec:Real_f_design1}. Note that the achievability condition \eqref{achievabilityCond1} requires that $\|\tilde{f}_q(H)\|\leq 1$, but we only guarantee $\|\tilde{f}_q(H)\|\leq 1+\varepsilon_\text{tr}$. However, this can be easily accounted for introducing an inoffensive sub-normalization $\alpha=(1+\varepsilon_\text{tr})^{-1}$ (see, e.g., Lemma 14 in Ref. \cite{Low2019hamiltonian}), which we neglect here throughout. Choosing $\tilde{f}_q$ as the truncated Chebyshev series of $f$ with truncation error $\varepsilon_\text{tr} \leq \varepsilon'$, we obtain the desired block-encoding of $f(H)$. 
For $f$ analytical, the error fulfills \cite{elliot1987} 
\begin{equation}
\label{eq:bound_error}
\varepsilon_\text{tr}\leq\frac{\underset{\lambda\in[\lambda_{\text{min}},\lambda_{\text{max}}]}{\max}\left|f^{(q/2+1)}(\lambda)\right|}{2^{\frac{q}{2}}(q/2+1)!},
\end{equation}
with $f^{(q/2+1)}$ the  $(q/2+1)$-th derivative of $f$. This allows one to obtain the truncation order $q/2$ and Chebyshev coefficients $\boldsymbol{b}\coloneqq\{b_k\}_{0\leq k\leq q/2}$ \cite[Sec. XIII]{TLS2022Sup}. %(see App.\ \ref{App:Chebyshev}).
Then, from $\boldsymbol{b}$, one can calculate the required $\boldsymbol{\Phi}_1$ \cite[Sec. XIV]{TLS2022Sup}. %(see App.\ \ref{app:QSPpulses}).

Next, we explicitly show how to generate $V_{\boldsymbol{\Phi}_1}$. Using the short-hand notation $\ket{0_{\lambda}}\coloneqq\ket{\lambda}\ket{0}\in\mathbb{H}_{\mathcal{SA}_{O_1}}$ and Def.\,\ref{def:block_enc_or}, one writes  $O_1\ket{0_{\lambda}}=\lambda\ket{0_{\lambda}}+\sqrt{1-\lambda^{2}}\ket{0_{\lambda}^{\perp}}$ with $\braket{0_{\lambda}}{0_{\lambda}^{\perp}}=0$. This defines the 2-dimensional subspace $\mathbb{H}_\lambda\coloneqq\text{span}\{\ket{0_{\lambda}},\ket{0_{\lambda}^{\perp}}\}$. To exploit the single-qubit formalism from Sec.\ \ref{subsec:Real_f_design1}, one needs an iterate that acts as an $SU(2)$ rotation within each $\mathbb{H}_\lambda$. In general, $O_1$ itself is not appropriate for this due to leakage out of $\mathbb{H}_\lambda$ by repeated applications of $O_1$. However, there is a simple oracle transformation -- qubitization -- that maps $O_1$ into another block-encoding $O'_1$ of the same $H$ but with the desired property \cite{Low2019hamiltonian}. The transformed oracle reads \cite[Sec. XV]{TLS2022Sup} %(see App.\ \ref{app:qubitization}) 
\begin{equation}
 O'_1 =\bigoplus_{\lambda}e^{-i\theta_\lambda {Y}_\lambda},
\end{equation}
with $\theta_\lambda{:}=\cos^{-1}(\lambda)$ and $Y_\lambda\coloneqq i(\ket{0_{\lambda}^{\perp}}\bra{0_{\lambda}} - \ket{0_{\lambda}}\bra{0_{\lambda}^{\perp}})$.

Although the qubit resemblance could be considered in a direct analogy to QSP
for a single qubit, it leads to a more strict class of achievable
functions than if we resort to one additional qubit (single-ancilla QSP) \cite{Low2019hamiltonian}. This extra
ancilla controls the action of the oracle $O'_1$ through the iterate 
\begin{equation}\label{eq:v0be}
V_{0}=\mathds{1}\otimes\ket{+}\bra{+}+O'_1\otimes\ket{-}\bra{-} 
\end{equation}
on $\mathbb{H}_{\mathcal{SA}}$, where $\ket{\pm}$ are the
eigenstates of the Pauli operator ${X}$  for the QSP qubit ancilla. Throughout this section, $\mathds{1}$ is the identity operator on $\mathbb{H}_{\mathcal{SA}_{O_1}}$ and $M$ denotes the single-qubit Hadamard gate.
Let us define the operators $V_{\phi}=V_{0}\left(\mathds{1}\otimes e^{i\phi{Z}}\right)$ and $\bar{V}_{\phi}=V_{0}^\dagger\left(\mathds{1}\otimes e^{i\phi{Z}}\right)$ for a given phase $\phi\in[0,2\pi]$, which play the role of $R(\theta,\phi)$ of the previous sub-section  with $\theta_\lambda$ playing the role of $\theta$ for each $\lambda$. These operators can be phase iterated to generate
\begin{equation}
\label{eq:Vvecphi}
 V_{{\boldsymbol{\Phi}_1}}\coloneqq W_{\text{out}}\left(%\prod_{k=1}^{q}V_{\phi_{2k-1}}\bar{V}_{\phi_{2k}}
 \bar{V}_{\phi_{q}}{V}_{\phi_{q-1}}\cdots \bar{V}_{\phi_{2}}{V}_{\phi_{1}} \right)W_{\text{in}}
\end{equation}
on $\mathbb{H}_{\mathcal{SA}}$, with ancilla pre- and post-processing unitaries 
$W_{\text{in}}\coloneqq\mathds{1}\otimes M$
and
$W_{\text{out}}: = \mathds{1}\otimes \big(M \,e^{i\phi_{q+1}{Z}}\big)$,
respectively, with $M$ the single-qubit Hadamard matrix. The resulting circuit, $P_1$, is depicted in Figs. \ref{fig:generalf}.a) and \ref{fig:generalf}.b). 

The following pseudocode gives the entire procedure. 
\begin{center}
{\setlength{\fboxsep}{1pt}
\begin{minipage}[t]{0.9\columnwidth}
\centering
\begin{algorithm}[H]\label{alg:fBlockEncode}
\SetAlgoLined
\SetAlgorithmName{Algorithm}{Primitive}{Primitive}
\caption{ Operator-function design from block-encoded Hamiltonian oracles}
\SetKwInOut{Input}{input}\SetKwInOut{Output}{output}
\SetKwData{Even}{even}
\Input{analytical function $f:\text{Dom}\to$ $\text{Img}$, with $[\lambda_\text{min},\lambda_\text{max}]\subseteq\text{Dom}$ and $\text{Img}\subseteq[-1,1]$, error $\varepsilon'>0$, oracle $O_1$ for $H$, and its inverse $O^{\dagger}_1$.}
\Output{a unitary quantum circuit $P_{1}$.} 
\BlankLine
Obtain a truncation order $q/2$ s.t. $\varepsilon_\text{tr}\leq\varepsilon'$\;
calculate the Chebyshev coefficients $\boldsymbol{b}$\;
calculate the rotation angles $\boldsymbol{\Phi}_1$ \cite[Sec. XIV]{TLS2022Sup}\;%(App.\ \ref{app:QSPpulses})\;
\Begin(construction of $P_{1}$:){
apply $W_{\text{in}}=\mathds{1}\otimes M$  
on $\mathcal{A}$\;
\For{$k=1 $ \KwTo $k=q$}{apply $\mathds{1}\otimes e^{i\phi_{k}Z}$\;
   \lIf{$k$ is odd}{apply $V_0$ from Eq.\ \eqref{eq:v0be}} 
   \lElse{apply $V_0^\dagger$}
}
apply $W_{\text{out}}=\mathds{1}\otimes \big(M \,e^{i\phi_{q+1}{Z}}\big)$ %from Eq.\ \eqref{eq:WoutBE} 
on $\mathcal{A}$\;}
\end{algorithm}
\end{minipage}
}
\end{center}
%%%%%%%%%%%%%%%%%%%%%%%%%%%%%%%%%%%%%%%%%%%%%%%%%%%%%%%%%
\begin{figure}[t!]
\centering
\includegraphics[width=0.6\columnwidth]{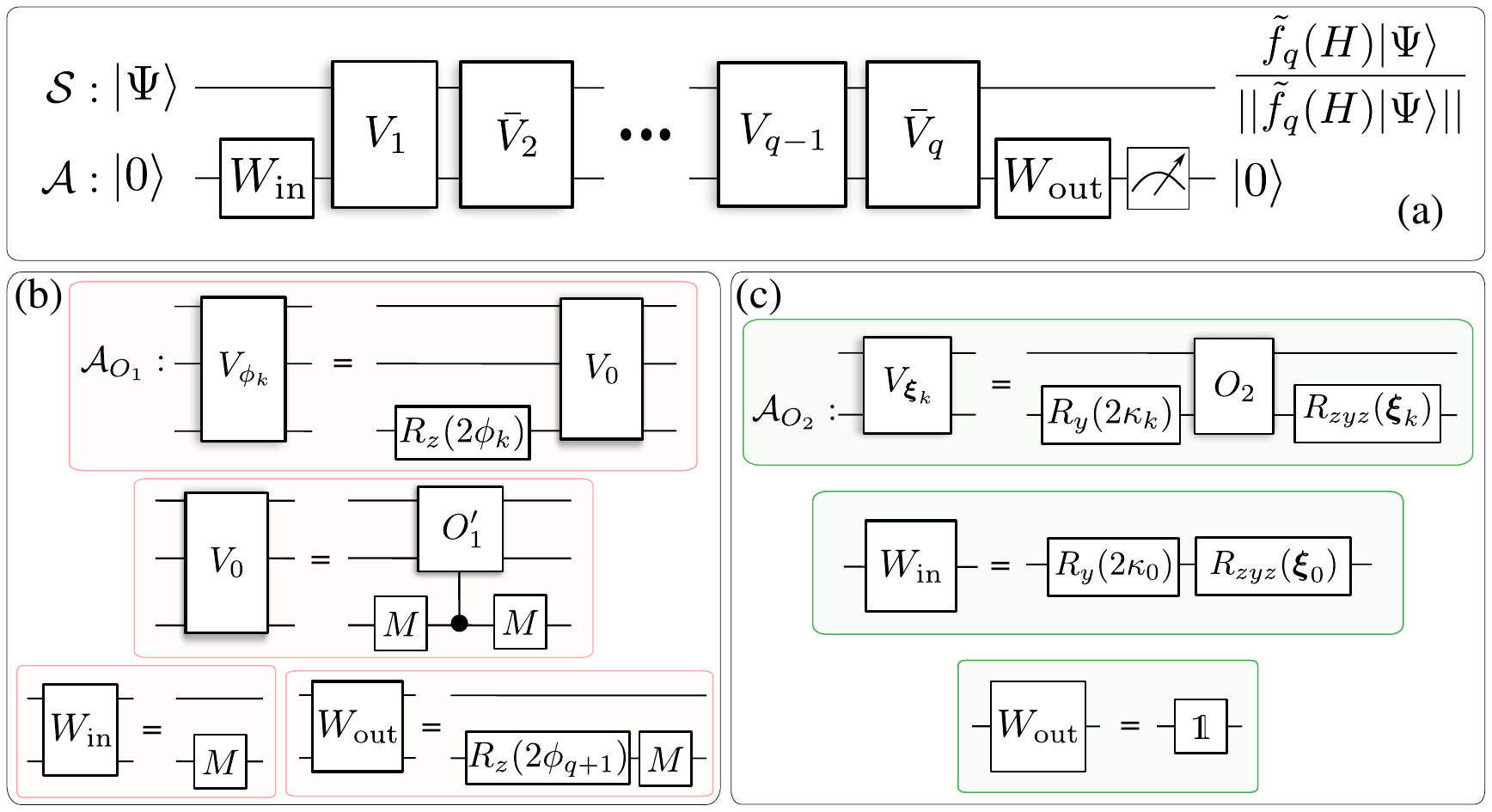}
\caption{\label{fig:generalf} \textbf{QSP primitives for generic operator function design.} a) Both circuits $P_1$ from Alg. \ref{alg:fBlockEncode} and $P_2$ from Alg. \ref{alg:frealtime} have the same structure. If the ancillas are initialised and post-selected in $\ket{0}_{\mathcal{A}}$, the circuit prepares the system state $\frac{\tilde{f}_q(H)\ket{\Psi}}{||\tilde{f}_q(H)\ket{\Psi}||}$, which $\varepsilon$-approximates the target output $\frac{f(H)\ket{\Psi}}{\|f(H)\ket{\Psi}\|}$. 
The details specific to $P_1$ and $P_2$ are respectively shown in panels b) and c). $W_{\text{in}}$ and $W_{\text{out}}$ are fixed ancillary unitaries, and $M$ is a single-qubit Hadamard gate. The basic blocks $V_k$ in panel a) represent the gates $V_{\phi_k}$ in b) and $V_{\boldsymbol{\xi}_k}$ in c).
Each $V_{\phi_k}$ involves one query to the qubitized oracle $O'_1$, which in turn requires one query to $O_1$ and one to its inverse $O^{\dagger}_1$ \cite[Fig. S8]{TLS2022Sup}. %(see Fig.\ \ref{fig:qubitization}). 
Whereas each $V_{\boldsymbol{\xi}_k}$ involves one query to the oracle $O_2$. $\bar{V}_k$ is defined as $V_k$ but with $O_1^{\prime\dagger}$ substituting $O'_1$ or $O_2^{\dagger}$ substituting $O_2$. Hence, the query complexities of $P_1$ and $P_2$ are respectively $2q$ and $q$. The approximating function $\tilde{f}_q$ is determined by the angles ${\boldsymbol{\Phi}_1}=\big(\phi_{1},\cdots,\phi_{q+1}\big)$  or $\boldsymbol{\Phi}_2=\{\boldsymbol{\xi}_0,\cdots,\boldsymbol{\xi}_{q}\}$ in the rotations $R_z(2\phi_k)=e^{i\phi_k{Z}}$ or $R_y(2\kappa_k)=e^{i\kappa_k Y}$ and $R_{zyz}(\boldsymbol{\xi}_k)=R_z(\zeta_k+\eta_k)\,R_y(2\varphi_k)\,R_z(\zeta_k-\eta_k)$, with $\boldsymbol{\xi}_k=\{\zeta_k,\eta_k,\varphi_k,\kappa_k\}$. For $P_1$ and $P_2$, these angles are chosen such that $\tilde{f}_q$ is a high-precision Chebyshev and Fourier approximation of  $f$, respectively.
}
 \end{figure}
%%%%%%%%%%%%%%%%%%%%%%%%%%%%%%%%%%%%%%%%%%%%%%%%%%%%%%%%%

The correctness and complexity of Algorithm \ref{alg:fBlockEncode} are addressed by the following lemma, proven in \cite[Sec. XVI]{TLS2022Sup}. %App.\ \ref{app:proofbe_function} .
\begin{lem}\label{be_function} Let $f$, $\varepsilon'$, $O_1$, and $O^{\dagger}_1$ be as in the input of Alg.\,\ref{alg:fBlockEncode}. Then, for any $q$ s.t. $\varepsilon_\text{tr}$ in Eq.\ \eqref{eq:bound_error} is no greater than $\varepsilon'$, there exists $\boldsymbol{\Phi}_1\in\mathbb{R}^{q+1}$ such that $V_{\boldsymbol{\Phi}_1}$ in Eq.\,\eqref{eq:Vvecphi} is a ($\varepsilon^\prime$,1)-block-encoding of $f(H)$. The circuit $P_1$ generating $V_{\boldsymbol{\Phi}_1}$ requires a single-qubit ancilla, $q$ queries to $O_1$ and $O^{\dagger}_1$ each, 
and $g_{P_1}=\mathcal{O}\left(g_{O_1}+|\mathcal{A}_{O_1}|\right)$ gates per query, with $g_{O_1}$ the gate complexity of $O_1$.
Furthermore, the classical runtime (calculations of $\boldsymbol{b}$ and $\boldsymbol{\Phi}_1$) is within complexity $\mathcal{O}\big({\rm poly}(q/2)\big)$.
\end{lem}

Some final comments about the input function are in place. The restriction of $f$ being analytical is needed to determine the truncation order through Eq.\ \eqref{eq:bound_error}. In fact, to evaluate the RHS of the equation exactly, one needs in general closed-form expression for $f$. However, if the required truncation order is given in advance, the corresponding Chebyshev coefficients can be obtained from $q/2+1$ evaluations of $f$ in specific points (the nodes of the Chebyshev polynomials). In that case, a closed-form expression for $f$ is not required and a classical oracle for evaluating it suffices. 
Moreover, it is important to note that a satisfactory Chebyshev approximation is guaranteed to exist for all bounded and continuous functions \cite{Fraser1965}. If the Chebyshev expansion is given, then step 1 of Algorithm \ref{alg:fBlockEncode} can obviously be skipped and $f$ is not required at all. We further note that Alg. \ref{alg:fBlockEncode} can also be applied even to non-continuous functions over restricted domains without the discontinuities. This is for instance the case of the inverse function, which can be well-approximated over the sub-domain $[-1,-\delta]\cup[\delta,1]$ by a pseudo-inverse polynomial of $\delta$-dependent degree \cite{Childs2017}.

%%%%%%%%%%%%%%%%%%%%%%%%%%%%%%%%%%%%%%%%%%%%%%%%%%%%%%%%%
%%%%%%%%%%%%%%%%%%%%%%%%%%%%%%%%%%%%%%%%%%%%%%%%%%%%%%%%%
\textbf{QITE-primitive from a block-encoding oracle.} 
\label{sec:QuITE_be_or}
QITE primitive 1 corresponds to the output of Alg. \ref{alg:fBlockEncode}  for $f(\lambda)=F_{\beta}(\lambda)=e^{-\beta(\lambda-\lambda_{\rm min})}$. The Chebyshev coefficients can be readily obtained from the Jacobi-Anger expansion \cite{abramowitz1966} 
\begin{equation}
\label{eq:jacobianger}
e^{-\beta\lambda}=I_{0}(\beta)+2\sum_{k=1}^{\infty}(-1)^{k}I_{k}(\beta)\,T_{k}({\lambda}),
\end{equation}
where $I_{k}(\beta)$ is a modified Bessel function. The proof of Theorem \ref{teoQuITEbe} thus follows straightforwardly from Lemma \ref{be_function}. 

\begin{proof}[Proof of Theorem \ref{teoQuITEbe}.] The function $F_{\beta}:[\lambda_{\rm min},\lambda_{\rm max}]\rightarrow(0,1]$, with $F_{\beta}(\lambda)=e^{-\beta(\lambda-\lambda_{\rm min})}$ for all $\lambda\in[\lambda_{\rm min},\lambda_{\rm max}]$ satisfies all the assumptions of Lemma \ref{be_function}.  Hence, on input $f=F_{\beta}$, Alg. \ref{alg:fBlockEncode} outputs an $(\beta,\varepsilon^\prime,1)$-{QI}TE-primitive. By Eq.\ \eqref{eq:bound_error}, the corresponding truncation error is ($q'\coloneqq q/2$)
{\begin{align}
\nonumber
\varepsilon_\text{tr}\leq\frac{\beta^{q'+1}}{2^{q'}(q'+1)!} & \leq\sqrt{\frac{2}{\pi(q'+1)}}\left(\frac{e\beta}{2(q'+1)}\right)^{q'+1}\label{eq:error}\\
 & \leq\left(\frac{\beta e}{2q'}\right)^{q'}, 
\end{align}
where Stirling inequality has been invoked and we assumed $e\beta/2\leq q'$}. (We note also that the
first inequality can also be obtained from explicit summation using Eq.\ \eqref{eq:jacobianger} and the properties of the Bessel functions
\cite{luke1972}.) Then, imposing $\big(\frac{\beta e}{2q'}\big)^{q'}\leq\varepsilon^\prime$
and solving for $q'$ \cite{gilyen2019} gives the query complexity of Eq.\ \eqref{eq:BEquery}.
 \end{proof}
 
 {Primitive 1 is based on the Jacobi-Anger expansion \cite{abramowitz1966}. {This gives a Chebyshev-polynomial series \cite{Fraser1965,elliot1987} for the exponential function, which can be synthesized with quantum signal processing (see Sec.\ \ref{sec:methods}). The expansion has been applied to real-time evolution \cite{Berry2015a,low2017optimal,Low2019hamiltonian,gilyen2019} and even to the QITE propagator $F_{\beta}(H)$ \cite{CSS21}, for partition function estimation. However, the algorithm from \cite{CSS21} performs only a statistical simulation of $F_{\beta}(H)$ based on post-processing and hence cannot simulate QITE on states. In particular, it cannot be used for Gibbs-state sampling, e.g. Moreover, the query complexity from \cite{CSS21} is $\mathcal{O}\big(N+ \beta + \ln(1/\varepsilon')\big)$, which is worse than Eq.\ \eqref{eq:BEquery} in that it contains the extra term $N$ and lacks the denominator in the second term of Eq.\ \eqref{eq:BEquery}. Traditionally %\cite{Berry2015a,low2017optimal,Low2019hamiltonian,gilyen2019, CSS21},
\cite{Berry2015a}, the truncation error in the expansion is bounded u}sing properties of the Bessel functions \cite{luke1972}. In contrast, here, we use a generic upper bound (Lemma \ref{be_function}, in Methods) fo{r a}rbitrary Hermitian-operator functions. This gives the same bound {as \cite{Berry2015a} for the exponential} but holds for any analytical real function, {hence being useful i}n general.}
 
{A further remark about the query complexity of P$_1$. The solution for $q'$ satisfying $\big(\frac{\beta e}{2q'}\big)^{q'}\leq\varepsilon^\prime$ given in Ref.\ \cite{gilyen2019} is based on upperbounds for $q'$ in two regimes. When $\beta\geq 2 \ln(1/\varepsilon')/e^2$, it is shown that $q'\leq e^2\beta/2$. On the other hand, it applies that $q'\leq 4\ln(1/\varepsilon')/\ln(e+2\ln(1/\varepsilon')/(e\beta))$ for $\beta\leq 2 \ln(1/\varepsilon')/e^2$. Therefore, for any $\beta$, \begin{equation}
q=2q'\leq 8\left[\frac{e\beta}{2}+\frac{\ln(1/\varepsilon')}{\ln(e+2\ln(1/\varepsilon')/(e\beta))}\right]                                                                                                                                                                                                                                                                                                                                                                                                                                                                   
                                                                                                                                                                                                                                                                \end{equation}
is a valid upperbound for the query complexity. Consequently, the muliplicative factor implied by the $\mathcal{O}()$ notation in Eq.\ \eqref{eq:BEquery} is known and modestly equal to $8$.}

%%%%%%%%%%%%%%%%%%%%%%%%%%%%%%%%%%%%%%%%%%%%%%%%
%%%%%%%%%%%%%%%%%%%%%%%%%%%%%%%%%%%%%%%%%%%%%%%%
\subsubsection{Operator function design from real-time evolution oracles}
\label{sec:QSP_RTE}
Here, we synthesize an ($\varepsilon^\prime,\alpha$)-block-encoding of $f(H)$ from an oracle for $H$ as in Def. \ref{def:real_t_or}. We proceed as in Sec.\,\ref{sec:QSP_BE}, but with a circuit $P_2$ generating a perfect block-encoding $V_{\boldsymbol{\Phi}_2}$ of a target Fourier expansion $\tilde{g}_q(H)\coloneqq\sum_{m=-q/2}^{q/2} c_m e^{imHt}$ that $\varepsilon_{\text{tr}}$-approximates $\alpha\, f(H)$, for some $\varepsilon_{\text{tr}}\leq\varepsilon'$, $\alpha\leq 1$, and a suitable $t>0$. This is done by adjusting $\boldsymbol{\Phi}_2$ according to Lemma \ref{lem:qsp2}. The function $\tilde{g}_q$ is a Fourier approximation of an intermediary function $g$ such that $g(\lambda,t)\coloneqq g(x_{\lambda})=\alpha\, f(\lambda)$, for $t$ chosen so that $x_{\lambda}=\lambda\, t$ is in the interval of convergence of $\tilde{g}_q$ to $g$ for all $\lambda\in[\lambda_{\text{min}},\lambda_{\text{max}}]$. The reason for this intermediary step here is to circumvent the well-known Gibbs phenomenon, by virtue of which convergence of a Fourier expansion cannot in general be guaranteed at the boundaries. In turn, the sub-normalization factor $\alpha$ arises because our $\tilde{g}_q$ converges to $g$ only for $|x_{\lambda}|<{\pi/2}$, whereas Lemma \ref{lem:qsp2} requires that $|\tilde{g}_q(x_{\lambda})|\leq 1$ for all $|x_{\lambda}|\leq \pi$. This forces one to sub-normalize the expansion so as to guarantee normalization over the entire domain. (As in Sec.\ \ref{sec:QSP_BE}, the inoffensive sub-normalization factor $(1+\varepsilon_\text{tr})^{-1}$ is neglected.)

More precisely, we employ (see Ref.\ \cite{TLS2022}) a construction from Ref.\,\cite{vanApeldoorn2020quantumsdpsolvers} that, given $0<\delta\leq\pi/2$ and a power series that $\frac{\varepsilon_{\text{tr}}}{4}$-approximates $g$, gives $\boldsymbol{c}\coloneqq\{c_m\}_{|m|\leq q/2}$ such that $\tilde{g}_q$ $\varepsilon_{\text{tr}}$-approximates $g$ for all $x_{\lambda}\in[-\pi/2+\delta,\pi/2-\delta]$, if
\begin{equation}
\label{eq:q_rt}
 q\geq\left\lceil\frac{2\pi}{\delta}\ln\left(\frac{4}{\varepsilon_{\text{tr}}}\right)\right\rceil.
\end{equation}
For $f$ analytical, one can obtain the power series of $g$ from a truncated Taylor series of $f$ using that $g(x_{\lambda})=\alpha f(\lambda)$. The truncation order $L$ can be obtained from the remainder:
\begin{equation}
\label{eq:Taylor_remainder}
\frac{\varepsilon_\text{tr}}{4}\leq\frac{\underset{\lambda\in[\lambda_{\text{min}},\lambda_{\text{max}}]}{\max}\left|\alpha\, f^{(L+1)}(\lambda)\right|}{(L+1)!}.
\end{equation}
In turn, the conditions $[\lambda_{\text{min}},\lambda_{\text{max}}]\subseteq[-1,1]$ and $x_{\lambda}\in[-\pi/2+\delta,\pi/2-\delta]$ lead to the natural choice $t=\pi/2-\delta$. (This renders $\tilde{g}_q$ periodic in $x_{\lambda}$ with period $2\pi$.)
In addition, in Ref.\ \cite{TLS2022}  the sub-normalization constant $\alpha$ is bounded in terms of the obtained $t$ and Taylor coefficients $\boldsymbol{a}\coloneqq\{a_l\}_{0\leq l\leq L}$ of $f$. It suffices to take $\alpha$ such that 
\begin{equation}
\label{eq:sub_norm}
\sum_{l=0}^L\big|a_l/(1-2\,\delta/\pi)^l\big|\leq\alpha^{-1}.
\end{equation}
%The sum converges if $\lim_{l\rightarrow \infty}|\frac{a_{l+1}}{a_l}|<1-\frac{2\delta}{\pi}$. 
 Note that $L$ and $\alpha$ are inter-dependent. One way to determine them is to increase $L$ and iteratively adapt $\alpha$ until Eqs. \eqref{eq:Taylor_remainder} and \eqref{eq:sub_norm} are both satisfied. Alternatively, if the expansion converges sufficiently fast (e.g., if $\lim_{l\rightarrow \infty}|\frac{a_{l+1}}{a_l}|<1-\frac{2\delta}{\pi}$), one can simply substitute $L$ in Eq.\ \eqref{eq:sub_norm} by $\infty$.
This is indeed the case with QITE primitives. There, the substitution introduces a slight increase of unnecessary sub-normalization but makes the resulting $\alpha$ independent of $L$, thus simplifying the analysis. 
Then, from the obtained $\boldsymbol{c}$, one can finally calculate the required $\boldsymbol{\Phi}_2$ \cite{TLS2022}.

Next, we explicitly show how to generate $V_{\boldsymbol{\Phi}_2}$. The iterate is now taken simply as the oracle itself: $O_2=\mathds{1}\otimes\ket{0}\bra{0}+ e^{-iHt}\otimes \ket{1}\bra{1}$.
Notice that, in contrast to $O_1$, $O_2$ readily acts as an $SU(2)$ rotation on each 2-dimensional subspace $\text{span}\{\ket{\lambda}\ket{0},\ket{\lambda}\ket{1}\}$. This relaxes the need for qubitization.  
The basic QSP blocks for the unitary operator $V_{{\boldsymbol{\Phi}_2}}\coloneqq W_{\text{out}}\left(
 \bar{V}_{\boldsymbol{\xi}_{q}}{V}_{\boldsymbol{\xi}_{q-1}}\cdots \bar{V}_{\boldsymbol{\xi}_{2}}{V}_{\boldsymbol{\xi}_{1}} \right)W_{\text{in}}$  are 
\begin{subequations}
 \begin{equation}
 \label{eq:vphiRT}
  {V}_{\boldsymbol{\xi}_k}\coloneqq\left[\mathds{1}\otimes\left(e^{i\frac{\zeta_k+\eta_k}{2} Z}e^{-i\varphi_k Y}e^{i\frac{\zeta_k-\eta_k}{2} Z}\right)\right]O_2 \left[\mathds{1}\otimes e^{-i\kappa_k Y}\right],
 \end{equation}
 and 
 \begin{equation}\label{eq:vphibarRT}
  \bar{V}_{\boldsymbol{\xi}_k}\coloneqq\left[\mathds{1}\otimes\left(e^{i\frac{\zeta_k+\eta_k}{2} Z}e^{-i\varphi_k Y}e^{i\frac{\zeta_k-\eta_k}{2} Z}\right)\right]O_2^\dagger \left[\mathds{1}\otimes e^{-i\kappa_k Y}\right]
 \end{equation}
\end{subequations}
with ${\boldsymbol{\xi}_k}\coloneqq\{\zeta_k,\eta_k,\varphi_k,\kappa_k\}$. 
$V_{\boldsymbol{\xi}_k}$ and $\bar{V}_{\boldsymbol{\xi}_k}$ play a similar role to $R_2(x,\omega_k,\boldsymbol{\xi}_k)$ in Sec.\,\ref{subsec:Real_f_design2} (with {$x_\lambda$} inside $O_2$ playing the role of $x$ there for each $\lambda$). Here we take
\begin{equation}
\label{eq:WinRT}
W_{\text{in}}=\mathds{1}\otimes\big[e^{i\frac{\zeta_0+\eta_0}{2} Z}e^{-i\varphi_0 Y}e^{i\frac{\zeta_0-\eta_0}{2} Z}e^{-i\kappa_0 Y}\big]
\end{equation}
and $W_{\text{out}}=\mathds{1}$.
The circuit is depicted in Figs. \ref{fig:generalf}.a) and \ref{fig:generalf}.c).

The following pseudocode presents the entire procedure.

\begin{center}
{\setlength{\fboxsep}{1pt}

%\framebox{
\begin{minipage}[t]{0.9\columnwidth}
\centering
\begin{algorithm}[H]\label{alg:frealtime}
\SetAlgoLined
\SetAlgorithmName{Algorithm}{Algorithm}{Algorithm}
\caption{Operator-function design from real-time evolution Hamiltonian oracles}
\SetKwInOut{Input}{input}\SetKwInOut{Output}{output}
\SetKwData{Even}{even}

\Input{analytical function $f:\text{Dom}\to$ $\text{Img}$, with {$[\lambda_\text{min},\lambda_\text{max}]\subseteq\text{Dom}$ and} $\text{Img}\subseteq[-1,1]$, error $\varepsilon'>0$, $\delta\in(0,\pi/2)$, oracle $O_2$ for $H$ at a time $t={\pi/2-\delta}$, and its inverse $O^{\dagger}_2$.}
\Output{unitary quantum circuit $P_2$.} 
\BlankLine
Obtain $\alpha$, $L$, and $\boldsymbol{a}$ s.t. $\varepsilon_\text{tr}\leq\varepsilon'$\;
obtain $q$ s.t. Eq.\ \eqref{eq:q_rt} holds\;
calculate the Fourier coefficients $\boldsymbol{c}$\;
calculate the rotation angles $\boldsymbol{\Phi}_2$ \cite{TLS2022}\;
\Begin(construction of $P_2$:){
apply $W_{\text{in}}$ from Eq.\ \eqref{eq:WinRT} on $\mathcal{A}$\;
\For{$k= 1$ \KwTo $k=q$}{
   \lIf{$k$ is odd}{apply $V_{\boldsymbol{\xi}_k}$ from Eq.\ \eqref{eq:vphiRT}} %
   \lElse{appy $\bar{V}_{\boldsymbol{\xi}_k}$ from Eq.\ \eqref{eq:vphibarRT}}
}
}
\end{algorithm}
\end{minipage}
}
\end{center}

 The correctness and complexity of Alg. \ref{alg:frealtime} are addressed by the following lemma, proved in Ref.\ \cite{TLS2022}.

\begin{lem}\label{rt_function} Let  $f$, $\varepsilon'$, $\delta$, $O_2$, and $O_2^\dagger$ be as in the input of Alg. \ref{alg:frealtime}. Then, for any $q$ satisfying Eq.\ \eqref{eq:q_rt} and $\alpha$ satisfying Eq.\ \eqref{eq:sub_norm}, there exists $\boldsymbol{\Phi}_2\in\mathbb{R}^{4(q+1)}$ such that $V_{\boldsymbol{\Phi}_2}$ is an ($\varepsilon^\prime,\alpha$)-block-encoding of
 $f(H)$. The circuit $P_2$ generating $V_{\boldsymbol{\Phi}_2}$ requires $q/2$ queries to $O_2$ and $O_2^\dagger$ each and $4+g_{O_2}$ gates per query, with $g_{O_2}$ the gate complexity of $O_2$. Moreover, the classical runtime is within complexity $\mathcal{O}\big(\text{poly}(L,q/2)\big)$.
\end{lem}

Clearly, if a suitable power series for $f$ is a-priori available, analyticity of $f$ is not required and step 1 in Alg. \ref{alg:frealtime} can be skipped. Finally, we note that it is always possible to avoid sub-normalization by introducing a periodic extension of $f$ that is readily normalized over the entire domain of the Fourier expansion. However, Eq.\ \eqref{eq:q_rt} is then no longer valid and one must assess the query complexity on a case-by-case basis. 
{This can for instance be tackled numerically by variationally optimising the gate sequence $\mathcal{R}_2(x,\boldsymbol{\omega},\boldsymbol{\Phi}_2)$ to block-encode the periodic extension of $f$ \cite{PerezSalinas2021}. Either way}, clearly, if a normalized Fourier expansion is a-priori available, one can skip steps 1 to 3 in Alg. \ref{alg:frealtime}.

\textbf{QITE-primitive from a real-time evolution oracle.} 
\label{sec:QuITE_rt_or}
QITE primitive 2 is the output of Alg. \ref{alg:frealtime}  for $f=F_{\beta}$. The proof of Theo. \ref{teoQuITEberealtime} thus follows straight from Lemma \ref{rt_function}.

\begin{proof}[Proof of Theorem \ref{teoQuITEberealtime}]
The function $F_\beta:[\lambda_{\text{min}},\lambda_{\text{max}}]\rightarrow(0,1]$, with $F_{\beta}(\lambda)=e^{-\beta(\lambda-\lambda_{\rm min})}$, satisfies the assumptions of Lemma \ref{rt_function} with $L$ given by Eq.\ \eqref{eq:Taylor_remainder} for $\max_{\lambda\in[\lambda_{\text{min}},\lambda_{\text{max}}]}\big|\alpha\, f^{(L+1)}(\lambda)\big|=\alpha\,\beta^{L+1}$. Its Taylor coefficients are $a_l=\frac{e^{\beta\lambda_{\rm min}}(-\beta)^l}{l!}$, for all $l\in\mathbb{N}$. 
To obtain $\alpha$, we note that $\sum_{l=0}^{L}|a_l/\big(1-\frac{2\delta}{\pi}\big)^l|\leq\sum_{l=0}^{\infty}|a_l/\big(1-\frac{2\delta}{\pi}\big)^l|= e^{\beta\lambda_{\rm min}} e^{\frac{\beta}{1-2\delta/\pi}}$. Hence, by Eq.\ \eqref{eq:sub_norm}, we can take $\alpha=e^{-\beta(\lambda_{\rm min}+\frac{1}{1-2\delta/\pi})}$. Introducing $\gamma\coloneqq-\beta+\frac{\beta}{1-2\delta/\pi}=\frac{\beta\delta}{\pi/2-\delta}$, we re-write $\alpha=e^{-\beta(1+\lambda_{\rm min})-\gamma}$.  
This allows us to specify $e^{-\gamma}$ instead of the Fourier convergence interval, i.e. 
to subordinate $\delta$ to the desired $\gamma$. This is done by fixing $\delta=\frac{\pi}{2}\frac{1}{1+\frac{\beta}{\gamma}}$. This, together with Eq.\ \eqref{eq:q_rt}, leads to Eq.\ \eqref{eq:real_time_or}. 
\end{proof}

{In fact, $4(\beta/\lambda+1)\ln(4/\epsilon')$ is an upper-bound for the query complexity of P$_2$, as can be inferred from Eq.\ \eqref{eq:q_rt}. In other words, the multiplicative factor implied in the $\mathcal{O}()$ in Eq.\ \eqref{eq:real_time_or} notation is actually known to be equal to $4$.}

\subsection{Traditional master QITE algorithms}
\label{app:QuITEmasters} 
The average number of times a QITE primitive is applied  in probabilistic and coherent master QITE algorithms is $1/p_{\Psi}(\beta,\varepsilon^{\prime}, \alpha)$ and $1/\sqrt{p_{\Psi}(\beta,\varepsilon^{\prime}, \alpha)}$, respectively; see Fig.\ \ref{fig:master}. Conveniently, for $\varepsilon\ll1$, $1/p_{\Psi}(\beta,\varepsilon^{\prime}, \alpha)$ can be approximated by the more practical expression $1/p_{\Psi}(\beta,\alpha)$ up to error $\mathcal O(\varepsilon)$. This follows from a  Taylor expansion.
The probabilistic algorithm applies the primitive on independent input state preparations [Fig.\,\ref{fig:master}-a)] and stops at the moment that the first successful postselection on the ancillas happens. In contrast,  the coherent one (Fig.\,\ref{fig:master}-b) leverages quantum amplitude amplification \cite{Brassard2002}, which we briefly discuss next. 

%%%%%%%%%%%%%%%%%%%%%%%%%%%%%%%%%%%%%%%%%%
%%%%%%%%%%%%%%%%%%%%%%%%%%%%%%%%%%%%%%%%%%
\begin{figure*}[t!]
\centering{}\includegraphics[width=1.0\textwidth]{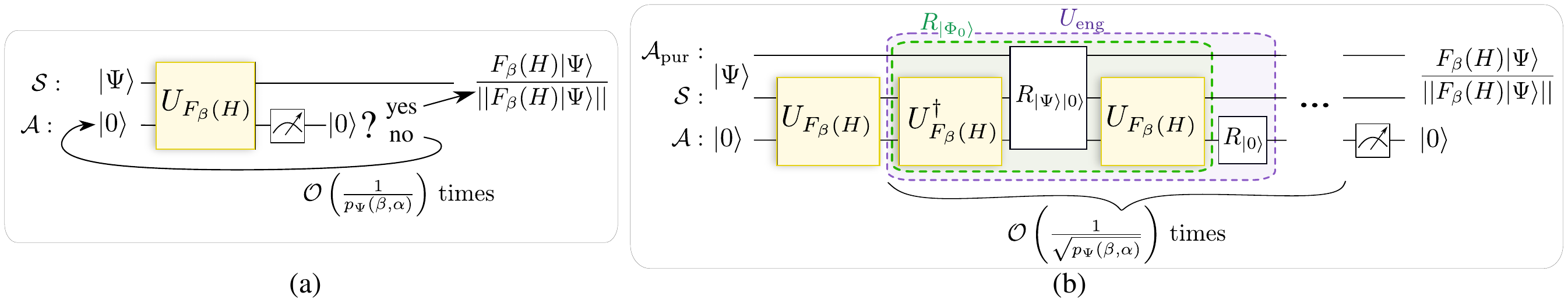}
\caption{ {\bf Probabilistic versus coherent master QITE algorithms}. a) The probabilistic approach repeatedly applies the unitary $U_{F_{\beta}(H)}$ generated by the primitive (on independent preparations of $\ket{\Psi}\ket{0}\in\mathbb H_\mathcal{SA}$) until the post-selection is successful, i.e. the measurement on the ancillas returns $\ket{0}$ as outcome. This takes on average $\mathcal{O}\big(1/p_\Psi(\beta,\alpha)\big)$ repetitions. No a priori knowledge of the input state $\ket\Psi$ is required (it can be fully generic, even mixed). b) The coherent approach is based on quantum amplitude amplification. It operates only on pure input states. So, if the input state is mixed, an extra ancillary register $\mathcal{A}_{\mathrm{pur}}$ of $|\mathcal{A}_{\mathrm{pur}}|=|\mathcal{S}|=N$ qubits is required to purify it. This is the case in quantum Gibbs state sampling, where $\ket\Psi$ is a purification of the maximally mixed state on $\mathbb H_\mathcal{S}$. The primitive is repeated applied sequentially (on the same preparation of $\ket{\Psi}\ket{0}$), interleaved with reflection operators $R_{\ket{\Phi_0}}$ and $R_{\ket{0}}$  around the states $\ket{\Phi_0}\coloneqq U_{F_\beta(H)}\ket{\Psi}\ket{0}\in\mathbb H_\mathcal{SA}$ and $\ket{0}\in\mathbb H_\mathcal{A}$, respectively. In practice, this requires full a priori knowledge of $\ket\Psi$. The total number of repetitions of the primitive is $\mathcal{O}\big(1/\sqrt{p_\Psi(\beta,\alpha)}\big)$, after which the desired output is obtained with probability close to 1. Hence, the coherent master algorithm displays a significantly lower overall query complexity than the probabilistic one. However, in return, the former requires much larger circuit depth than the latter.}
\label{fig:master}%
\end{figure*}
%%%%%%%%%%%%%%%%%%%%%%%%%%%%%%%%%%%%%%%%%%
%%%%%%%%%%%%%%%%%%%%%%%%%%%%%%%%%%%%%%%%%%

The coherent amplification process is realized by repeatedly applying to $\ket{\Phi_0}\coloneqq U_{F_\beta(H)}\ket{\Psi}\ket{0}$ the unitary operator 
\begin{equation}
U_\mathrm{eng}\coloneqq R_{\ket{\Phi_0}}\,R_{\ket{0}},
\label{eq:U_engine}
\end{equation}
where $R_{\ket{\Phi_0}}$ and $R_{\ket{0}}$ are respectively the reflection operators around $\ket{\Phi_0}\in\mathbb H_\mathcal{S}\otimes \mathbb H_{\mathcal{A}_{\rm pur}}\otimes\mathbb H_\mathcal{A}$ and $\ket{0}\in\mathbb H_\mathcal{A}$. The former reflection can in turn be decomposed as $R_{\ket{\Phi_0}}=U_{F_\beta(H)}\,R_{\ket{\Psi}\ket{0}}\,U^{\dagger}_{F_\beta(H)}$, where $U^{\dagger}_{F_\beta(H)}$ is the inverse of the  block-encoding $U_{F_\beta(H)}$ of the QITE propagator $F_\beta(H)$ and $R_{\ket{\Psi}\ket{0}}$ the reflection around $\ket{\Psi}\ket{0}$. Unitary $U_\mathrm{eng}$ is sometimes referred to as the amplification engine. Importantly, $U_\mathrm{eng}$ acts on the 2-dimensional subspace spanned by $\ket{\Phi_0}$ and $\ket{\Phi_{\rm target}}\propto\ketbra{0}{0}\,\ket{\Phi_0}$ as an SU$(2)$ rotation. For $k\in\mathbb{N}$, it gives
\begin{eqnarray}
\nonumber
U_\mathrm{eng}^k\,\ket{\Phi_0}&=&\sin[(2k+1)\theta]\,\ket{\Phi_{\rm target}}\\
&+&\cos[(2k+1)\theta]\,\ket{\Phi_{\perp}} ,
\label{eq:afternsteps}
\end{eqnarray}
where $\ket{\Phi_{\perp}}\propto\ketbra{1}{1}\,\ket{\Phi_0}$ %is the state in Span$\{\ket{\Phi_{\rm target}},\ket{\Phi_0}\}$ orthogonal to $\ket{\Phi_{\rm target}}$ 
and $\sin(\theta)\coloneqq\alpha\|F_\beta(H)\ket{\Psi}\|$. Hence, taking $k=k_{\rm opt}=\mathcal O (1/\theta)$ yields $\sin[(2k+1)\theta]\approx1$ and therefore probability close to 1 for desired output. For $\theta\ll1$, this entails $k=\mathcal O (1/\alpha\|F_\beta(H)\ket{\Psi}\|)=\mathcal O (1/\sqrt{p_\Psi(\beta,\alpha)})$ repetitions of the primitive, as in Eq.\ \eqref{eq:Qcomp}.

Finally, since $\alpha^2\|F_\beta(H)\ket{\Psi}\|$ is in general unknown, one has no a priori knowledge of $k_{\rm opt}$. However, fortunately, this can be accounted for with successive attempts with $k$ randomly chosen within a range of values that grows exponentially with the number of attempts (see \cite[Theorem 3]{Brassard2002}). Remarkably, the resulting average number of applications of the primitive remains within $\mathcal O (1/\sqrt{p_\Psi(\beta,\alpha)})$.

\subsection{ Minimum query complexity of QITE primitives based on block-encoding oracles}
\label{sec:lowerbound}

Our proof strategy for Theorem \ref{th:lowerbound} is analogous to that of the no-fast-forwarding theorem for real-time evolutions \cite{Berry2007,Berry2014,Berry2015a}. That is, it is based on a reduction to QITE of the task of determining the parity ${\rm par}(\boldsymbol{x})\coloneqq x_{0}\oplus x_{1}\oplus \cdots\oplus x_{N-1}$, with $\oplus$ the bit-wise sum, of an unknown $N$-bit string $\boldsymbol{x}\coloneqq x_{0}\,x_{1}\cdots x_{N-1}$ from a parity oracle $U_{\boldsymbol x}$ for $\boldsymbol{x}$; together with known fundamental complexity bounds for the latter task \cite{Farhi1998,Beals1998}. More precisely, our proof relies on three facts: $i$) No algorithm can find ${\rm par}(\boldsymbol{x})$ from $U_{\boldsymbol x}$ with fewer than a known number of queries to it \cite{Farhi1998,Beals1998}; $ii$) a QITE primitive querying an oracle for an appropriate Hamiltonian $H_{\boldsymbol x}$ gives an algorithm to find ${\rm par}(\boldsymbol{x})$; and $iii$) a block-encoding oracle for $H_{\boldsymbol x}$ can be synthesized from one call to $U_{\boldsymbol x}$. The three facts are established in the following lemmas.

The first lemma, proven in \cite{Farhi1998}, lower-bounds the complexity of any quantum circuit able to obtain ${\rm par}(\boldsymbol{x})$ from queries to $U_{\boldsymbol x}$. For our purposes, it can be stated as follows.
\begin{lem}
Let $C$ be a quantum circuit composed of $\boldsymbol x$-independent gates and $q$ times the $\boldsymbol x$-dependent unitary
\begin{equation}
U_{\boldsymbol x} \coloneqq \sum_{j=0}^{N} \ket j \bra j \otimes X^{x_j} \ ,
\label{eq:def_U_x}
\end{equation}
with $\{\ket j\}_{j\in [N+1]}$ an orthogonal basis, such that, acting on an $\boldsymbol x$-independent input state and upon measurement on an $\boldsymbol x$-independent basis, outputs ${\rm par}(\boldsymbol{x})$ with probability greater than $1/2$ for all $\boldsymbol{x}\in\{0,1\}^N$. Then, $q\geq \lceil N/2\rceil$.
\label{lem:parity}
\end{lem}

The second lemma, proven in \cite[Sec. XVII]{TLS2022Sup}, 
%App.\ \ref{app:proofQuITEparity},
reduces parity finding to QITE and is the key technical contribution of this section.
%, is based on the same 2-sparse Hamiltonians used to prove the no-fast-forwarding theorem \cite{Berry2007,Berry2014,Berry2015a}.
%A Hamiltonian matrix is $d$-sparse if its columns (or rows) contain at most $d$ non-null entries each. 
\begin{lem}
\label{lem:QuITEparity}
Let $\beta>0$, $\varepsilon'>0$, $\alpha\in(0,1]$, and $\boldsymbol x$ an $N$-bit string such that 
\begin{equation}
\left|\frac{1-e^{-\frac{\beta}{2N}}}2  \right|^N>\frac{2\,\varepsilon'}{\alpha}\ .
\label{eq:Paritycond}
\end{equation}
 Then, there exists $H_{\boldsymbol x}$, with $\|H_{\boldsymbol x}\| \leq 1$, such that a $(\beta,\varepsilon',\alpha)$-QITE-primitive with calls to a block-encoding oracle for $H_{\boldsymbol x}$, acting on an $\boldsymbol x$-independent input state and upon measurement on an $\boldsymbol x$-independent basis, outputs ${\rm par}(\boldsymbol{x})$ with probability greater than $1/2$ for all $\boldsymbol{x}\in\{0,1\}^N$. 
\end{lem}

Finally, the missing link between Lemmas \ref{lem:parity} and \ref{lem:QuITEparity} is a sub-routine to query $H_{\boldsymbol x}$ given queries to $U_{\boldsymbol x}$  in Eq.\ \eqref{eq:def_U_x}. This is provided by the following lemma, proven in \cite[Sec. XVIII]{TLS2022Sup}.
%App.\ \ref{app:prooffromU_xtoH_x}.
\begin{lem}
A block-encoding oracle for $H_{\boldsymbol x}$ can be generated from a single query to $U_{\boldsymbol x}$ and $\mathcal{O}(N)$ $\boldsymbol x$-independent gates, for all $\boldsymbol{x}\in\{0,1\}^N$. (See \cite[Fig. S10]{TLS2022Sup} %Fig.\ \ref{fig:circuitHfromU} 
for circuit.)
\label{lem:fromU_xtoH_x}
\end{lem}

\begin{proof}[Proof of Theorem \ref{th:lowerbound}]
With these three Lemmas, the proof of Theo. \ref{th:lowerbound} is straightforward. First, note that the left-hand side of Eq.\ \eqref{eq:Paritycond} decreases monotonically in $N$. Hence, for any fixed $(\beta,\varepsilon',\alpha)$, the largest $N\in\mathbb{N}$ that satisfies Eq.\ \eqref{eq:Paritycond} is $N=\lfloor 2\tilde{q}\rfloor$, with $\tilde{q}\in\mathbb{R}$ defined by Eq.\ \eqref{eq:tildeq}. This, together with Lemmas \ref{lem:QuITEparity} and \ref{lem:fromU_xtoH_x}, implies that  any $(\beta,\varepsilon',\alpha)$-QITE-primitive synthesized from queries to the parity oracle $U_{\boldsymbol x}$ provides a quantum circuit to determine the parity of any string $\boldsymbol x$ of length $\lfloor 2\tilde{q}\rfloor$. Then, by virtue of Lemma \ref{lem:parity}, the query complexity  $q_{\min}(\beta,\varepsilon',\alpha)\in\mathbb{N}$ of the primitive cannot be smaller than $q\geq \lceil N/2\rceil=\left\lceil \frac{\lfloor 2\tilde q\rfloor}{2}\right\rceil$. 
Note that this number is the nearest integer to $\tilde{q}$. Ergo, $q_{\min}(\beta,\varepsilon',\alpha)\geq\tilde{q}$.
\end{proof}

Finally, a comment on why Theo. \ref{th:lowerbound} does not hold for QITE primitives based on real-time evolution (RTE) oracles is useful at this point. The reason is that, by virtue of the RTE no-fast-forwarding theorem \cite{Berry2007,Berry2014,Berry2015a}, a single call to an RTE oracle suffices to find the parity with probability greater than $1/2$. Hence, it is the query complexity RTE oracles themselves what is lower-bounded by the parity considerations above, but not that of RTE-based QITE primitives. It is an open question whether similar bounds can be obtained for RTE-based QITE primitives by other arguments.

%%%%%%%%%%%%%%%%%%%%%%%%%%%%%%%%%%%%%%%%%%%%%%%%%%%%%%
\section*{Data and code availabitity}

The datasets  generated and/or an-
alyzed during the current study are available from the corre-
sponding author on reasonable request. The programming codes utilized are available at \url{https://doi.org/10.5281/zenodo.5595705}.
%%%%%%%%%%%%%%%%%%%%%%%%%%%%%%%%%%%%%%%%%%%%%%%%%%%%%%%%
%%%%%%%%%%%%%%%%%%%%%%%%%%%%%%%%%%%%%%%%%%%%%%%%%%%%%%%%
\section*{Acknowledgments}
We thank Martin Kliesch fo{r h}elpful comments that led us to write \cite[Sec. VIII]{TLS2022Sup}. 
%App.\ \ref{app:non_interacting}. 
We acknowledge financial support fro{m %the Technology Innovation Institute of Abu Dhabi, United Arab Emirates, 
t}he Serrapilheira Institute (grant number Serra-1709-17173), and 
the Brazilian agencies CNPq (PQ grant No. {305420/2018-6}) and FAPERJ (PDR10 E-26/202.802/2016 and JCN E-26/202.701/2018). MMT acknowledges also the Government of Spain (FIS2020-TRANQI and Severo Ochoa CEX2019-000910-S), Fundació Cellex, Fundació Mir-Puig, Generalitat de Catalunya (CERCA, AGAUR SGR 1381) and ERC AdG CERQUTE. 
%\end{acknowledgments}
%%%%%%%%%%%%%%%%%%%%%%%%%%%%%%%%%%%%%%%%%%%%%%%%%%%%%%%%
%%%%%%%%%%%%%%%%%%%%%%%%%%%%%%%%%%%%%%%%%%%%%%%%%%%%%%%%
\section*{Competing interest}

The authors declare no competing interests.

\section*{Authors contribution}
LA and TLS conceived the idea. MMT proved Theorem 3, and the lemmas related to it. TLS proved the other theorems with LA contribution.  SC produced the code and numerical results, which TLS revised. LA, TLS, and MMT wrote the manuscript, which  SC revised.

\appendix

\section{Operator-function design from imperfect oracles}
\label{sec:aprooximate_oracles}
{Both algorithms for operator-function design we present assume access to ideal oracles as given by Def. \ref{def:block_enc_or} and Def. \ref{def:real_t_or}. This is not the case in an experimental scenario where only approximate oracles are available. This leads to a total error in the algorithm that is proportional to the number of approximate oracle calls. The following lemma deals with that error for generic operator functions.
\begin{lem}(Primitives from imperfect oracles)
\label{lem:aprooximate_oracles}
Let a circuit with $q$ queries to an oracle $O$ for $H$ generate an $(\varepsilon,\alpha)$-block-encoding $V_{f(H)}$ of $f(H)$, for arbitrary $f$. If $O$ is substituted by $\tilde{O}$, with $\|O-\tilde{O}\|\leq \varepsilon_O$, the circuit generates an $(\tilde{\varepsilon},\alpha)$-block-encoding $\tilde{V}_{f(H)}$ of $f(H)$ with $\tilde{\varepsilon}\leq \varepsilon +q\, \varepsilon_O $.
\end{lem}
 
 \begin{proof}
 For $q\,\varepsilon_O\ll 1$, straightforward matrix multiplication shows that the total error is $||V_{f(H)}-\tilde{V}_{f(H)}||=\mathcal{O}(q\,\varepsilon_O)$. This implies  $||f(H)-\bra{0}\tilde{V}_{f(H)}\ket{0}||\leq \varepsilon_P+q\,\varepsilon_O$ by virtue of the triangle inequality. This concludes the proof.  
\end{proof}

As an exemplary application of Lemma \ref{lem:aprooximate_oracles}, we calculate the gate complexity $g_{\tilde{O}_{2}}$ required to implement an approximate real-time evolution oracle $\tilde{O}_{2}$ for QITE primitive $P_2$ by Trotter-Suzuki-like simulation methods. To attain a  total error $\varepsilon^{\prime}$, we choose $\varepsilon=\varepsilon^{\prime}/2$ and $\varepsilon_O=\varepsilon^\prime/(2q)$.  For real time $t$, the gate complexity of state-of-the-art algorithms \cite{campbell2019,COS19} based on first-order product formulae is $\mathcal{O}(t^2/\varepsilon_O)$. Substituting for $\varepsilon_O$, and then for $t$ and $q$ with the real time and query complexity given in Theorem \ref{teoQuITEberealtime}, we obtain the oracle gate complexity $g_{\tilde{O}_{2}}=\mathcal{O}\left(\frac{\pi^2\beta}{\gamma}\frac{\ln(8/\varepsilon^{\prime})}{2\,\varepsilon^\prime (1+\gamma/\beta)}\right)$.
This scaling is exponentially worse in $1/\varepsilon^{\prime}$} than what we would get if we simulated $O_{2}$ with more sophisticated real-time evolution algorithms \cite{Berry2015a,BCCKS15,low2017optimal,Low2019hamiltonian,gilyen2019}. However, those algorithms assume more powerful oracles, which require significantly more ancillary qubits. The fact that $P_2$ can be synthesised with only 1 ancillary qubit throughout and with no qubitization makes it appealing for near-term implementations on intermediate-sized quantum hardware.
%%%%%%%%%%%%%%%%%%%%%%%%%%%%%%%%%%%%%%%%%%
%%%%%%%%%%%%%%%%%%%%%%%%%%%%%%%%%%%%%%%%%%
\section{Post-selection probability propagation onto output state error}
\label{app:post_prob_prop} 
{Since $\alpha\,F_\beta(H)$ is not unitary, the (operator-norm) error $\varepsilon^{\prime}$ in its approximation by an $(\alpha,\varepsilon')$-block-encoding gets amplified at the post-selection due to state renormalization. The following allows us to control the error in the output state.
\begin{lem}(Error propagation from block encoding to output state)
\label{lem:error_prop}
Let $p_{\Psi}(\beta, \alpha)\leq1$ be the post-selection probability of an $(\alpha,\varepsilon')$-block-encoding $U_{F_\beta(H)}$ of $F_\beta(H)$ on an input state $\ket{\Psi}\in\mathbb H_\mathcal{S}$. Then, if $\varepsilon^{\prime}\leq\,\varepsilon\, \sqrt{p_{\Psi}(\beta, \alpha)}/2$ and $\varepsilon\ll1$, the output-state trace-distance error is $\mathcal{O}(\varepsilon)$.
\end{lem}
 \begin{proof}
%For  $U_{F_\beta(H)}$ an $(\alpha,\varepsilon')$-block-encoding of $F_\beta(H)$, 
%\begin{equation}
%\|\braket{0|\ U_{F_\beta(H)}}{0} - \alpha\, F_\beta(H) \| \leq \varepsilon'.
%\label{eq:distant}
%\end{equation}
By definition, $\|\braket{0|\ U_{F_\beta(H)}}{0} - \alpha\, F_\beta(H) \| \leq \varepsilon'$. 
Then, for some $\ket{\Xi}\in\mathbb H_S$, with $\|\ket{\Xi}\|\leq\varepsilon'$, it is
\begin{subequations}
\begin{equation}
\braket{0|\ U_{F_\beta(H)}}{0}\ \ket{\Psi}=\alpha\, F_\beta(H)\ket{\Psi}+\ket{\Xi}
\end{equation}
and, so, for some $\epsilon\in\mathbb{R}$, with $|\epsilon|\leq\|\ket{\Xi}\|\leq\varepsilon'$, we get
\begin{equation}
\label{eq:explicit_nonperfect_blockenc}
\|\braket{0|\ U_{F_\beta(H)}}{0}\ \ket{\Psi}\|	=\alpha\,\|F_\beta(H)\ket{\Psi}\|+\epsilon.
\end{equation}
\end{subequations}
Next, we Taylor-expand the output state in terms of $\varepsilon$. To that end, note first that 
\begin{multline}
\frac{\braket{0|\ U_{F_\beta(H)}}{0}\ \ket{\Psi}}{\|\braket{0|\ U_{F_\beta(H)}}{0}\ \ket{\Psi}\|}=\frac{\alpha\, F_\beta(H)\ket{\Psi}+\ket{\Xi}}{\alpha\,\|F_\beta(H)\ket{\Psi}\|+\epsilon}\\
 =\frac{F_\beta(H)\ket{\Psi}}{\|F_\beta(H)\ket{\Psi}\|} \left(1 - \frac{\epsilon}{\alpha\,\|F_\beta(H)\ket{\Psi}\|}\right)\\
 + \frac{\ket{\Xi}}{\alpha\,\|F_\beta(H)\ket{\Psi}\|} + \mathcal O\left(\frac{\varepsilon'\,\ket{\Xi}}{\alpha^2\,\|F_\beta(H)\ket{\Psi}\|^2}\right) \ .
\label{eq:actualvector}
\end{multline}
So, the output state-vector error is upper-bounded as
\begin{align}
\left\|\frac{\braket{0|\ U_{F_\beta(H)}}{0}\ \ket{\Psi}}{\|\braket{0|\ U_{F_\beta(H)}}{0}\ \ket{\Psi}\|}\right.-&\left.\frac{F_\beta(H)\ket{\Psi}}{\|F_\beta(H)\ket{\Psi}\|}\right\|  \label{eq:errorvector}\\  \leq \frac{2\,\varepsilon'}{\alpha\,\|F_\beta(H)\ket{\Psi}\|}&+\mathcal O \left(\frac{\varepsilon'^2}{\alpha^2\,\|F_\beta(H)\ket{\Psi}\|^2}\right). \nonumber
\end{align}
Hence, if $\varepsilon'\leq\varepsilon\,\alpha\,\|F_\beta(H)\ket{\Psi}\|/2=\varepsilon\, \sqrt{p_{\Psi}(\beta, \alpha)}/2$, the error in  $l_2$-norm distance of the output state vector is at most $\varepsilon+ \mathcal O(\varepsilon^2)$, which equals $\mathcal O(\varepsilon)$ for $\varepsilon\ll1$. 

With the $l_2$-norm distance between the two state vectors, we can upper-bound the trace distance between their corresponding rank-1 density operators using well-known inequalities.  First, note that the $l_2$-norm distance between two arbitrary normalised vectors $\ket\phi$ and $\ket\psi$ can be expressed as
\begin{equation}
\|\ket\phi-\ket\psi\| = \sqrt2\sqrt{1-{\rm Re}\braket\psi\phi}.
\label{eq:kets2}
\end{equation}
Second, recall that the trace distance 
\begin{equation}
\|\ket\phi\bra\phi-\ket\psi\bra\psi\|_{\tr} \coloneqq \frac12 \tr\sqrt{\left(\ket\phi\bra\phi-\ket\psi\bra\psi\right)^2} \ .
\label{eq:tracedist}
\end{equation}
between rank-1 density operators $\ket\phi\bra\phi$ and $\ket\psi\bra\psi$ can be written in terms of the overlap as 
\begin{equation}
\left\|\ket\phi\bra\phi-\ket\psi\bra\psi\right\|_{\tr}=\sqrt{1-|\braket\psi\phi|^2} \ .
\label{eq:trdistmatrix}
\end{equation}
Then, note that, for all $\ket\phi$ and $\ket\psi$, the RHS of Eq.\ \eqref{eq:kets2} upper-bounds the RHS of Eq.\ \eqref{eq:trdistmatrix}. Hence, $\|\ket\phi\bra\phi-\ket\psi\bra\psi\|_{\tr}\leq\|\ket\phi-\ket\psi\|$. This, together with Eq.\ \eqref{eq:errorvector} gives the promised trace-norm error of the output state.
 \end{proof}}
%%%%%%%%%%%%%%%%%%%%%%%%%%%%%%%%%%%%%%%%%%
%%%%%%%%%%%%%%%%%%%%%%%%%%%%%%%%%%%%%%%%%%

\section{Optimal sub-normalisation for QITE Primitive 2}
\label{app:gamma}

In this appendix, we prove that, for $\varepsilon^{\prime}\ll1$, 
\begin{equation}
\label{eq:gamma_opt_approx}
\gamma_{\kappa}(\beta)\coloneqq\frac{\beta}{2}\left(\sqrt{1+\frac{2}{\mu_\kappa\,\beta}}-1\right).
\end{equation}
gives approximately the optimal value for the subnormalization of $P_2$, such that the overall query complexity
given by Eq.\ \eqref{eq:Qcomp} is minimized. Notice that, by diminishing the value of $\gamma$ we increase the success probability of $P_2$, decresing the number of times it needs to be realized. At the same time, it leads to an increment on the query complexity of $P_2$ given by Eq.\ \eqref{eq:real_time_or}. Therefore, the optimal subnormalization is a tradeoff between these two contributions.

From Eqs. \eqref{eq:Qcomp} and \eqref{eq:real_time_or}, given $\beta$, $\varepsilon$, and the master algorithm type $\kappa$, we get that the optimal $\gamma$ minimizes
\begin{equation}
\begin{split}
 Q_\kappa(\beta,\varepsilon,\gamma)=\frac{e^{2\mu_\kappa\gamma}}{\|F_\beta(H)\ket{\Psi}\|^{2\mu_\kappa}}&\left(\frac{\beta}{\gamma}+1\right)\\
 &\bigg[\ln\bigg(\frac{8}{\|F_\beta(H)\ket{\Psi}\|\varepsilon}\bigg)+g\gamma\bigg],
 \end{split}
\end{equation}
with $g=1$ introduced for convenience.
Here we have used 
$p_{\Psi}(\beta,\gamma)=e^{-2\gamma}\|F_\beta(H)\ket{\Psi}\|$  and $\varepsilon'=p_{\Psi}(\beta,\gamma)\varepsilon/2$.

Instead, let $\gamma=\gamma_\kappa (\beta)$ be such that $Q_\kappa(\beta,\varepsilon,\gamma)|_{g=0}$ is minimized. By solving $\frac{\partial}{\partial \gamma}Q_\kappa(\beta,\varepsilon,\gamma)\big|_{g=0}=0$ we obtain $\gamma_\kappa (\beta)$ given by Eq.\ \eqref{eq:gamma_opt_approx}. In order to prove that that expression is a good approximation for the actual optimal value $\gamma^{(\text{opt})}_{\kappa}(\beta)$, first notice that, with $g=1$, it is easy to verify that  $Q_\kappa(\beta,\varepsilon,\gamma)|_{g=1}>Q_\kappa(\beta,\varepsilon,\gamma_\kappa(\beta))|_{g=1}$ if $\gamma>\gamma_\kappa(\beta)$. Therefore,   $\gamma^{(\text{opt})}_{\kappa}(\beta)<\gamma_{\kappa}(\beta)$. Defining $\Delta Q\coloneqq Q_\kappa(\beta,\varepsilon,\gamma^{(\text{opt})}_{\kappa}(\beta))|_{g=1}-Q_\kappa(\beta,\varepsilon,\gamma_\kappa(\beta))|_{g=1}$, a Taylor expansion of $Q_\kappa(\beta,\varepsilon,\gamma)|_{g=1}$ at $\gamma=\gamma_{\kappa}(\beta)$ shows that, up to first order in $\Delta\gamma=\gamma^{(\text{opt})}_{\kappa}(\beta)-\gamma_{\kappa}(\beta)$, we have
\begin{equation}
 \frac{\Delta Q}{Q_\kappa(\beta,\varepsilon,\gamma_\kappa(\beta))|_{g=1}}\simeq\frac{\Delta\gamma}{\ln\bigg(\frac{8e^{2\gamma_\kappa(\beta)}}{\|F_\beta(H)\ket{\Psi}\|\varepsilon}\bigg)},
\end{equation}
which is much smaller than $1$ provided that $\varepsilon'\ll 1$, since $\gamma_\kappa(\beta)<1$ and, therefore, $\Delta\gamma<1$.

%%%%%%%%%%%%%%%%%%%%%%%%%%%%%%%%%%%%%%%%%%%%%%%%%%%%%%%%%%%%%%%%%%%%

\section{Proof of Theorem \ref{teo_correct_master}}
\label{app:theo_fragmented}
Let us first prove a convenient auxiliary lemma. Ideally, the input and output states of the $l$-th fragment of should be $\ket{\Psi_{l-1}}=\frac{F_{\beta_{l-1}}(H)\ket{\Psi}}{\|F_{\beta_{l-1}}(H)\ket{\Psi}\|}$ and $\ket{\Psi_{l}}=\frac{F_{\Delta\beta_{l}}(H)\ket{\Psi_{l-1}}}{\|F_{\Delta\beta_{l}}(H)\ket{\Psi_{l-1}}\|}=\frac{F_{\beta_{l}}(H)\ket{\Psi}}{\|F_{\beta_{l}}(H)\ket{\Psi}\|}$, respectively. However, the actual operator that each $P_{\Delta\beta_l,\varepsilon^{\prime}_l,\alpha_l}$ perfectly block-encodes is 
\begin{equation}
\alpha_l\tilde{F}_{\Delta\beta_{l}}(H)\coloneqq\braket{0|\ U_{F_{\Delta\beta_{l}}(H)}}{0}=\alpha\,F_{\Delta\beta_{l}}(H)+E_l,
\end{equation}
for some $E_l$ on $\mathbb H_\mathcal{S}$ with $||E_l||=:\varepsilon_l'$. Hence, in analogy to Eq.\ \eqref{eq:explicit_nonperfect_blockenc}, the actual input and output states are $\ket{\tilde{\Psi}_{l-1}}=\ket{{\Psi}_{l-1}}+\ket{\Xi_{l-1}}$ and $\ket{\tilde{\Psi}_{l}}=\ket{{\Psi}_{l}}+\ket{\Xi_{l}}$, for some ``error states'' $\ket{\Xi_{l-1}}$ and $\ket{\Xi_{l}}$, respectively. (Note that $\ket{\Xi_{0}}=0$.)
% with $\|\ket{\Xi_{l-1}}\|\leq\varepsilon_{l-1}$ and $\|\ket{\Xi_{l}}\|\leq\varepsilon_{l}$, respectively, where $\varepsilon_{l-1}$ and $\varepsilon_{l}$ are the corresponding state errors.
The following lemma controls the output-state error in terms of the errors in the input state and block encoding.

\begin{lem}(Error propagation from input state and block encoding to output state)
\label{lem:error_prop_with_input_state}
For all $l\in[r]$, let $\varepsilon_{l-1}\coloneqq\|\ket{\Xi_{l-1}}\|$ and $\varepsilon_l'$ be the input-state and block-encoding errors, respectively; and $\varepsilon_l$ the tolerated output-state error. If $\varepsilon_{l-1}+\varepsilon_{l}'\leq \varepsilon_{l}\,\|\alpha_l\, F_{\Delta\beta_{l}}(H)\ket{\Psi_{l-1}}\|/2$ and $\varepsilon_{l}\ll1$, then $\|\ket{{\Xi}_{l}}\|=\mathcal{O}(\varepsilon_l)$.
\end{lem}
\begin{proof}
Similar to the proof of Lemma \ref{lem:error_prop} but where also the input state is approximate.
\end{proof}

Note that the proof of Lemma \ref{lem:error_prop} (and therefore also that of Lemma \ref{lem:error_prop_with_input_state}) makes no use of the specific form of $F_\beta$ (or $F_{\Delta\beta_{l}}$). Consequently, both lemmas hold for $(\alpha,\varepsilon')$-block-encodings of any operator function, not just  the QITE propagator. We are now in a good position to prove Theorem \ref{teo_correct_master}.

\begin{proof}[Proof of Theorem \ref{teo_correct_master}]
We begin by proving soundness. The proof consists of showing that Eqs. \eqref{eq:F_errors} imply that Lemma \ref{lem:error_prop_with_input_state} holds for all $l\in[r]$ and gives $\varepsilon_r\coloneqq\|\ket{{\Xi}_{r}}\|=\mathcal{O}(\varepsilon)$. To see this, apply the lemma  from $l=1$ till $l=r$, with $\varepsilon_0=0$, and iteratively use the property 
\begin{eqnarray}
\nonumber
p_{\Psi_{l-1}}(\Delta\beta_{l})&=&\|F_{\Delta\beta_{l}}(H)\ket{\Psi_{l-1}}\|\\
\nonumber
&=&\frac{\|F_{\beta_{l}}(H)\ket{\Psi}\|}{\|F_{\beta_{l-1}}(H)\ket{\Psi}\|}\\
\label{eq:iterative_property}
&=&\frac{p_{\Psi}(\beta_{l})}{p_{\Psi}(\beta_{l-1})}.
\end{eqnarray}

We next prove complexity. To this end, we must show the validity of Eq.\ \eqref{eq:F_overall_query}. By Def. \ref{def:masters}, the overall query complexity of Algorithm \ref{alg:Fragmented-QuITE-algorithm} is the sum over the query complexities of each primitive applied. Each primitive $P_{\Delta\beta_l,\varepsilon^{\prime}_l,\alpha_l}$ has complexity $q(\Delta\beta_l,\varepsilon^{\prime}_l, \alpha_l)$. Hence, we must show that the average number of times $n_l$ that each $P_{\Delta\beta_l,\varepsilon^{\prime}_l,\alpha_l}$ is run is given by $\frac{p_{\Psi}(\beta_{l-1})}{p_{\Psi}(\beta)\prod_{k=l}^r\alpha^2_k}$. To see this, note that, by definition, it is $n_r\coloneqq\big({\alpha_r}^2\, p_{\Psi_{r-1}}(\Delta\beta_{r})\big)^{-1}$, $n_{r-1}\coloneqq n_r\times\big({\alpha_{r-1}}^2\, p_{\Psi_{r-2}}(\Delta\beta_{r-1})\big)^{-1}$, $\hdots$, and  $n_{1}\coloneqq n_2\times\big({\alpha_{1}}^2\, p_{\Psi}(\Delta\beta_{1})\big)^{-1}$. Then, use the latter together with Eq.\ \eqref{eq:iterative_property} to get $n_l=\frac{p_{\Psi}(\beta_{l-1})}{p_{\Psi}(\beta)\prod_{k=l}^r\alpha^2_k}$ for all $l\in[r]$.
\end{proof}
Interestingly, we note that the only specific detail about $F_\beta$ that the proof of Theorem \ref{teo_correct_master} uses is the fact that $F_\beta(H)=\prod_{l=1}^{r} F_{\Delta\beta_l}(H)$ (this is required for Eq.\ \eqref{eq:iterative_property} to hold). Consequently, Theorem \ref{teo_correct_master} applies not only to fragmented QITE but actually also to the fragmentation of any operator function in terms of suitable factors.

%%%%%%%%%%%%%%%%%%%%%%%%%%%%%%%%%%%%%%%%%%%%%%%%%%%%%%%%
\section{Validity of the query complexity bounds of QITE {primitives} for low beta}
\label{app:P1query}
 
As discussed in Secs. \ref{sec:master_alg} and \ref{sec:GS_sampling}, the first steps of fragmented QITE involve small inverse temperatures, often with $\Delta\beta_1<1$. However, Eqs. \eqref{eq:BEquery} and \eqref{eq:real_time_or} are in principle only asymptotic upper bounds for the actual query complexities. Hence, it is licit to question how valid they are to access the complexity of the first steps of Alg. \ref{alg:Fragmented-QuITE-algorithm}. In this appendix, we discuss the validity and tightness  for all $\beta>0$ of 
 \begin{equation}\label{eq:tildeq1}
  \tilde{q}_1(\beta,\varepsilon^{\prime})\coloneqq2\left(\frac{e\beta}{2}+\frac{\ln\left(1/\varepsilon'\right)}{\ln\left(e+2\ln(1/\varepsilon')/e\beta\right)}\right)
 \end{equation}
 and 
 \begin{equation}
  \tilde{q}_2(\beta,\varepsilon^{\prime})\coloneqq4\left(\beta/\gamma+1\right)\ln(4/\varepsilon')
 \end{equation}
(without multiplicative factors implied by the big-$\mathcal{O}$ notation) as exact expressions for the query complexities of primitives 1 and 2, respectively. These  formulas are the ones used in our numeric experiments. First of all, notice that although these  expressions are  continuous
functions, the actual query complexities used by Algs. \ref{alg:fBlockEncode} and \ref{alg:frealtime} take only even integer values. Thus, we take the value of the query
for each round of fragmentation as $2\lceil\tilde{q}_k(\beta,\varepsilon^{\prime})/2\rceil$, $k=1,2$.
In particular, for any value of $\beta$ such that $0<\tilde{q}_1(\beta,{\varepsilon')}\leq2$ or $0<\tilde{q}_2(\beta,{\varepsilon')}\leq2$
at least two queries to the corresponding oracle is used.

For $P_2$ there is actually no concern because the reasoning of Sec.\ \ref{sec:QSP_RTE} that led to Eq.\ \eqref{eq:real_time_or} is valid for any
beta. This is because $q_2=2\lceil\tilde{q}_2(\beta,\varepsilon^{\prime})/2\rceil$  is the exact expression furnished by Lemma 37 in Ref.\ \cite{vanApeldoorn2020quantumsdpsolvers}, used to prove Lemma\ \ref{rt_function}, for $\alpha F_\beta(\lambda)$. Therefore, although this formula is not strictly tight, because it comes from successive approximations \cite{vanApeldoorn2020quantumsdpsolvers}, it is exact and valid for any $\beta>0$ in the sense that making exactly $2\lceil\tilde{q}_2(\beta,\varepsilon^{\prime})/2\rceil$ queries to the oracle ensures that the error is below the tolerated. {Consequently, the value $2\lceil\tilde{q}_2(\beta,\varepsilon^{\prime})/2\rceil$ is an upper bound for the necessary number of queries without any multiplicative factor that could be implied by the big-$\mathcal{O}$ notation.}

$P_1$, on the other hand, requires a further analysis since, according to Eq.\ \eqref{eq:BEquery}, ${q}_1(\beta,\varepsilon^{\prime})$ is known to be equal to $\tilde{q}_1(\beta,\varepsilon^{\prime})$ only in big-$\mathcal{O}$ notation and, particularly, for large $\beta$.  Next, we
%briefly review how Eq.\ \eqref{eq:BEquery} was obtained and 
numerically show that $\tilde{q}_1(\beta,{\varepsilon')}$ is, in fact,  an over-estimation of the actual $q_1(\beta,\varepsilon')$ needed to guarantee error $\varepsilon'$.

 We notice that the first inequality in Eq.\ \eqref{eq:error} gives a tighter upper bound for the query complexity. Therefore, to attain a target error $\varepsilon'$, it is enough that   $q_1= 2\left\lceil\tilde{q}_1(\beta,\varepsilon')/2\right\rceil$ queries yields to a truncation error which satisfies
\begin{equation}
\varepsilon_{\rm tr}\leq\varepsilon_{\rm up}^{({q}_1)}(\beta)=\frac{\beta^{{q}_1/2+1}}{2^{{q}_1/2}({q}_1/2+1)!}<\varepsilon'.\label{eq:firstbound}
\end{equation}
 As we explain next, Fig.\ \ref{fig:epsilonbounds} shows that it is, in fact, what  happens.  In Fig.\ \ref{fig:epsilonbounds}a) we show the interval of
$\beta$'s for which the value $\tilde{q}_1(\beta,\varepsilon')$
is between   $q_1=18$ and $q_1=20$. For any $\beta$ inside this interval, the number of queries made in $P_1$ will be $q_1=20$. However, in Fig.\ \ref{fig:epsilonbounds}b) when we look at the upper bound for the truncation error with $20$ queries, $\varepsilon_{\rm up}^{({q}_1=20)}$, in the same interval, we see that it is much smaller than the tolerated error $\varepsilon'$. Consequently, the truncation error $\varepsilon_{\rm tr}$, which is the actual error of $P_1$, is also below $\varepsilon'$. This is also observed in Figs.\ \ref{fig:epsilonbounds}c) and \ref{fig:epsilonbounds}d), where the interval of $\beta$'s for which the minimum number of $q_1=2$ queries is made and the corresponding $\varepsilon_{\rm tr}^{({q}_1=2)}$ in that interval are shown, respectively.   We tested for other values of $\varepsilon'$ ranging from $10^{-1}$ to $10^{-10}$ and the same behavior is observed. 
This shows that $\tilde{q}_1(\beta,\varepsilon')$
is valid as the query complexity of $P_1$ even for small values of $\beta$. Moreover, it actually overestimates the minimum query complexity for a given target error $\varepsilon'$.

\begin{figure}[t]
\begin{centering}
\includegraphics[width=0.8\columnwidth]{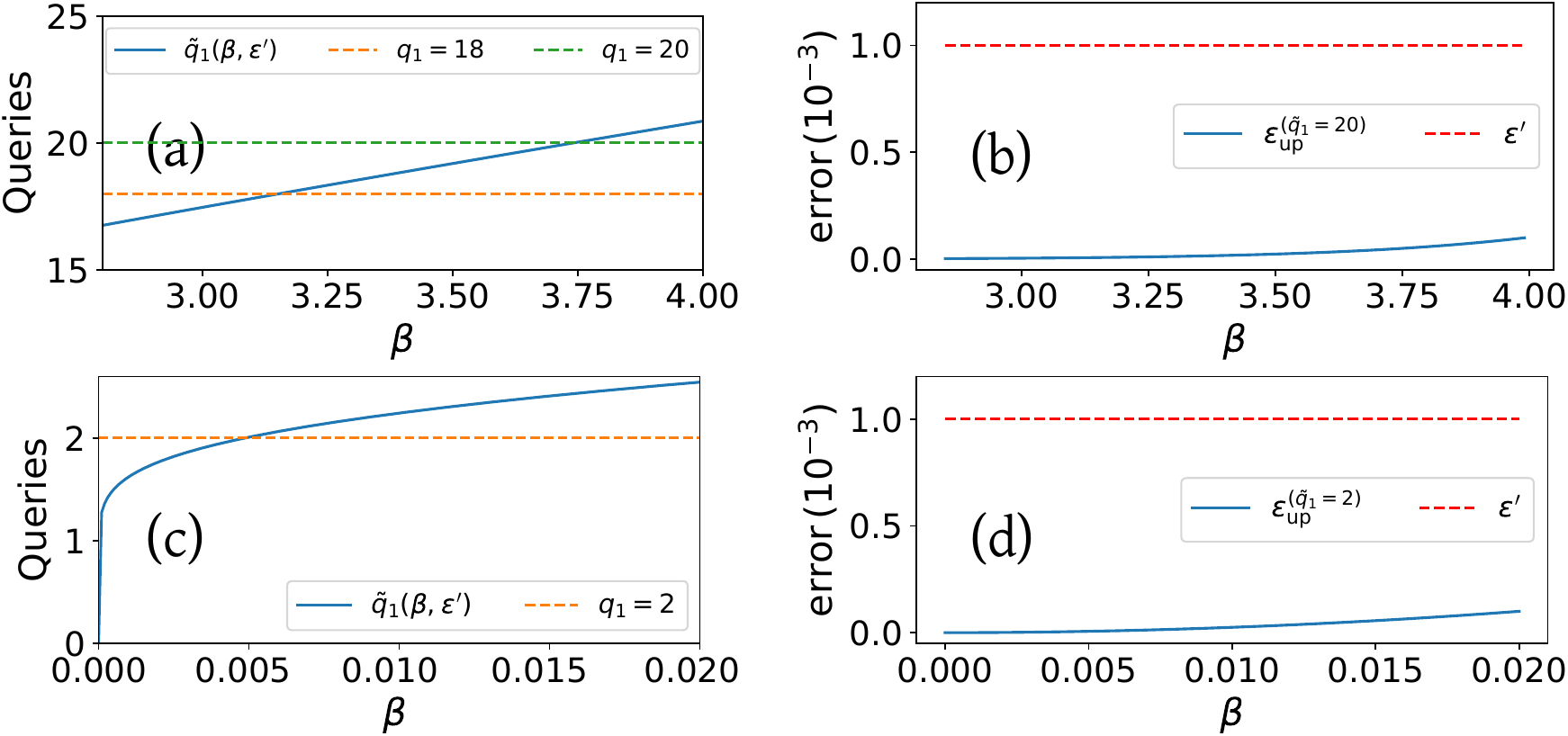}
\caption{\label{fig:epsilonbounds} (a) Zoom in on the query complexity function $\tilde{q}_1(\beta,\varepsilon')$ for fixed error $\varepsilon'=10^{-3}$. For any value of inverse temperature $\beta$ for which $18\leq\tilde{q}_1(\beta,\varepsilon')\leq 20$, $q_1=20$ queries are made in $P_1$.  (b) the corresponding tighter error bound $\varepsilon_{up}^{(q_1=20)}(\beta)$ from Eq.\ \eqref{eq:firstbound} for $q_1=20$ queries in the same interval of $\beta$, which is way below the tolerated error $\varepsilon'$. Parts (c) and (d) show the same thing for the minimum number of queries $q_1=2$, showing that $\varepsilon_{up}^{(q_1=2)}(\beta)$ is much smaller than $\varepsilon'$ for any $\beta$ such that $\tilde{q}_1(\beta,\varepsilon')\leq 2$.  {Therefore (see text), $2\lceil\tilde{q}_1(\beta,\varepsilon^{\prime})/2\rceil$ is a valid upperbound for the number of queries of P$_1$ even when $\beta$ is small.} }
\end{centering}
\end{figure}

%%%%%%%%%%%%%%%%%%%%%%%%%%%%%%%%%%%%%%%%%%%%%%%%%%%%%%%%%%%%%%%%%%%%%%%
%%%%%%%%%%%%%%%%%%%%%%%%%%%%%%%%%%%%%%%%%%%%%%%%%%%%%%%%%%%%%%%%%%%%%%%
\section{Fragmented QITE outperforms coherent QITE}
\label{app:analytical_evidence}
{Here we analytically show, for primitive $P_1$, that there exists an inverse temperature $\beta_c$ above which fragmented QITE outperforms coherent  QITE (based on quantum amplitude amplification) in terms of overall query complexity. Before the proof, it is useful to recall that the success probability (of any QITE primitive) is given by $p_{\Psi}(\beta)=\sum_\lambda |\braket{\lambda}{\Psi}|^2 e^{-2\beta(\lambda-\lambda_{\text{min}})}$. Note also that $p_{\Psi}(\beta)\geq |\braket{\lambda_{\text{min}}}{\Psi}|^2={o}^{2}$. Hence, 
\begin{equation}\label{eq:psuc}
 \sqrt{p_\Psi(\beta)}\geq{{o}},
\end{equation}
for all $\beta\geq 0$. Moreover,  $p_{\Psi}:[0,\infty)\rightarrow[{o}^{2},1]$ is a monotonically decreasing function. Therefore, it has an inverse function, which we denote as $p_{\Psi}^{-1}$.

\begin{proof}[Proof of Theorem \ref{thm:beta_c}]
{We assume here that Primitive $1$ is run using a number of queries equals to the upper-bound that guarantees a given error. A multiplicative factor would evenly affect the complexities of all master algorithms. For simplicity, we omit here the multiplicative factor (shown to be equal to $8$ in the proof of Theorem \ref{teoQuITEbe}) and take the complexity of $P_1$ directly as ${q}_1(\beta,\varepsilon^{\prime})=\frac{e\beta}{2}+\frac{\ln\left(1/\varepsilon'\right)}{\ln\left(e+2\ln(1/\varepsilon')/e\beta\right)}$.   This is also justified in the tightness analysis of App.\ \ref{app:P1query}. }

We need to show that there exists a schedule $S_2=\{\Delta\beta_1,\Delta\beta_2\}$ such that the average overall query complexities satisfy $Q_{S_2}(\beta,\varepsilon)\leq Q_{\text{coh}}(\beta,\varepsilon)$. According to Eqs. \eqref{eq:Qcomp} and \eqref{eq:F_overall_query}, this holds if
 \begin{equation}\label{eq:fragxcoh}
  \sum_{l=1}^{2} \frac{p_{\Psi}(\beta_{l-1})}{\left(p_{\Psi}(\beta)\right)^{1/2}}\, {q}_1(\Delta\beta_l,\varepsilon^{\prime}_l)\leq {q}_1(\beta,\varepsilon^{\prime}),
 \end{equation}
where $\varepsilon^{\prime}_l$  is given by Eq.\,\eqref{eq:F_errors}. 
It is useful to introduce the upper bound $\tilde{{q}}_1(\beta,\varepsilon')\coloneqq\frac{e\beta}{2}+\ln\left(\frac{1}{\varepsilon'}\right)$. (Note that $\tilde{{q}}_1(\beta,\varepsilon')>{q}_1(\beta,\varepsilon^{\prime})$ for all $\beta>0$.) Then, Eq. \eqref{eq:fragxcoh} holds if  
\begin{equation}\label{eq:fragxcoh2}
  \sum_{l=1}^{2} \frac{p_{\Psi}(\beta_{l-1})}{\left(p_{\Psi}(\beta)\right)^{1/2}}\, \tilde{q}_1(\Delta\beta_l,\varepsilon^{\prime}_l)\leq {q}_1(\beta,\varepsilon^{\prime}).
 \end{equation}
Next, we break this inequality into two inequalities, one for $\Delta\beta_1$ and another for $\Delta\beta_2$. Then, we  show that, under the theorem's assumptions, both inequalities can be satisfied.
 
More precisely, Eq.~\eqref{eq:fragxcoh2} is satisfied if the following two inequalities are simultaneously satisfied: 
\begin{equation}\label{eq:fragbetainv}
\frac{p_{\Psi}(\beta_{l-1})}{\left(p_{\Psi}(\beta)\right)^{1/2}}\, \tilde{q}_1(\Delta\beta_l,\varepsilon^{\prime}_l)\leq \frac{1}{2}{{q}_1}(\beta,\varepsilon^{\prime}),\, \text{ for } l=1,2 \ ;
 \end{equation}
and we construct an $S_2$ that fulfills this. First, substituting for $\tilde{q}_1$ {and ${q}_1$}, we find that each fragment must satisfy 
\begin{eqnarray}
\label{eq:deltabetal}
\nonumber
\Delta\beta_l&\leq& \frac{1}{2}\frac{\sqrt{p_{\Psi}(\beta)}}{p_{\Psi}(\beta_{l-1})}\Bigg(\beta+\frac{2}{e}{\frac{\ln\big(\frac{1}{\varepsilon'}\big)}{\ln\big(e+2\ln(1/\varepsilon')/e\beta\big)}}\Bigg)\\
&-&\frac{2}{e}\ln\left(\frac{1}{\varepsilon'_l}\right),\, \text{ for } l=1,2.
\end{eqnarray}
For consistency, the right-hand side of Eq.\ \eqref{eq:deltabetal} should be positive for each $l$, the most critical case being that of $l=1$, when the first term is less positive while the second term is more negative as compared with the case $l=2$. So it suffices to enforce positivity of the right-hand side  of Eq.\ \eqref{eq:deltabetal} for $\beta\geq\beta_c$ only for $l=1$. One can directly see that this is already satisfied by plugging the expression for $\beta_c$ into Eq.\ \eqref{eq:deltabetal} and using Eq.\ \eqref{eq:psuc}.

Because $\beta=\Delta\beta_1+\Delta\beta_2$, we are free to choose the size of only one fragment, say $\Delta\beta_1\equiv\beta_1$. Eq.\ \eqref{eq:deltabetal} imposes an upper on $\beta_1$ for $l=1$ and a lower bound for $l=2$. We first find a $\beta_1$ that satisfies the lower bound and then show that, for $\beta\geq\beta_c$, the upper bound is automatically satisfied. Using Eq.\ \eqref{eq:F_errors} and $\Delta\beta_2=\beta-\beta_1$, we re-write Eq.\ \eqref{eq:deltabetal} for $l=2$  as
\begin{eqnarray}
\label{eq:frag2}
\nonumber
 0&\leq&\beta_1+\bigg[\frac{\sqrt{p_\Psi(\beta)}}{2\,p_\Psi(\beta_{1})}-1\bigg]\beta
 \\ 
 \nonumber
&+&\frac{2}{e}\bigg[\frac{\sqrt{p_\Psi(\beta)}}{2\,p_\Psi(\beta_{1}){\,\ln{\big[}e+2\ln(1/\varepsilon')/e\beta{\big]}}}-1\bigg]\ln\bigg(\frac{1}{\varepsilon'}\bigg)\\
 &+&\frac{2}{e}\ln\bigg(\frac{1}{2\sqrt{p_\Psi(\beta_{1}})}\bigg).
\end{eqnarray}
Sufficient to satisfy this inequality is that each term on the right-hand side is positive. Since $\ln{\big[}e+2\ln(1/\varepsilon')/e\beta{\big]}>1$, the third term is always smaller than the second one. Therefore, requiring that 
\begin{equation}
\label{eq:prelim_condition}
\sqrt{p_\Psi(\beta)}\geq 2\,p_\Psi(\beta_{1})\ln{\big[}e+2\ln(1/\varepsilon')/e\beta{\big]}
\end{equation}
 ensures that both the second and the third terms are positive. This condition is in turn satisfied if $p_\Psi(\beta_{1})\leq \frac{\sqrt{p_\Psi(\beta)}}{2\,\ln{\big[}e+2\ln(1/\varepsilon')/e\beta{\big]}} $, which is equivalent  to demanding 
\begin{equation}
\label{eq:inv_condition}
\beta_1\geq p^{-1}_\Psi\big(\sqrt{p_\Psi(\beta)}/2\,{\ln{\big[}e+2\ln(1/\varepsilon')/e\beta{\big]}}\big).
\end{equation}
Notice that Eq. \eqref{eq:prelim_condition} also ensures that $\sqrt{p_\Psi(\beta_{1})}\leq1/2$ for $\beta\geq\beta_c$, due to $p_\Psi(\beta)\leq1/4$ by theorem assumption. This makes the last term in Eq.\ \eqref{eq:frag2} also positive. Now, because of Eq.\ \eqref{eq:psuc}, sufficient to satisfy Eq. \eqref{eq:inv_condition} is to demand that
  \begin{equation}\label{eq:beta1}
   \beta_1\geq p^{-1}_\Psi\left(\frac{{o}}{2}\frac{1}{{\ln[e+2\ln(2/{o}\,\varepsilon)/e\beta]}}\right).
  \end{equation}
This is our final lower bound on $\beta_1$. Its fulfillment guarantees Eq. \eqref{eq:inv_condition} and, therefore, also Eq.\ \eqref{eq:frag2}.
  
On the other hand, for $l=1$, Eq.\ \eqref{eq:deltabetal} can be re-written as
\begin{eqnarray}
\label{eq:fragment1}
\nonumber
 \beta_1&\leq &\frac{\sqrt{p_{\Psi}(\beta)}}{2}\beta-\frac{2}{e}\ln\left(4\right)\\
  &+&\frac{2}{e}\bigg(\frac{\sqrt{p_{\Psi}(\beta)}}{2\,{\ln{\big[}e+2\ln(1/\varepsilon')/e\beta{\big]}}}-1\bigg)\ln\Big(\frac{1}{\varepsilon'}\Big),
\end{eqnarray}
where we have used Eq.\ \eqref{eq:F_errors} again and $p_{\Psi}(\beta_0)=p_{\Psi}(0)=1$.
For $\beta_1$ satisfying Eq.\ \eqref{eq:beta1}, Eq. \eqref{eq:fragment1} is satisfied if 
\begin{eqnarray}
\label{eq:beta_bound}
\nonumber
\beta&\geq&\frac{2}{\sqrt{p_\Psi(\beta})} \bigg[{p^{-1}_\Psi\left(\frac{{o}}{2}\frac{1}{{\ln[e+2\ln(2/{o}\,\varepsilon)/e\beta]}}\right)}\\
 &+&\frac{2}{e}\ln\left(\frac{4}{\varepsilon'}\right)\bigg]-\frac{\ln\big(\frac{1}{\varepsilon'}\big)}{e\,{\ln{\big[}e+2\ln(1/\varepsilon')/e\beta{\big]}}}, 
\end{eqnarray}
which, in turn, by virtue of Eq.\ \eqref{eq:psuc}, is satisfied if
\begin{equation}\label{eq:lowerbeta}
 \beta\geq \frac{2}{{o}}\Big[p^{-1}_\Psi\Big(\frac{{o}}{2}{\frac{1}{{\ln[e+2\ln(2/{o}\,\varepsilon)/e\beta]}}}\Big)+\frac{2}{e}\ln\Big(\frac{8}{{o}\,\varepsilon}\Big)\Big].
\end{equation}

The RHS of Eq.\ \eqref{eq:lowerbeta} still depends on $\beta$. To remove this dependence, we use a $\beta$-independent upper bound for the logarithmic term. First, notice that, for Eq.\ \eqref{eq:lowerbeta} to hold, $\beta>\frac{4}{e\,{o}}\ln\left(\frac{8}{{o}\,\varepsilon}\right)$ must necessarily hold. This, in turn, is equivalent to $\frac{2\ln(2/{o}\,\varepsilon)}{e\beta}<\frac{{o}\ln(2/{o}\,\varepsilon)}{2\ln(8/{o}\,\varepsilon)}$; and, since %$2/{o}\,\varepsilon>e$ and 
$\frac{{o}\ln(2/{o}\,\varepsilon)}{2\ln(8/{o}\,\varepsilon)}<\frac{{o}}{2}\leq 1/4$, one gets that  %$1<\ln[e+2\ln(2/{o}\,\varepsilon)/e\beta]<\ln[e+1/4]<1.1$. 
$\ln[e+2\ln(2/{o}\,\varepsilon)/e\beta]<\ln[e+1/4]<1.1$. 
Consequently, $\frac{{o}}{2}{\frac{1}{{\ln[e+2\ln(2/{o}\,\varepsilon)/e\beta]}}}>{o}/2.2$;  and, since $p^{-1}_\Psi$ is a monotonically decreasing function, $p^{-1}_\Psi\big(\frac{{o}}{2}{\frac{1}{{\ln[e+2\ln(2/{o}\,\varepsilon)/e\beta]}}}\big){<}p^{-1}_\Psi({o}/2.2)$. That latter is the desired $\beta$-independent upper bound. Here, we note that the theorem assumption ${o}<1/2.2$ ensures that $p^{-1}_\Psi({o}/2.2)$ is well defined, because it guarantees that ${o}/2.2>{o}^2=p^{-1}_\Psi(\beta\rightarrow \infty)$. Then, finally, taking $\beta_c=\frac{2}{{o}}\left[p^{-1}_\Psi\left(\frac{{o}}{2.2}\right)+\frac{2}{e}\ln\left(\frac{8}{{o}\,\varepsilon}\right)\right]$ and demanding that $\beta>\beta_c$ is sufficient to satisfy Eq.  \eqref{eq:lowerbeta}. This is our final lower bound on $\beta$.   

In conclusion, the schedule $S_2=\{\beta_1,\beta-\beta_1\}$, with $\beta_1$ in the range determined by Eqs.\ \eqref{eq:beta1} and \eqref{eq:fragment1}, and $\beta\geq\beta_c$, satisfies Eq.\ \eqref{eq:fragbetainv} for both $l=1$ and $l=2$. In particular, in the theorem statement, $\beta_1$ is taken precisely as the RHS of Eq.\ \eqref{eq:beta1}. This proves the last remaining implication.
\end{proof}

%%%%%%%%%%%%%%%%%%%%%%%%%%%%%%%%%%%%%%%%%%%%%%%%%%%%%%%%%%%%%%%%%%%%%%%%%%%
%%%%%%%%%%%%%%%%%%%%%%%%%%%%%%%%%%%%%%%%%%%%%%%%%%%%%%%%%%%%%%%%%%%%%%%%%%%
\section{Dependence on the success probability: Fragmented QITE versus coherent master algorithm }
\label{fragxcohPsuc}

Theorem \ref{thm:beta_c} establishes that it is possible to find an inverse temperature and a  fragmentation schedule for which fragmented QITE outperforms the coherent master QITE algorithm. Nevertheless, it says nothing about scaling advantages of fragmentation. One of the reasons for that is the intricate dependence of the query complexity in Theorem \ref{teo_correct_master}  on all the parameters  ($\beta$, $\epsilon$, $p_\Psi(\beta)$) not allowing for a clear claim of advantage. More importantly, the complexity depends on the particular fragmentation schedule, which in turn depends on the particular Hamiltonian instance.

However, we can numerically compare the advantage got from amplitude amplification with the results for fragmentation. Compared to repeat until success, amplitude amplification gives a quadratic advantage relatively to the probability of success. That is, while the total query complexity of the probabilistic algorithms is given as $q(\beta,\varepsilon')/p_\Psi(\beta)$, the coherent algorithm has total query complexity of $q(\beta,\varepsilon')/\sqrt{p_\Psi(\beta)}$. In an attempt to isolate the dependence on the success probability, in Fig. \ref{fig:compare_Psuc} we show the ratios $Q_{S_r}(\beta,\varepsilon)/q(\beta,\varepsilon')$ and $Q_{\text{coh}}(\beta,\varepsilon)/q(\beta,\varepsilon')$. The plots show results for different Hamiltonian instances and different values of $\beta$ larger than the observed average $\beta_c$. In this way, we observe the behavior of fragmented QITE after the point for which we get advantage over amplitude amplification. We observe the same scaling for the two master algorithms, but  the  fragmented algorithm presents a better multiplicative factor. Therefore, from this empirical observation we cannot claim any supra-square advantage of fragmentation over RUS. Nevertheless, we reinforce that the  advantage obtained in total query complexity comes without a large circuit depth overhead.

 %%%%%%%%%%%%%%%%%%%%%%%%%%%%%%%%%%%%%%%%%%%%%%%%%%%%%%%%%%%%%%%%%%%%%%%
\begin{figure}[t!]
\centering
  \includegraphics[width=0.5\columnwidth]{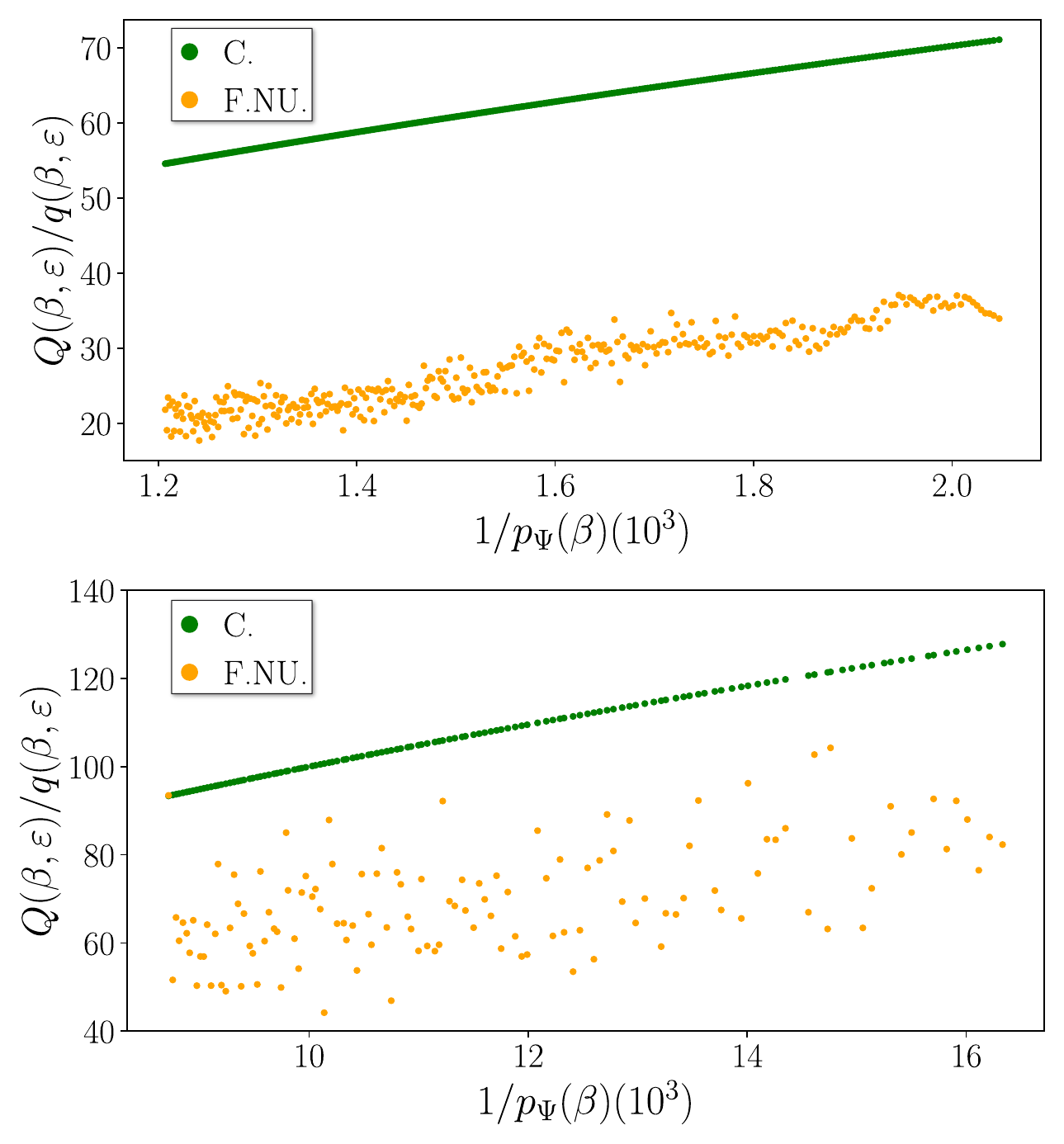}
  \caption{{\bf Dependence of the total query complexities on the probability of success}. Ratio between the total success probability and the query complexity of  one run of the probabilistic master algorithm as a function of the inverse of the success probability for two Hamiltonian models. In a) the result for instances of weighted MaxCut for 12 qubits and in b) the result for intances of the Heisenberg fully-connected model with 14 qubits. Each point in the graphs corresponds to a different Hamiltonian instance and one value of inverse temperature above $\beta_c$.  
  \label{fig:compare_Psuc}} 
\end{figure}
%%%%%%%%%%%%%%%%%%%%%%%%%%%%%%%%%%%%%%%%%%%%%%%%%%%%%%%%%%%%%%%%%%%%%%%

%%%%%%%%%%%%%%%%%%%%%%%%%%%%%%%%%%%%%%%%%%%%%%%%%%%%%%%%%%%%%%%%%%%%%%%
%%%%%%%%%%%%%%%%%%%%%%%%%%%%%%%%%%%%%%%%%%%%%%%%%%%%%%%%%%%%%%%%%%%%%%%
\section{Close-to-optimality of Primitive 1 for non-interacting Hamiltonians}
\label{app:non_interacting}
Let us consider {the simple case of a non-interacting qubit Hamiltonian} given by
\begin{equation}
\label{eq:non-interactive}
 H=\sum_{j}b_j Z_j,
\end{equation}
with $b_j\geq0$ and $\sum_{j=1}^N b_j=1$,} such that $\|H\|=1$. {Clearly, the analysis for this case covers also any other Hamiltonian unitarily-equivalent to Eq. \eqref{eq:non-interactive}. We focus on the task of Gibbs-state sampling (maximally mixed state as input). Hence, ${o}=2^{-N/2}$ and the success probability for $P_1$} is given by $ p_\Psi(\beta)=e^{-2\,\beta}\prod_{j=1}^N\cosh(2\,\beta\,b_j)$.
For simplicity, we {take} $b_j=1/N$ for all $j$, obtaining 
\begin{equation}
 p_\Psi(\beta)=e^{-2\,\beta}\cosh(2\,\beta/N)^N.
\end{equation}
{This is simple enough to obtain a closed-form expression for its inverse function and -- so -- explicit expressions for the fragmentation schedule of Lemma  \ref{thm:beta_c}, which we do next.

By virtue of Lemma  \ref{thm:beta_c}, $\Delta\beta_1={p^{-1}_\Psi\big(\frac{{o}}{2}\frac{1}{{\ln[e+2\ln(2/{o}\,\varepsilon)/e\beta]}}\big)}$. Hence, using that $1<{\ln[e+2\ln(2/{o}\,\varepsilon)/e\beta]}<1.1$, we obtain $p^{-1}_\Psi\big(\frac{{o}}{2}\big)=\frac{N}{4}\ln\big(\frac{2^{1/N}}{2^{1/2}-2^{1/N}}\big)<\Delta\beta_1<p^{-1}_\Psi\big(\frac{{o}}{2.2}\big)=\frac{N}{4}\ln\big(\frac{(2.2)^{1/N}}{2^{1/2}-(2.2)^{1/N}}\big)$. This, using that, for $N\geq3$, $\ln\big(\frac{(2.2)^{1/N}}{2^{1/2}-(2.2)^{1/N}}\big)\leq 2.44$ and $\ln\big(\frac{2^{1/N}}{2^{1/2}-2^{1/N}}\big)>0.88$, gives 
\begin{equation}
\label{eq:Delta_beta_1}
0.88\,\frac{N}{4}\leq\Delta\beta_1\leq 2.44\, \frac{N}{4}
\end{equation} 
and $\beta_c= 2^{N/2 +1}\big[2.44\,\frac{N}{4}+\ln\big(8\frac{2^{N/2}}{\varepsilon}\big)\big]$. This, in turn leads to
\begin{equation}
\label{eq:beta_2_non_int}
\Delta\beta_2\geq 2^{N/2 +1}\Big[2.44\,\frac{N}{4}+\ln\Big(8\frac{2^{N/2}}{\varepsilon}\Big)\Big]-2.44\, \frac{N}{4}
\end{equation} 
for all $\beta\geq \beta_c$.} The case $N=2$ is not {covered} by Lemma \ref{thm:beta_c} because {it has} ${o}= 1/2 >1/2.2$.
 
{As clear from Eqs. \eqref{eq:Delta_beta_1} and \eqref{eq:beta_2_non_int}, the first fragment is exponentially shorter in $N$ than the second one. Moreover, we show next that, the first fragment satisfies $\beta_1 < 8 \ln(4/\varepsilon_1')$. This implies that Primitive $1$ performs better (in query complexity) than the one from Ref.\ \cite{gilyen2019} at the fragmentation scheme from Lemma \ref{thm:beta_c}, as discussed after Theor.\ \ref{teoQuITEbe}. In fact, it also implies that Primitive $1$'s performance is close to optimal (see discussion after Theor. \ref{th:lowerbound}) at that first fragment.
From Eq.\ \eqref{eq:F_errors}, for $r=2$, we have $\varepsilon_1'=\varepsilon\sqrt{p_\Psi(\beta)}/8$. Now, because $\beta>p^{-1}_\Psi({o}/2.2)$, $p_\Psi(\beta)<{o}/2.2=2^{-N/2}/2.2$. This leads to 
  \begin{eqnarray}
 \nonumber
    8\ln\left(\frac{4}{\varepsilon_1{'}}\right)&=&8\ln\left(\frac{4\times8}{\varepsilon\sqrt{p_\Psi(\beta)}}\right)\\
 \nonumber
    &>&8\ln\left(\frac{4\times8}{\varepsilon\sqrt{2^{-N/2}/2.2}}\right)\\
 \nonumber
    &>&8\Big[\ln\big(2^{N/4}\big)+\ln\Big(\frac{21}{\varepsilon}\Big)\Big]\\
 \nonumber
&>&2.44\,\frac{N}{4}+8\ln\Big(\frac{21}{\varepsilon}\Big)\\
&>&\beta_1,
  \end{eqnarray}
where, in the last equation, we have used Eq. \eqref{eq:Delta_beta_1} and the fact that $\beta_1=\Delta\beta_1$. This finishes the proof.}
 %%%%%%%%%%%%%%%%%%%%%%%%%%%%%%%%%%%%%%%%%%%%%%%%%%%%%%%%%%%%%%%%%%%%%%%%%%
 
 %%%%%%%%%%%%%%%%%%%%%%%%%%%%%%%%%%%%%%%%%%%%%%%%%%%%%%%%%%%%%%%%%%%%%%%
\begin{figure}[t!]
\centering
  \includegraphics[width=0.7\columnwidth]{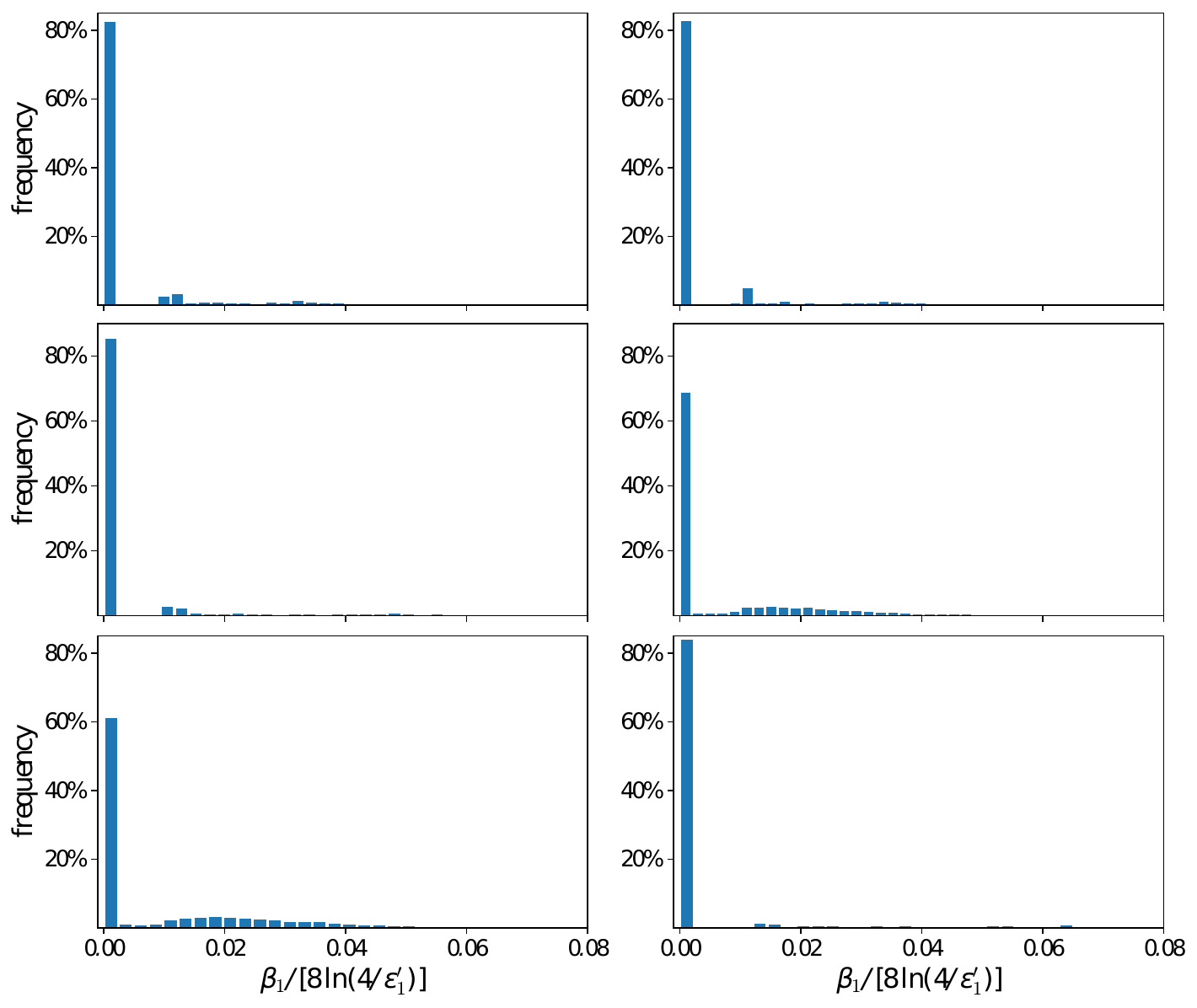}
  \caption{{{\bf Histograms of $\beta_1/8 \ln(4/\varepsilon_1')$ for the optimized schedules $S_{r,a}$ from Eq.\ \eqref{eq:schedule_ansatz}}.  The upper panels correspond to the quantum spin glass Hamiltonians with $N=14$ (left) and $N=15$ (right) qubits, the central ones to the quatum RBM with $N=4$ (left) and $N=5$ (right), and the lower ones to the weighted MaxCut with $N=12$ (left) and $N=13$ (right). Similar behaviours are observed for other values of $N$ and $\varepsilon$. For the upper and central panels we took $\varepsilon=0.001$, whereas for the lower ones $\varepsilon=0.01$. 
For each Hamiltonian instance $H$, we took several different values of $\beta$; and each pair $(H,\beta)$ constitutes a different event in the histogram. For the panels where $N\leq12$, the histograms are built from $81$ uniformly chosen values of $\beta$  from 0 to 2000 for each of the $1000$ Hamiltonian instances in each class, giving a total of 81000 events. Whereas for those where $N>12$, they are built from 21 uniformly chosen  values of $\beta$ from 0 to 10000 and $100$ Hamiltonian instances from each class, giving a total of 2100 events. The results show that $\beta_1$ is typically 1000 (and at worst 20) times smaller than $8 \ln(4/\varepsilon_1')$. This confirms that the first fragment operates deep into the regime of optimality of $P_1$ even for uniformly sampled $\beta$. If, instead of uniform values of $\beta$ between 0 and the high values mentioned above, we choose $\beta$ values close to the corresponding critical inverse temperature, then $\beta_1/8 \ln(4/\varepsilon_1')$ concentrates even more at the first column close to zero. }
  \label{fig:histograms}} 
\end{figure}
%%%%%%%%%%%%%%%%%%%%%%%%%%%%%%%%%%%%%%%%%%%%%%%%%%%%%%%%%%%%%%%%%%%%%%%

{\section{Close-to-optimality of Primitive 1 for interacting Hamiltonians}
\label{app:histograms}
Here, we study the ratio between $\Delta\beta_1=\beta_1$ and $8 \ln(4/\varepsilon_1')$ for fragmented quantum-Gibbs-state sampling with $P_1$, for the generic Hamiltonians studied in Sec. \ref{sec:GS_sampling} and the optimal schedules $S_{r,a}$ used for the central panel of Fig. \ref{fig:Opt_schedules}. Recall that $\beta_1\approx 8 \ln(4/\varepsilon_1')$ is the point at which the complexity upper bound in Eq. \eqref{eq:BEquery} starts outperforming the complexity upper bound derived in \cite{gilyen2019} for the powerful QITE primitive obtained there, as mentioned after Theo. \ref{teoQuITEbe}. The results are displayed in Fig.\ \ref{fig:histograms}, showing the histograms of $\beta_1 / 8 \ln(4/\varepsilon_1')$ for the random  instances considered for the three classes of Hamiltonian and for different $N$. This clearly shows that $\beta_1 \ll 8 \ln(4/\varepsilon_1')$, supporting our claim that, for the first fragment, the fragmented master algorithm operates deep into the optimality regime of Primitive $1$. Finally, we performed the same analysis for $\Delta\beta_2$ (not shown) and consistently observed that $\Delta\beta_2$ is smaller than $8 \ln(4/\varepsilon_2')$ too. That is, also the second fragment operates close to the optimality regime of $P_1$.}

%%%%%%%%%%%%%%%%%%%%%%%%%%%%%%%%%%%%%%%%%%%%%%%%%%%%%%%%%%%%%%%%%%%%%%%
\begin{figure*}[ht!]
\centering
  \includegraphics[width=0.95\textwidth]{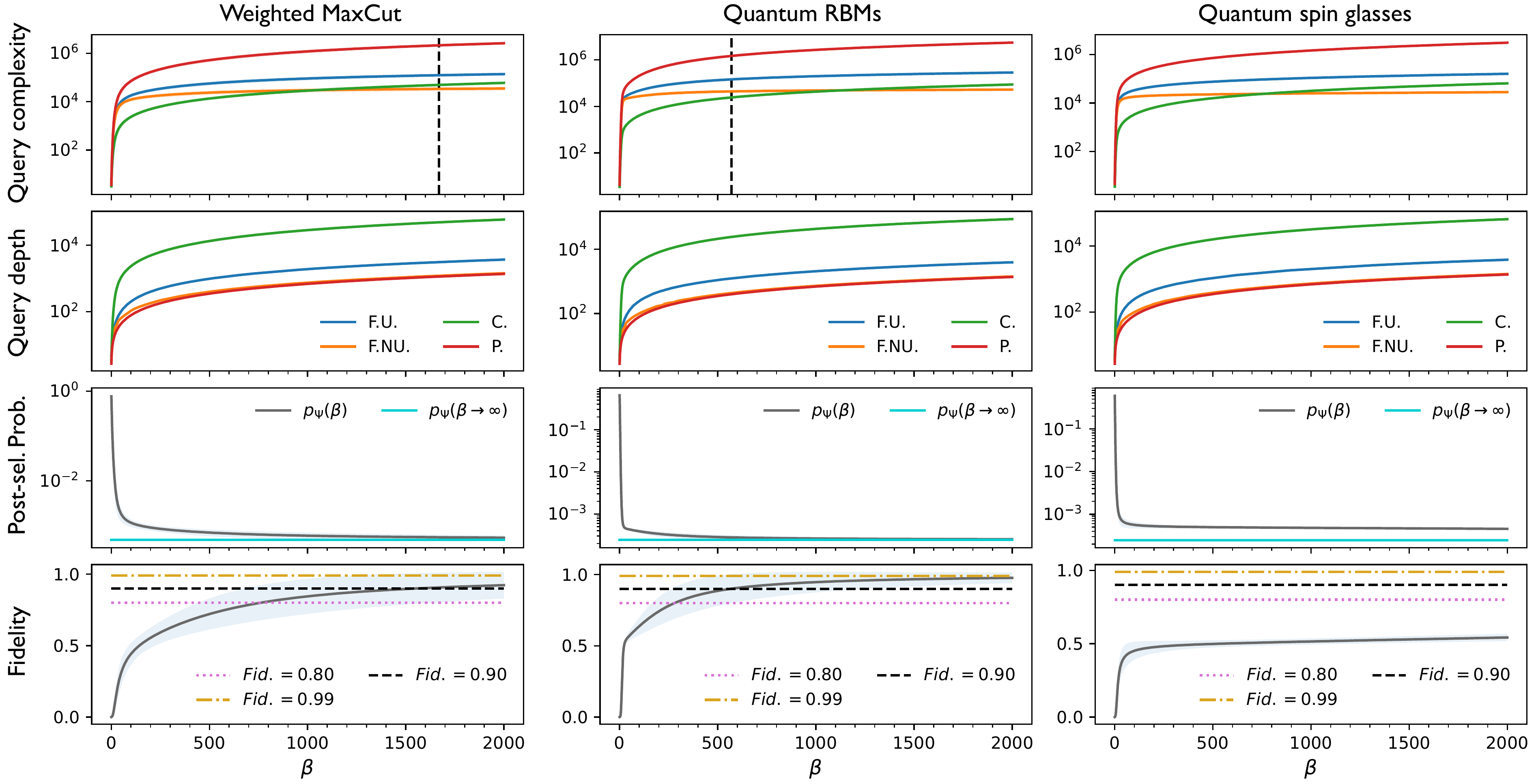}
  \caption{{{\bf Quantum Gibbs-state samplers running on Primitive 1 for the 10\% hardest Hamiltonian instances: runtime, circuit depth, success probability, and fidelity with the ground state versus inverse temperature}. The Hamiltonian classes are the same as in Fig.~\ref{fig:Frag_QITE_P1}; and the number of qubits and tolerated error are also $N=12$ and $\varepsilon=10^{-3}$, respectively. Out of the 1000 random instances used  for each class in Fig.~\ref{fig:Frag_QITE_P1}, we select the 100 ones with the smallest spectral gap. Solid curves represent the means over that $10\%$. The color code is the same as in Fig.~\ref{fig:Frag_QITE_P1} too: red for probabilistic, green for coherent, blue for fragmentation with uniform schedule $S_{r}$ for the best $r$, and orange for fragmentation with a schedule $S_{r,a}$ as in Eq.\ \eqref{eq:schedule_ansatz} for the best $r$ and $a$. From top to bottom: the upper panels show the overall query complexity, the second row the average query depth, the third one the post-selection probability, and the lower pannels the average fidelity between the ideal Gibbs state at inverse temperature $\beta$ and the ground state in question.
  %The last two rows of panels focus only on fragmentation with non-uniform schedules (the orange case in the two upper rows) \la{Thais, please fix color code in the last two rows}.
  The black vertical dashed lines in the first row represent the value $\beta^{(\text{hard})}_{0.9}$ of inverse temperature required for a modest average fidelity of $0.9$ with the ground state, defined by the intersection between the black horizontal dashed lines and the blue solid curves in the last row. In the third panel of the first row there is no black vertical dashed line because, for this class, $\beta^{(\text{hard})}_{0.9}$ is way beyond the range of $\beta$'s shown, as is clear also from the third panel of the last row. In the first row, both for the first and third panels, the critical value $\beta_\text{c}$ at which fragmentation outperforms coherent QITE is below $\beta^{(\text{hard})}_{0.9}$; whereas, for the second panel, $\beta^{(\text{hard})}_{0.9}$ is smaller than $\beta_\text{c}$ but the difference in query complexity at $\beta^{(\text{hard})}_{0.9}$ is already very small. If the target fidelity is instead 0.99, all three classes require inverse temperatures greater than $\beta_\text{c}$.
This implies that fragmented QITE is superior to coherent QITE for highly-relevant ranges of $\beta$ for ground state preparation, e.g. Interestingly, in addition, the  complexities (and therefore also $\beta_\text{c}$) and depths shown in the first and second rows are very similar to those of the average case in Fig.~\ref{fig:Frag_QITE_P1}. In the third row, to guide the eye, the cyan solid line shows the asymptotic limit $p_\Psi(\beta\rightarrow\infty)$ of the post-selection probability. These plots give a first glimpse of how far away from ground state the Gibbs-state in question is. Finally, in the last row we make this observation more quantitative with the fidelity as figure of merit. There, horizontal dashed lines represent the fidelity values $0.8$ (pink), $0.9$ (black), and $0.99$ (gold).} }
  \label{fig:Frag_QITE_P1_hard}
\end{figure*}

%%%%%%%%%%%%%%%%%%%%%%%%%%%%%%%%%%%%%%%%%%%%%%%%%%%%%%%%%%%%%%%%%%%%%%%%%%%%%%%%%%~

%%%%%%%%%%%%%%%%%%%%%%%%%%%%%%%%%%%%%%%%%%%%%%%%%%%%%%%%%%%%%%%%%%%%%%%%%%%%%%%%%%%%%%%%%%%%%%
\section{Fragmented Gibbs-state samp{ling with $P_1$ on the 10\% hardest} instances}
\label{app:QuITEprimitive1_hard}
{Here, we analyze the performance of $P_1$ for Gibbs-state sampling but restricted to the 10\% hardest Hamiltonian instances from those used to produce Fig. \ref{fig:Frag_QITE_P1}. The hardest instances are given by the choices of $H$ with the smallest spectral gap $\Delta$ between the first excited and the ground states. This is due to the fact that the inverse temperature required to approximate the ground state (up to any constant target precision) scales as $\beta=\mathcal{O}(1/\Delta)$. That is, the lower the gap is, the higher the query complexity is.

The analysis is shown in Fig. \ref{fig:Frag_QITE_P1_hard}. Apart from the average run-time and circuit depth, as in Fig. \ref{fig:Frag_QITE_P1}, we also plot the evolution of the average post-selection probability and fidelity with the ground state. Both post-selection probabilities and ground-state fidelities are calculated via brute-force diagonalization of each Hamiltonian. As can be seen in the figure, for Weighted MaxCut and, especially, Quantum Spin Glasses, fragmented QITE becomes superior to coherent QITE well before the value $\beta^{(\text{hard})}_{0.9}$ needed for ground state preparation up to a modest fidelity 0.9. In turn, for Quantum RBMs, fragmented QITE becomes superior to coherent QITE after $\beta^{(\text{hard})}_{0.9}$, but the difference in their query complexity at $\beta^{(\text{hard})}_{0.9}$ is already very small. If, instead, the target fidelity is raised to 0.99, all three classes require inverse temperatures greater than $\beta_\text{c}$. This implies that fragmented QITE is superior to coherent QITE for ranges of $\beta$ that are highly relevant for ground state preparation. }

%If no previous knowledge of that state is available, the goal is guaranteed to be acomplished with high fidelity by preparing a Gibbs state  with enoughly high inverse temperature. To have an idea of how large $\beta$ should be, we can write down the fidelity  as (disconsidering QITE approximation error) 
%\begin{equation}
% \mathcal{F}(\beta)=\frac{e^{-\beta \lambda_0}}{\sum_\lambda e^{-\beta\lambda}}=\frac{1}{1 + \sum_{\lambda\neq\lambda_0}e^{-\beta(\lambda-\lambda_0)}}.
%\end{equation}
%High values of fidelity are obtained when the sum is close to zero. This happens for $\beta=\mathcal{O}(1/\Delta)$ with $\Delta=\lambda_1-\lambda_0$ the gap between the first excited state (energy $\lambda_1$) and the ground state.

%%%%%%%%%%%%%%%%%%%%%%%%%%%%%%%%%%%%%%%%%%%%%%%%%%%%%%%%%%%%%%%%%%%%%%%%%%%%%%%%%%
\section{Fragmented Gibbs-state samplers with $P_2$}
\label{app:QuITEprimitive2}

%%%%%%%%%%%%%%%%%%%%%%%%%%%%%%%%%%%%%%%%%%%%%%%%%%%%%%%%%%%%%%%%%%%%%%%
\begin{figure*}
  \includegraphics[width=1\textwidth]{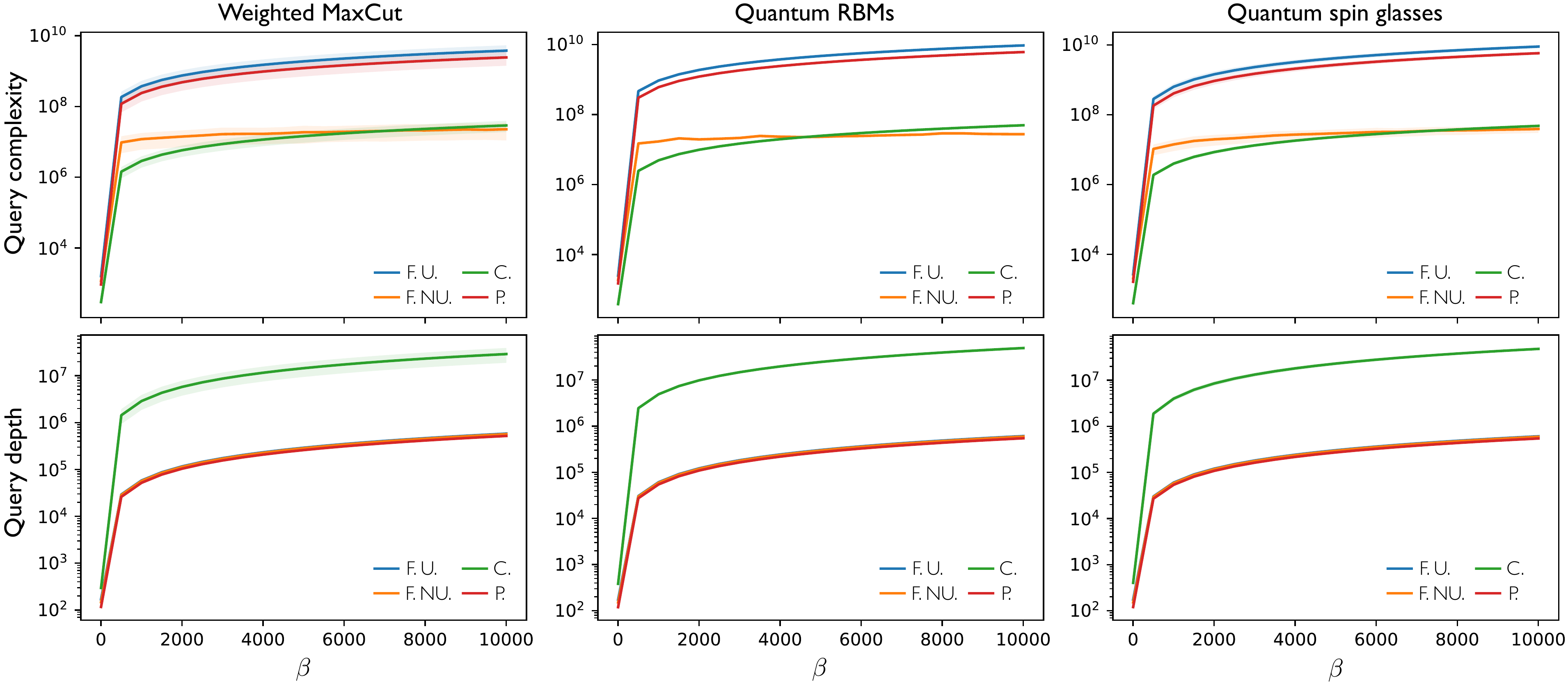}
  \caption{{\bf Runtimes and circuit depths of quantum Gibbs-state samplers running on Primitive 2 versus inverse temperature}. The color code is the same as in Fig.~\ref{fig:Frag_QITE_P1}: red corresponds to probabilistic, green to coherent, blue to fragmented with uniform schedule $S_{r}$ for the best $r$, and orange to fragmented with a schedule $S_{r,a}$ as in Eq.\ \eqref{eq:schedule_ansatz} for the best $r$ and $a$ (see also Fig.\ \ref{fig:Opt_schedules_P2}). The Hamiltonian classes are also the same as in Fig.~\ref{fig:Frag_QITE_P1}. For simplicity, here we use just a subset of 100 random instances from each class out of the 1000 used for Fig.~\ref{fig:Frag_QITE_P1}.
The examples correspond to $N=12$ qubits and an error $\varepsilon=10^{-3}$, but qualitatively identical behaviors are observed for all $N$ between 2 and 15 as well as for $\varepsilon=10^{-2}$ and $\varepsilon=10^{-1}$. Upper panels: average overall query complexity. %Both fragmented algorithms outperform the probabilistic one already at small $\beta$. In addition, 
As mentioned, also for $P_2$ does fragmentation with non-uniform schedule outperform coherent QITE, except the critical inverse temperature $\beta_c$ is now higher than in Fig.~\ref{fig:Frag_QITE_P1}. The dependance of $\beta_c$ with $N$ is shown and discussed in Fig.\ \ref{fig:Beta_crit_P2}.
Lower panels: average query depth. As with $P_1$, also for $P_2$ does coherent QITE lie orders of magnitude above probabilistic QITE; while fragmented QITE is almost identical to the latter.}
  \label{fig:Frag_QITE_P2}
\end{figure*}

%%%%%%%%%%%%%%%%%%%%%%%%%%%%%%%%%%%%%%%%%%%%%%%%%%%%%%%%%%%%%%%%%%%%%%%%%%%%%%%%%%

Here, we numerically study the performance of Alg. \ref{alg:Fragmented-QuITE-algorithm} at quantum Gibbs-state sampling, as in Sec.\ \ref{sec:GS_sampling}, but for Primitive 2. 
In Fig.~\ref{fig:Frag_QITE_P2} we show the average runtimes and circuit depths. The schedule optimization algorithms follow
the same approach as for $P_1$.
The critical inverse temperatures for $P_2$ are shown in Fig.~\ref{fig:Beta_crit_P2}. For $P_2$ we achieve good
optimization stability only from larger values of $N$. This can be understood from the fact that for small $N$ the minimum success probability ($2^N$) is relatively large. Because the subnormalization factors $\alpha_k$ of each fragment only depend on the inverse temperature (not on $N$) and has a minimum value, the optimization seeks for large $\beta$ such that the success probability of each fragment approaches $2^N$. 
Finally, in Fig.~\ref{fig:Opt_schedules_P2} we show the optimal fragmentation schedules for $P_2$. 

%%%%%%%%%%%%%%%%%%%%%%%%%%%%%%%%%%%%%%%%%%%%%%%%%%%%%%%%%%%%%%%%%%%%%%%
\begin{figure}[t!]
\centering
  \includegraphics[width=0.7\columnwidth]{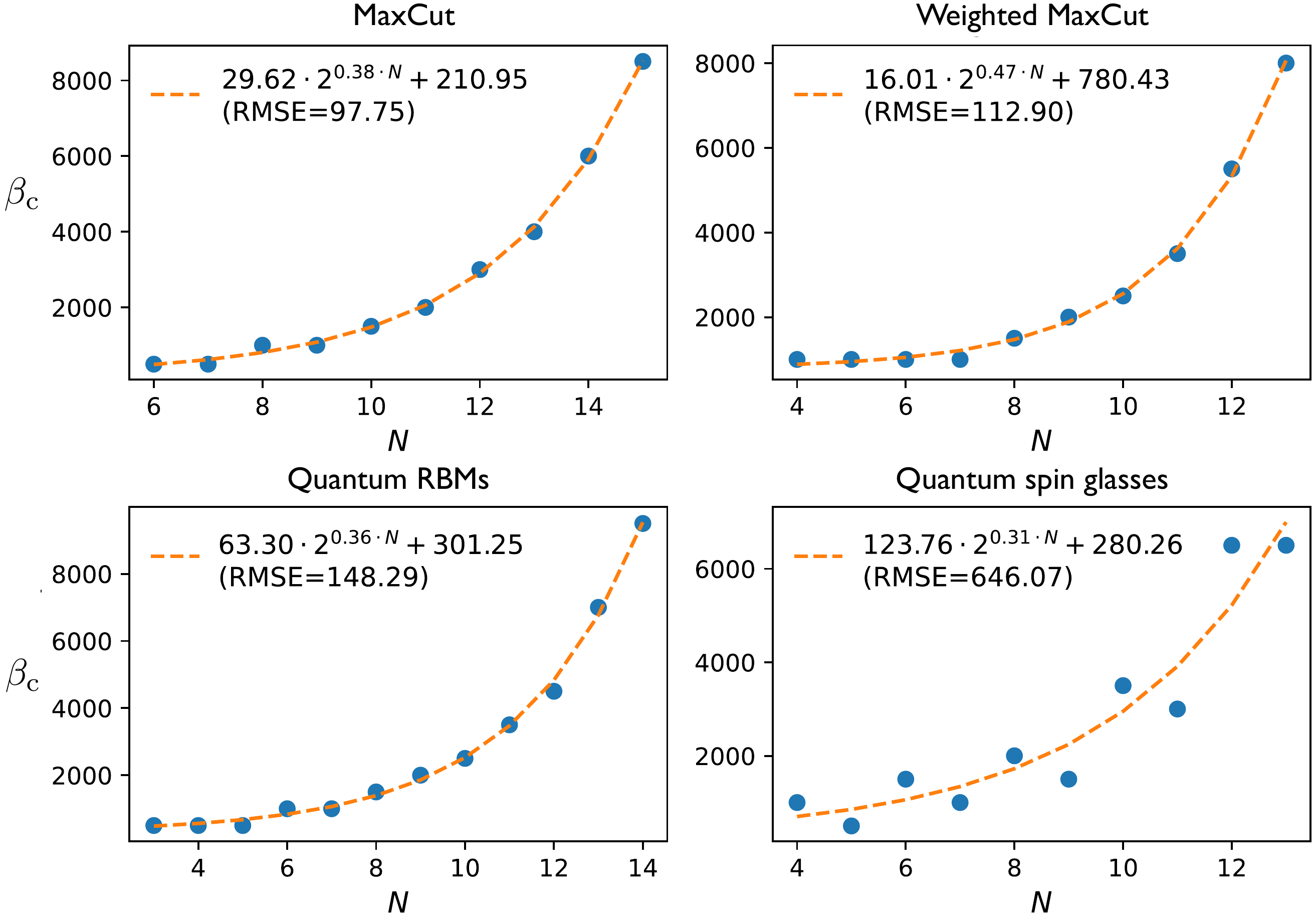}
  \caption{{\bf Critical inverse temperatures for $P_2$ versus number of qubits}. The tolerated error and Hamiltonian classes are the same as in Fig.~\ref{fig:Beta_crit}. Blue dots are the means over 100 instances from each class, whereas dashed orange curves their fits over the same Ansatz as in Fig.~\ref{fig:Beta_crit}. The fit results are shown in the insets. The scalings of $\beta_{\rm c}$ with $N$ are similar to those for $P_1$ but with somewhat higher pre-factors and additive constants. The fact that $\beta_{\rm c}$'s are higher than for $P_1$ comes from the non-unit factors $\alpha_k$ inside $n_l$ in Eq.\ \eqref{eq:F_overall_query}.
  \label{fig:Beta_crit_P2}}
\end{figure}
%%%%%%%%%%%%%%%%%%%%%%%%%%%%%%%%%%%%%%%%%%%%%%%%%%%%%%%%%%%%%%%%%%%%%%%

%%%%%%%%%%%%%%%%%%%%%%%%%%%%%%%%%%%%%%%%%%%%%%%%%%%%%%%%%%%%%%%%%%%%%%%
\begin{figure}[ht!]
\centering
  \includegraphics[width=0.7\columnwidth]{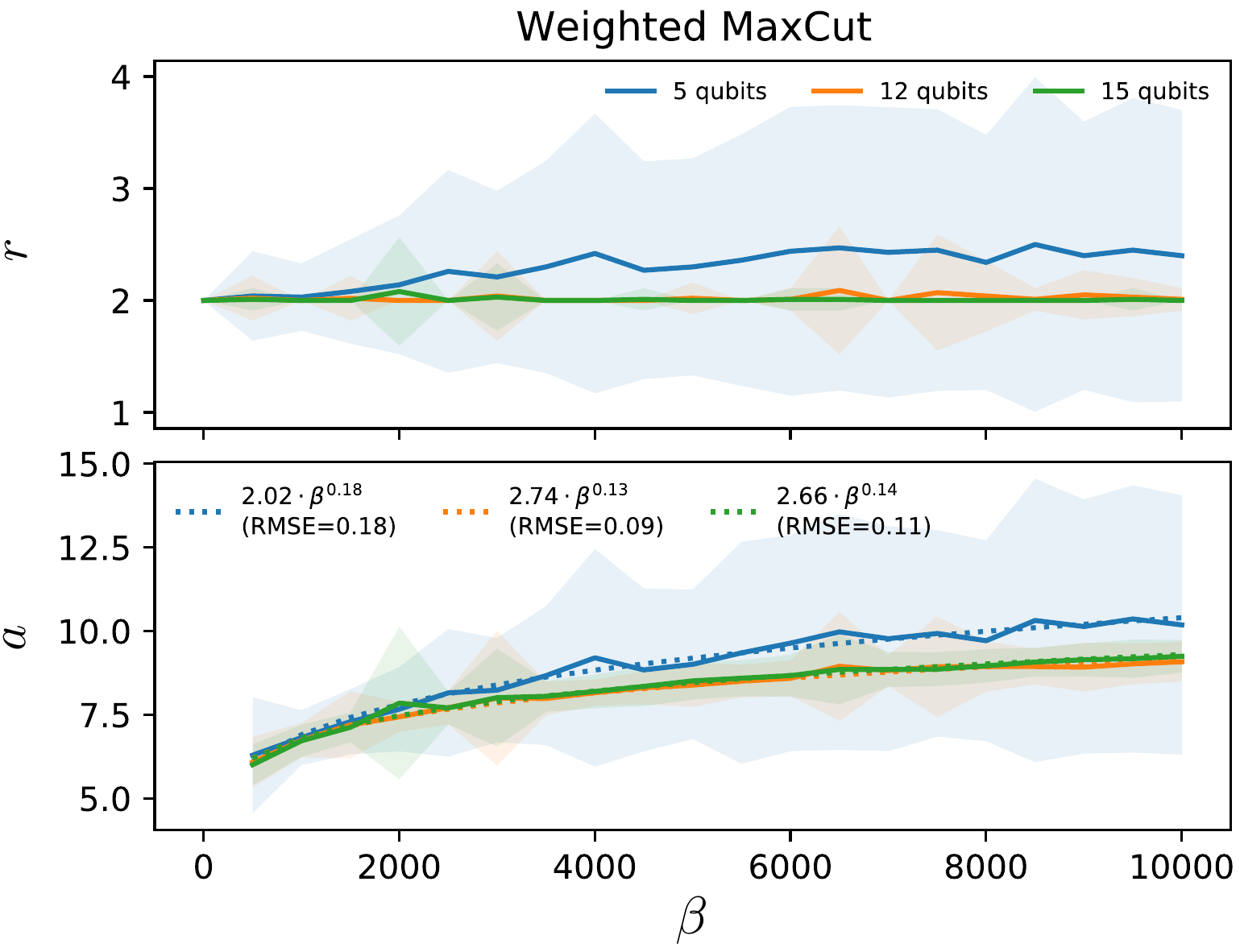}
  \caption{{\bf Optimal fragmentation schedules for $P_2$ versus inverse temperature}. System sizes and color code are the same as in Fig.~\ref{fig:Opt_schedules}: $N=5$ (blue), $N=12$ (orange), and $N=15$ (green). Solid curves represent the means over 100 random weighted-MaxCut Hamiltonians, whereas (the thicknesses of) shaded curves are the standard deviations. The error is $\varepsilon=10^{-3}$. Qualitatively identical behaviors are observed for all $N$ between 2 and 15 as well as for $\varepsilon=10^{-2}$ and $\varepsilon=10^{-1}$; and the same holds for the other Hamiltonian classes. For uniform schedules $S_{r,1}$ (not shown in the figure), a constant $r=2$ is observed to minimize $Q_{S_r}$ but the resulting complexity does not reach $Q_{{\rm coh}}$ over the domain scanned in Fig.\ \ref{fig:Frag_QITE_P2}. 
  The upper and lower panels respectively show the optimal $r$ and $a$ for non-uniform schedules $S_{r,a}$. The dotted curves in the lower panel  represent fits over the Ansatz $a(\beta)=A\, \beta^{\eta}$, with $A,\eta\in\mathbb{R}$. The fit results are shown in the inset. As with $P_1$, also here is $r$ constant not only with $\beta$ but also with $N$, remarkably.}
  \label{fig:Opt_schedules_P2}
\end{figure}
%%%%%%%%%%%%%%%%%%%%%%%%%%%%%%%%%%%%%%%%%%%%%%%%%%%%%%%%%%%%%%%%%%%%%%%

%%%%%%%%%%%%%%%%%%%%%%%%%%%%%%%%%%%%%%%%%%%%%%%%%%%%%%%%%%%%%%%%%%%%%%%%%%%%%%%%%%
\section{QSP achievability -- method 1}
\label{app:qsp}

The following Lemma (adapted from Theorem 5 of Ref. \cite{gilyen2019}) presents the set of conditions for a pair of real
polynomials $\mathscr{B}(\cos\theta)$ and $\mathscr{D}(\cos\theta)$ to be achieved as the real/imaginary parts of $B(\cos\theta)$
and $D(\cos\theta)$ of Eq.\,\eqref{eq:poly}. This conditions have been summarised in Eqs. \eqref{achievabilityCond1} and \eqref{achievabilityCond2} in the main text.

\begin{lem}\label{cond}Given two polynomials $\mathscr{B}(x)$, $\mathscr{D}(x):[-1,1]\rightarrow\mathbb{R}$,
and $q\in\mathbb{N}$ even, there exists ${\boldsymbol{\Phi}_1}=\left(\phi_{1},\cdots,\phi_{q+1}\right)\in\mathbb{R}^{q+1}$
such that $\mathscr{B}(\cos\theta)=\textrm{Re}[B(\cos\theta)]$ (or
$\mathscr{B}(\cos\theta)=\textrm{Im}[B(\cos\theta)]$) and $\mathscr{D}(\cos\theta)=\textrm{Re}[D(\cos\theta)]$
(or $\mathscr{D}(\cos\theta)=\textrm{Im}[D(\cos\theta)]$) for all
$\theta\in[-\pi,\pi]$, with $B$ and $D$ as in Eq.\ \eqref{eq:poly},
if and only if $\mathscr{B}$ and $\mathscr{D}$ satisfy:

(i) $\mathscr{B}(x)=\sum_{k=0}^{q/2}b_k^{\prime}x^{2k}$ and $\mathscr{D}(x)=\sum_{k=0}^{q/2-1} d_{k}^{\prime}x^{2k+1}$;

(ii) $\forall x\in[-1,1]:\quad \mathscr{B}^{2}(x)+\left(1-x^{2}\right)\mathscr{D}^{2}(x)\leq1$.

Moreover, the achievable functions can also be expressed as $\mathscr{B}(\cos\theta)=\sum_{k=0}^{q/2} b_k\cos(2k\theta)$
and $\sin\theta\,\mathscr{D}(\cos\theta)=\sum_{k=1}^{q/2} d_{k}\sin\left(2k\theta\right)$.\end{lem}

The last part of the Lemma, presented before in Eq.\,\eqref{achievabilityCond2}, can be obtained from condition (i) using the properties of
the Chebyshev polynomials $T_{k}(\cos\theta)=\cos(k\theta)$ and $U_{k}(\cos\theta)=\sin\left((k+1)\theta\right)/\sin\theta$. Condition (ii) is equivalent to Eq.\,\eqref{achievabilityCond1} in the main text.
Given that the desired polynomials are achievable by QSP, the set
of rotation angles ${\boldsymbol{\Phi}_1}$ can be computed classically in time
$\mathcal{O}\left(\text{poly}(q)\right)$ \cite{Low2016PRX,Haah2019product,chao2020finding,dong2020efficient}.

%%%%%%%%%%%%%%%%%%%%%%%%%%%%%%%%%%%%%%%%%%%%%%%%
%%%%%%%%%%%%%%%%%%%%%%%%%%%%%%%%%%%%%%%%%%%%%%%%
\section{Chebyshev expansions }
\label{App:Chebyshev}

In this appendix we briefly approach polynomial approximations to continuous functions using Chebyshev polynomials.

It is known that \cite{Fraser1965}, whenever
a function $f$ is continuous and bounded on the interval
$[-1,1]$, it is endowed with a convergent Chebyshev series such that 
\begin{equation}
f(\lambda)=\sum_{k=0}^{\infty}\,b_k T_{k}(\lambda),\label{eq:expansion-1}
\end{equation}
with coefficients 
\begin{equation}
b_0=\frac{1}{\pi}\int_{-1}^{1}\frac{f(\lambda)}{\sqrt{1-\lambda^{2}}}d\lambda\label{eq:ak-0}
\end{equation}
\begin{equation}
b_k=\frac{2}{\pi}\int_{-1}^{1}\frac{f(\lambda)T_{k}(\lambda)}{\sqrt{1-\lambda^{2}}}d\lambda,\quad k=1,2,\cdots\label{eq:ak-1}
\end{equation}
The truncated version of \eqref{eq:expansion-1} up to order $q$
is close to the optimal polynomial approximation of degree $q$ for $f(\lambda)$ in the interval
$[-1,1]$ \cite{Fraser1965}. 
The following Lemma (Theorem 2.1 from \cite{elliot1987}) gives the maximal error $\varepsilon$ of this approximation
for a given degree $q$. 

\begin{lem}\label{erro}Let $\tilde{f}(\lambda)$ be the 
polynomial approximation on $[-1,1]$ to a function $f\in C^{(q+1)}[-1,1]$ obtained by truncating its Chebyshev expansion up to order $q$.
Then 
\begin{equation}\label{errorCheb}
\varepsilon=\underset{-1\leq\lambda\leq1}{\max}|f(\lambda)-\tilde{f}(\lambda)|=\frac{\left|f^{(q+1)}(\xi)\right|}{2^{q}(q+1)!},
\end{equation}
for some $\xi\in(-1,1)$.\end{lem}

Another way of obtaining a nearly optimal polynomial approximation
of order $q$ is to construct it from Chebyshev polynomials interpolation,
which can be done in polynomial time by solving a linear system \cite{elliot1987}. The latter construction is based on conveniently choosing a list of points $\lambda_0,\cdots,\lambda_{q}$, typically related to the zeros or extreme points of $T_{q+1}(\lambda),$ and finding the linear combination of all degree-less-than-$q$ Chebyshev polynomials which interpolates $f(\lambda_0),\cdots,f(\lambda_q)$. 
The error formula of Eq.\,\eqref{errorCheb} is not exclusive of the Chebyshev truncated series. It is also valid for the special interpolating polynomials just described and for the optimal polynomial approximation, the difference being the point $\xi$ where the $(n+1)$-th derivative must be evaluated. 

%%%%%%%%%%%%%%%%%%%%%%%%%%%%%%%%%%%%%%%%
%%%%%%%%%%%%%%%%%%%%%%%%%%%%%%%%%%%%%%%%

\section{A recipe for pulses -- QSP method 1}
\label{app:QSPpulses}

This appendix is based on the proof of Theorem 3 of Ref. \cite{gilyen2019}. We restrict ourselves to the case of even $q$, as this is what we use throughout our operator-function design algorithm. Henceforth, we write $x=\cos\theta$. Here we only show how to get the QSP pulse sequence $\boldsymbol{\Phi}$ given the complex polynomials $B(x)$ and $D(x)$, such that Eq.\ \eqref{eq:poly} is attained. We refer the reader to Theorem 5 of Ref. \cite{gilyen2019} for a proof of existence  of $B$ and $D$ given the real polynomials $\mathscr{B}(x)$ and $\mathscr{D}(x)$ satisfying Lemma \ref{cond}. The proof of the referred theorem is also an algorithm to calculate  $B$ and $D$ whose classical computational runtime scales as $\mathscr{O}(\text{poly}(q))$.

Given two polynomials $B(x)$ and $D(x)$ as in Eq.\ \eqref{eq:poly}, it can be checked from multiplying single qubit rotation operators that they satisfy:
\begin{enumerate}[i)]
 \item $\text{deg}(B)\leq q$ and $\text{deg}(D)\leq q-1$;
 \item $B$ is even and $D$ is odd;
 \item $\forall x \in [-1,1]: \quad |B(x)|^2+(1-x^2)|D(x)|^2=1$.
\end{enumerate}

Now we build a recurrence which allows one to find a pulse sequence $\boldsymbol{\Phi}=(\phi_1,\cdots,\phi_{q+1})$ such that, given two polynomials 
\begin{equation}
 B(x)=\sum_{k=0}^{l/2}b_{k} x^{2k} \quad \text{and}\quad D(x)=\sum_{k=0}^{l/2-1}d_{k} x^{2k+1}
\end{equation}
satisfying the three above conditions -- thus $l\leq q$ -- they are  obtained as the entries of the operator   $\mathcal{R}_1\left(\theta,{\boldsymbol{\Phi}}\right)=e^{i\phi_{q+1}{Z}}\prod_{k=1}^{q/2}R_1(-\theta,\phi_{2k})R_1(\theta,\phi_{2k-1})$. Observe that, if $l=0$, then actually $B$ is constant and condition iii) implies that $D(x)=0$ and $|B(x)|=1$. Thus, we can write $B(x)=e^{i\phi}$ and the pulse sequence $\boldsymbol{\Phi}=(\phi,0,\cdots,0)$ does the job.

Because the pair of polynomials satisfies condition iii), they always define a unitary operator 
\begin{equation}
\mathcal{R}_{BD}=\left(\begin{array}{cc}
{B}(x) & i\sqrt{1-x^2}{D}(x)\\
i\sqrt{1-x^2}{D}^*(x) & {B}^{*}(x)
\end{array}\right).
\end{equation}
 Let us also define the new polynomials $\tilde{B}(x)$ and $\tilde{D}(x)$ by 
%\begin{equation}
\begin{equation}\label{Rtilde}
 \left(\begin{array}{cc}
\tilde{B}(x) & i\sqrt{1-x^2}\tilde{D}(x)\\
i\sqrt{1-x^2}\tilde{D}^*(x) & \tilde{B}^{*}(x)
\end{array}\right)=
\mathcal{R}_{BD}R_1(-\theta,-\phi_1),
\end{equation}
%\end{equation}
with $\theta =\cos^{-1}(x)$. Performing the operator multiplications, one gets
\begin{equation}
 \begin{split}
  \tilde{B}(x) & =e^{-i\phi_1}\left[xB(x)+e^{2i\phi_1}(1-x^2)D(x)\right]\\
  \tilde{D}(x) & =e^{-i\phi_1}\left[-B(x)+e^{2i\phi_1}xD(x)\right].
 \end{split}
\end{equation}
It is easy to check that defining $e^{2i\phi_1}=\frac{b_{l/2}}{d_{l/2-1}}$ makes the higher order terms in both $\tilde{B}(x)$ and $\tilde{D}(x)$ vanish. This is a valid choice for $\phi_1$ because, if $l>0$, condition iii) implies that the polynomial $|B(x)|^2+(1-x^2)|D(x)|^2$ is constant, and thus each of its coefficients accompanying a greater-than-zero power of $x$ must vanish. In particular, it holds for the highest order term if $\left|\frac{b_{l/2}}{d_{l/2-1}}\right|=1$. Hence, the coefficients of $B(x)$ and $D(x)$ can be rearranged to give
\begin{equation}
 \tilde{B}(x)=\sum_{k=0}^{l/2-1}\tilde{b}_{k} x^{2k+1} \quad \text{and}\quad \tilde{D}(x)=\sum_{k=0}^{l/2-2}\tilde{d}_{k} x^{2k},
\end{equation}
which are an odd degree $l-1$ and an even degree $l-2$ polynomials, respectively. Thus, the right multiplication of $R_1(-\theta,-\phi_1)$ has the power of decreasing the degree of the polynomials in  $\mathcal{R}_{BD}$ by one. Because $\tilde{B}(x)$ and $\tilde{D}(x)$ compose a unitary operator, they also satisfy condition iii). Therefore, the aforementioned  argument  can be used again to verify that $|\tilde{b}_{l/2-1}|=|\tilde{d}_{l/2-2}|$ and  define $e^{2i\phi_2}=-\frac{\tilde{b}_{l/2-1}}{\tilde{d}_{l/2-2}}$. Multiplying Eq.\ \eqref{Rtilde} from the right by $R_1(\theta,-\phi_2)$ cancels out the highest order terms again and results in an operator 
\begin{equation}
 \mathcal{R}_{\tilde{\tilde{B}}\tilde{\tilde{D}}}=\left(\begin{array}{cc}
\tilde{\tilde{B}}(x) & i\sqrt{1-x^2}\tilde{\tilde{D}}(x)\\
i\sqrt{1-x^2}\tilde{\tilde{D}}^*(x) & \tilde{\tilde{B}}^{*}(x)
\end{array}\right)
\end{equation}
whose entries are polynomials with $\text{deg}(\tilde{\tilde{B}})=l/2-2$ and $\text{deg}(\tilde{\tilde{D}})=l/2-3$. This procedure, alternating the sign of $\theta$, can be repeated $l$ times. In each step, the rotation angle $\phi_k$ is determined from the highest degree coefficients of the polynomials from previous step and a rotation $R_1(\pm\theta,-\phi_k)$ is applied. It decreases the degree of the polynomials by one,  until we finally arrive at a constant polynomial which determines the last phase $\phi_{l+1}$ for the $Z$ rotation which is applied along with no $X$ rotation. If $l< q$, then $\phi_{l+2}=\cdots=\phi_{q+1}=0$ and we can always choose $q=l$ and not waste pulses. This algorithm provides the exact pulses  using $\mathcal{O}(q^2)$  multiplications and additions. 

%%%%%%%%%%%%%%%%%%%%%%%%%%%%%%%%%%%%%%%%
%%%%%%%%%%%%%%%%%%%%%%%%%%%%%%%%%%%%%%%%

\section{Qubitization}
\label{app:qubitization}
As mentioned in the main text, it is possible to transform a block-encoding oracle of a Hamiltonian into a new block-encoding under which each eigenvalue of $H$ is associated with an invariant subspace with dimension two. This transformation is called qubitization of the oracle and is the subject of this appendix.

Consider a unitary $U_H$ perfectly block-encoding a Hamiltonian $H$ using $|\mathcal{A}_{U_H}|$ ancillas. {Recall that} $\ket{\lambda}\in \mathbb{H}_{\mathcal{S}}$ is the $\lambda$-th eigenstate of $H$. Then the oracle $O_1$ of Def. \ref{def:block_enc_or} obtained from the action of $U_H$ controled by a single-qubit ancilla $\mathcal{A}_{\text{ctrl}}$ is also  a perfect block-encoding of $H$. The total number of ancillas used by this oracle is $|\mathcal{A}_{O_1}|=|\mathcal{A}_{U_H}|+|\mathcal{A}_{\text{ctrl}}|$. Denoting $\ket{0_{\lambda}}=\ket{\lambda}\ket{0}\in\mathbb{H}_{SA_{O_1}}$,
the action  of $O_1$ produces 
$O_1\ket{0_{\lambda}}=\lambda\ket{0_{\lambda}}+\sqrt{1-\lambda^{2}}\ket{0_{\lambda}^{\perp}}$,
with $\ket{0_{\lambda}^{\perp}}\in\mathbb{H}_{\mathcal{SA}_{O_1}}$
and $\braket{0_{\lambda}}{0_{\lambda}^{\perp}}=0$. Although $O_1\ket{0_{\lambda}}$ is in the subspace spanned by $\{\ket{0_{\lambda}},\ket{0_{\lambda}^{\perp}}\}$,
further applications of $O_1$ will not in general produce higher powers of $H$ as the state will leak out of this subspace. In order to avoid this, a convenient oracle for producing Hamiltonian functions would preserve the subspaces $\mathbb{H}_\lambda\coloneqq\mathrm{span}\{\ket{0_{\lambda}},\ket{0_{\lambda}^{\perp}}\}\subset \mathbb{H}_{\mathcal{SA}_{O_1}}$ for each eigenvalue $\lambda$. It was shown in ref. \cite{Low2019hamiltonian}
that it is always possible to obtain from $O_1$ its qubitized version
$O_1'$ which also block-encodes the Hamiltonian, i.e. $\bra{0}O_1'\ket{0}=H$ with $\ket{0}\in\mathbb{H}_{\mathcal{SA}_{O_1}}$,
and can be represented as 
\begin{equation}\label{iterate}
O_1'=\bigoplus_{\lambda}\left[e^{-i\theta_{\lambda}}\ket{0_{\lambda+}}\bra{0_{\lambda+}}+e^{i\theta_{\lambda}}\ket{0_{\lambda-}}\bra{0_{\lambda-}}\right],
\end{equation}
in the subspace $\bigoplus_{\lambda}\mathbb{H}_\lambda$ of $\mathbb{H}_{\mathcal{SA}_{O_1}}$.
Here $\theta_{\lambda}=\cos^{-1}(\lambda)$, and $\ket{0_{\lambda\pm}}=\left(\ket{0_{\lambda}}\pm i\ket{0_{\lambda}^{\perp}}\right)/\sqrt{2}$, such that $\braket{0_{\lambda+}}{0_{\lambda -}}=0$. As can be noticed, for any $\lambda$,  the subspace $\mathbb{H}_\lambda$ is invariant under $O_1'$. Furthermore, in that two-dimension subspace, $O_1'$ is isomorphic to the qubit rotation $e^{-i\theta_\lambda {Y}_\lambda}$, with $Y_\lambda$ the second Pauli operator in $\mathbb{H}_\lambda$.

\begin{figure}[t]
\begin{centering}
\includegraphics[width=0.8\columnwidth]{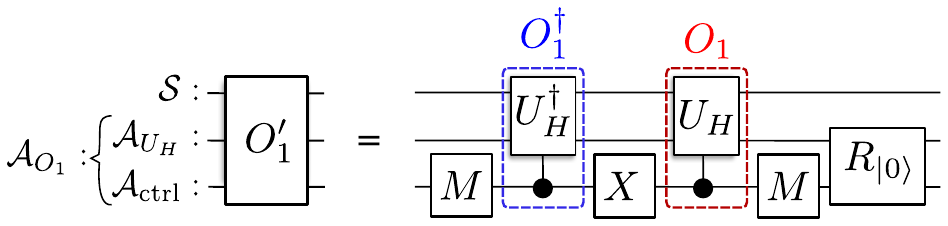}
\par\end{centering}
\caption{\label{fig:qubitization} Circuit implementing the transformation of the Hamiltonian oracle $O_1$ of Def. \ref{def:block_enc_or}  into its qubitized version $O_1'$. It requires %adding a new qubit ancilla $\mathcal{A}_{\text{qubit}}$, 
one use of $O_1$ and one of its inverse, three single qubit gates applied to the control ancilla $\mathcal{A}_{\text{ctrl}}$, and $\mathcal{O}(|\mathcal{A}_{O_1}|)$ additional gates to implement the reflection $R_{\ket{0}}$ about the ancilla state $\ket{0}\in \mathbb{H}_{\mathcal{A}_{O_1}}$.   }
\end{figure}

Applying $O_1'$ %requires enlarging the Hilbert space further by adding only one qubit ancilla, it 
uses only one query to $O_1$
and one to $O_1^{\dagger}$, and  $\mathcal{O}(|{\mathcal{A}_{O_1}}|)$ additional
quantum gates. The precise prescription given in \cite{Low2019hamiltonian} for obtaining $O_1'$ as in Eq.\ \eqref{iterate} is: %enlarge the Hilbert space adding one qubit ancilla $\mathcal{A}_{\text{qubit}}$; 
take the operators $U_H^{\prime}= O_1\otimes\ketbra{+}{+}+ O_1^{\dagger}\otimes\ketbra{-}{-} $, $R_{\ket{0}}=\mathds{1}_{\mathcal{S}}\otimes\left(2\ketbra{0}{0}-\mathds{1}_{\mathcal{A}_{O_1}} \right)$, and $S=\mathds{1}_{\mathcal{SA}_{U_H}}\otimes (M{X}M)$ on $\mathbb{H}_{\mathcal{SA}_{O_1}}$, then
\begin{equation}\label{eq:qubitization}
 O'_1=R_{\ket{0}}SU_H^{\prime}.
\end{equation}
A equivalent circuit representation of $O_1'$ is found in Fig.\ \ref{fig:qubitization}. We refer the reader to Ref. \cite{Low2019hamiltonian} for a demonstration of the equality between Eqs. \eqref{iterate} and \eqref{eq:qubitization}. 

 %%%%%%%%%%%%%%%%%%%%%%%%%%%%%%%%%%%%%%%%%

\section{Proof of Lemma \ref{be_function}}
\label{app:proofbe_function}

In order to prove Lemma \ref{be_function}, we first announce and prove the following lemma,  which states how to implement operator Chebyshev series in general.

\begin{lem}\label{lem:chebBlockEncoding}(Chebyshev series from block-encoding Hamiltonian oracles)  Given $\tilde{f}_q(\lambda)=\sum^{q/2}_{k=0} b_k \,T_k(\lambda)$, with $b_k\in \mathbb{R}$, such that  $|\tilde{f}_q(\lambda)|\leq 1 $ for all $ \lambda \in [-1,1]$, there is a ${\boldsymbol{\Phi}_1}=(\phi_1, \cdots,\phi_{q+1})\in\mathbb{R}^{q+1}$,  such that the unitary operator defined in Eq.\ \eqref{eq:Vvecphi} and implemented by Alg. \ref{alg:fBlockEncode} is a perfect block-encoding of $\tilde{f}_q(H)$, i.e. 
\begin{equation}
\bra{0} V_{{\boldsymbol{\Phi}_1}} \ket{0}=\sum_{\lambda}\tilde{f}_q(\lambda)\,\ket{\lambda}\bra{\lambda},\label{eq:berealfunction}
\end{equation}
with $\ket{0}$ in $\mathbb{H}_{\mathcal{A}}$. Moreover, the pulse sequence can be obtained classically in time $\mathcal{O}(\text{poly}(q))$.
\end{lem}

\begin{proof}[Proof of Lemma \ref{lem:chebBlockEncoding}]  Recalling that the identity operator does not produce any leakage of states out of the initial subspace they belong to, the identity operator in $\mathbb{H}_{\mathcal{SA}_{O_1}}$ can be effectively represented on the subspace $\bigoplus_\lambda \mathbb{H}_\lambda$ as
$\mathds{1}=\bigoplus_{\lambda}  \left(\ket{0_{\lambda+}}\bra{0_{\lambda+}}+\ket{0_{\lambda-}}\bra{0_{\lambda-}}\right)$.  Using this representation and Eq.\ \eqref{iterate}, it is straightforward to write the
operator $V_0$ of Eq.\ \eqref{eq:v0be}  as 
\begin{equation}
 \begin{split}
 V_0=\bigoplus_{\lambda} e^{-i\frac{\theta_\lambda}{2}}&\ket{0_{\lambda+}}\bra{0_{\lambda+}}\otimes e^{i\frac{\theta_\lambda}{2}X}+\\
  & e^{i\frac{\theta_\lambda}{2}}\ket{0_{\lambda-}}\bra{0_{\lambda-}}\otimes e^{-i\frac{\theta_\lambda}{2}X},
 \end{split}
\end{equation}
where $X$ is the first Pauli operator for an extra single-qubit ancilla we will denote by $\mathcal{A}_e$.
% Consider the operator $V_{\phi}=V_{0}\left(\mathds{1}\otimes e^{i\phi {Z}}\right)$ on $\mathbb{H}_{\mathcal{SA}}$ with $V_{0}=\mathds{1}\otimes\ketbra{+}{+}+O_1^\prime\ketbra{-}{-}$.
% Using Eq.\ \eqref{iterate}, it can be effectively represented in the subspace $\mathbb{S}_{\mathcal{SA}}\subset \mathbb{H}_{\mathcal{SA}}$ spanned by $\{\ket{0_{\lambda+}},\bra{0_{\lambda-}}\}_{\lambda}\otimes \{\ket{0},\ket{1}\}$
% as 
% \begin{equation}
% \begin{split}
% V_{\phi}&=\bigoplus_{\lambda}  \left(\ket{0_{\lambda+}}\bra{0_{\lambda+}}+\ket{0_{\lambda-}}\bra{0_{\lambda-}}\right)\otimes \left(\ket{+}\bra{+}e^{i\phi{Z}}\right)\\
% &\left(e^{-i\theta_{\lambda}}\ket{0_{\lambda+}}\bra{0_{\lambda+}}+e^{i\theta_\lambda}\ket{0_{\lambda-}}\bra{0_{\lambda-}}\right)\otimes \left(\ket{-}\bra{-}e^{i\phi{Z}}\right),
% \end{split}
% \label{eq:vphi1}
% \end{equation}
% from which, by reorganizing the terms, we get 
From $V_0$ one can also obtain a convenient representation for $V_{\phi}=V_{0}\left(\mathds{1}\otimes e^{i\phi {Z}}\right)$ in $\bigoplus_\lambda \mathbb{H}_\lambda\otimes\mathbb{H}_{\mathcal{A}_e}$, it reads
\begin{equation}
\begin{split}V_{\phi}=\bigoplus_{\lambda} & \left(e^{-i\frac{\theta_{\lambda}}{2}}\ket{0_{\lambda+}}\bra{0_{\lambda+}}\otimes R_1\left(\frac{\theta_{\lambda}}{2},\phi\right)\right.\\
 & \left.+e^{i\frac{\theta_{\lambda}}{2}}\ket{0_{\lambda-}}\bra{0_{\lambda-}}\otimes R_1\left(-\frac{\theta_{\lambda}}{2},\phi\right) \right),
\end{split}
\label{eq:vphi}
\end{equation}
with the qubit rotations acting on $\mathcal{A}_e$ previously defined as $R_1\left(\pm\frac{\theta_{\lambda}}{2},\phi\right)=e^{\pm i\frac{\theta_\lambda}{2}{X}}e^{i\phi{Z}}$.
%In order to remove the undesired phase factor $e^{\pm i\frac{\theta_{\lambda}}{2}}$, let us also define the operators $\bar{V}_{0}=\ket{-}\bra{-}\otimes\mathds{1}+\ket{+}\bra{+}\otimes W^{\dagger}$ and $\bar{V}_{\phi}=\bar{V}_{0}\left(e^{i\phi{Z}}\otimes\mathds{1}\right)$
If the operator defined as $\bar{V}_{\phi^\prime}= V_{0}^\dagger\left(\mathds{1}_{\mathcal{SA}_O}\otimes e^{i\phi^\prime {Z}}\right)$ is applied in the sequence, it removes the undesired phase factors $e^{\pm i\frac{\theta_{\lambda}}{2}}$.
 A direct consequence of Eq.\ \eqref{eq:vphi} is that the operator $V_{\boldsymbol{\Phi}_1}$ defined in Eq.\ \eqref{eq:Vvecphi} 
%$V^\prime_{{\boldsymbol{\Phi}}}=\left(\mathds{1}_{\mathcal{SA}_O}\otimes \left(M e^{i\phi_{0}{Z}}\right)\right)\prod_{k=1}^{q/2}V_{\phi_{2k-1}}\bar{V}_{\phi_{2k}} \left(\mathds{1}_{\mathcal{SA}_O}\otimes M \right)$,
% defined by the set of angles ${\boldsymbol{\Phi}}=\{\phi_{0},\cdots,\phi_{2q}\}$,
can be represented in $\bigoplus_\lambda \mathbb{H}_\lambda\otimes\mathbb{H}_{\mathcal{A}_e}$ as
\begin{equation}
\begin{split}V_{{\boldsymbol{\Phi}_1}}=\bigoplus_{\lambda} & \left[\ket{0_{\lambda+}}\bra{0_{\lambda+}}\otimes \left(M\,\mathcal{R}_1\left(\frac{\theta_{\lambda}}{2},{\boldsymbol{\Phi}_1}\right)M\right)\right.\\
 & \left.+\ket{0_{\lambda-}}\bra{0_{\lambda-}} \otimes \left(M\,\mathcal{R}_1\left(-\frac{\theta_{\lambda}}{2},{\boldsymbol{\Phi}_1}\right)M\right)\right],
\end{split}
\label{eq:vvecphi}
\end{equation}
where $\mathcal{R}_1\left(\pm\frac{\theta_{\lambda}}{2},{\boldsymbol{\Phi}_1}\right)$ is given in Eq.\ \eqref{eq:poly}.
In order to obtain functions with vanishing imaginary part, one can
initialize and post-select the ancillas in the state $\ket{0}\in \mathbb{H}_{\mathcal{A}}$,
which corresponds in $\mathbb{H}_{\mathcal{S}}$ to applying the operator
\begin{equation}
\bra{0}V_{{\boldsymbol{\Phi}_1}}\ket{0}=\sum_{\lambda}\text{Re}\left[B\left(\cos\frac{\theta_{\lambda}}{2}\right)\right]\,\ket{\lambda}\bra{\lambda},\label{eq:realfunction}
\end{equation}
where $B\left(\cos\frac{\theta_{\lambda}}{2}\right)$ is the the first
matrix element of $\mathcal{R}_1\left(\frac{\theta_{\lambda}}{2},{\boldsymbol{\Phi}_1}\right)$.
Because of the halved angle, the achievable functions can be written,
according to Lem. \ref{cond}, as 
\begin{equation}
\text{Re}\left[B\left(\cos\frac{\theta_{\lambda}}{2}\right)\right]=\sum_{k=0}^{q/2}b_k\cos(k\theta_{\lambda})=\sum_{k=0}^{q/2}b_k T_{k}(\lambda).\label{eq:cheb}
\end{equation}
The only real polynomial of interest, out of the four composing the QSP operator $\mathcal{R}_1\left(\frac{\theta_{\lambda}}{2},{\boldsymbol{\Phi}_1}\right)$,  is  $\text{Re}\left[B\left(\cos\frac{\theta_{\lambda}}{2}\right)\right]$, the other three can be determined as to keep achievability. Provided that $\text{Re}\left[B\left(\cos\frac{\theta_{\lambda}}{2}\right)\right]$  has the form in Eq.\,\eqref{eq:cheb}, the only remaining condition is that $\left|\text{Re}\left[B\left(\cos\frac{\theta_{\lambda}}{2}\right)\right]\right|\leq 1$. With the choice $\mathscr{D}(\lambda)=0$ and $\mathscr{B}(\lambda)=\tilde{f}_q(\lambda)$, condition \eqref{achievabilityCond1} is satisfied, assuring the existence of such ${\boldsymbol{\Phi_1}}$.

\end{proof}
 
 We are now able to prove Lemma \ref{be_function}, which is a direct consequence of  Lemmas \ref{lem:chebBlockEncoding} and \ref{erro}, and of the qubitization construction.
  
\begin{proof}[Proof of Lemma \ref{be_function}]
If the QSP sequence ${\boldsymbol{\Phi}_1}$ is taken as to reproduce the coefficients
of the finite Chebyshev expansion of a continuous target function $f$ (see App.\ \ref{App:Chebyshev})
\begin{equation}
\tilde{f}_q(\lambda)=\sum_{k=0}^{q/2}\,b_k T_{k}(\lambda),\label{eq:Chebexpansion}
\end{equation}
with $f(\lambda)=\tilde{f}_\infty(\lambda)$, then Eq.\,\eqref{eq:berealfunction} means that $V_{{\boldsymbol{\Phi}_1}}$ is a perfect block-encoding of $\tilde{f}_q(H)$ a (1,$\varepsilon'$)-block-encoding of $f(H)$, with $\varepsilon'$ given by the truncation  error $\text{max}_{\lambda\in[-1,1]}|{f}(\lambda)-\tilde{f}_q(\lambda)|$. $\tilde{f}_q(\lambda)$ can be the truncated Chebyshev expansion of $f(\lambda)$ or an interporlating series. $f$ is assumed to be an analytical function and, therefore, $\varepsilon'$ is related to the series order $q/2$ according to Lemma \ref{erro}. 
 The achievability condition \eqref{achievabilityCond1} is not a strong limitation as $f$ is bounded and we can always redefine $f^\prime(\lambda)=\frac{f(\lambda)}{f_{max}}$, with $f_{max}=\underset{-1\leq\lambda\leq1}{\max} f(\lambda)$.
Note that, even though $|\mathscr{B}(\lambda)|\leq1+\varepsilon'$, it can be rescaled
as to satisfy Eq.\ \eqref{achievabilityCond1} \cite{gilyen2019}.

% Take the Chebyshev series implemented by Lemma \ref{lem:chebBlockEncoding} such that $|f(\lambda)-\tilde{f}_q(\lambda)|\leq \varepsilon'$ for all $\lambda\in[-1,1]$. $\tilde{f}_q(\lambda)$ can be the truncated Chebyshev expansion of $f(\lambda)$ or an interporlating series. $f$ is assumed to be an analytical function and, therefore, $\varepsilon'$ is related to the series order $q/2$ according to Lemma \ref{erro}. 
 Each one of the $q$ operators $V_{\phi_k}$ or $\bar{V}_{\phi_k}$ in $V_{{\boldsymbol{\Phi}_1}}$ calls the qubitized oracle once which means calling the oracle $O_1$ twice.  Thus $2q$ queries are necessary in total. Moreover the number of gates per query is basicaly given by the number of gates used for qubitization times a constant factor due to the controlled action of $O_1'$. The control ancilla contributes with a constant small number of gates, such that the total number of gates per query is $\mathcal{O}(g_{O_1}+|\mathcal{A}_{O_1}|)$.
\end{proof}

%%%%%%%%%%%%%%%%%%%%%%%%%%%%%%%%%%%%%%%%%%%%%%%%%%%%%%%%

\section{Proof of Lemma \ref{lem:QuITEparity}}
\label{app:proofQuITEparity}
Our proof is by explicitly construction of $H_{\boldsymbol x}$. We use the same 2-sparse Hamiltonians used in the proof of the no-fast-forwarding theorem \cite{Berry2007,Berry2014,Berry2015a}. Recall that a Hamiltonian matrix is $d$-sparse if its columns (or rows) contain at most $d$ non-null entries each.
As a warm-up, let us first introduce necessary notation for the proof. 

Consider the $(N+1)$-dimensional symmetric subspace $\mathbb H_{\mathcal{S}_\mathrm{s}}\subset\mathbb H_\mathcal{S}$ spanned by kets $\{\ket j_\mathrm{s}\}_{j\in [N+1]}$, where $\ket j_\mathrm{s}$ is the permutationally-invariant superposition of $j$ qubits in state $\ket 1$ and $N\!-\!j$ ones in $\ket 0$. Let us define the ($\boldsymbol x$-independent) Hamiltonian  
\begin{equation}
H_{0}\coloneqq\sum_{i=1}^{N}\frac{X_i}{4\,N}\ ,
\label{eq:H0}
\end{equation}
where $X_i$ is the first Pauli operator of the $i$-th qubit in $\mathcal{S}$. Notice that $\|H_0\|\leq1$. Since $H_0$ is proportional to the total angular momentum operator along the $x$ direction, it couples each $\ket j_\mathrm{s}$ to $\ket{j-1}_\mathrm{s}$ and $\ket{j+1}_\mathrm{s}$ (or to one of them, depending on the value of $j$); and, importantly, it leaves  $\mathbb H_{\mathcal{S}_\mathrm{s}}$ invariant.  
Besides, the following overlap will be useful too:
%It will be useful to calculate the overlap between the fully excited state $\ket N_{\mathrm{s}}$ and the output state of
%a perfect QITE primitive with respect to $H_0$ on input state $\ket0_{\mathrm{s}}$\m{,}
\begin{equation}
\Big|\bra{N}_\mathrm{s}\ F_\beta(H_0) \ket0_\mathrm{s}\Big| = \left|\frac{1-e^{-\frac{\beta}{2N}}}2\right|^N  \ ,
\label{eq:overlapfinal}
\end{equation}
where the equality follows from the identity $e^{-\beta H_0}=\prod_{i=1}^{N}e^{-\frac{\beta X_i}{4\,N}}=\prod_{i=1}^{N}\big(\cosh(\frac{\beta}{4\,N})\mathds{1}-\sinh(\frac{\beta}{4\,N})X_i\big)$ and the fact that $\ket{N}_\mathrm{s}$ and $\ket{0}_\mathrm{s}$ are both product states.
%, and the identity $e^{-\eta X}=\cosh(\eta)\mathds{1}-\sinh(\eta)X$.
% and ${\rm sech}$ is the hyperbolic secant function.

%\m{With the results above, we can now move on to the actual proof, but first a remark is in order. The angular-momentum structure, with its well-known properties, is useful to describe the system above but, naturally, it is not necessary that the system be actually composed of spins. In fact, it suffices to consider $\mathcal S_\mathrm{s}\subset\mathcal S$ a system of qubits containing the $(N\!+\!1)$ symmetric $\ket j_\mathrm{s}$ states and a Hamiltonian $H_0$ that couples them as the one in Eq.\ \eqref{eq:H0} does.}

\begin{proof}[Proof of Lemma \ref{lem:QuITEparity}]
We begin by defining the Hamiltonian $H_{\boldsymbol{x}}$ that encodes ${\rm par}(\boldsymbol{x})$. It acts on the $(2N+2)$-dimensional Hilbert space $\mathbb H_{\mathcal{S}_\mathrm{s}}\otimes\mathbb H_{\mathcal{W}_\mathrm{p}}$, where $\mathcal{W}_\mathrm{p}$ is a single-qubit ancillary write register onto which our QITE-based parity-finding algorithm will write ${\rm par}(\boldsymbol{x})$. $H_{\boldsymbol x}$ is given by
\begin{multline}
H_{\boldsymbol x} \coloneqq \sum_{j=0}^N \frac{\sqrt{(N-j)(j+1)}}{4\,N} \times\big(\ket{j+1}\bra j_\mathrm{s}+ \text{H.c.}\big)\otimes {X_\mathrm{p}}^{x_j} \ ,
\label{eq:H_x}
\end{multline}
where $X_\mathrm{p}$ is the first Pauli matrix on $H_{\mathcal{A}_\mathrm{p}}$ and the short-hand notation $\ket{N+1}\coloneqq0$ is introduced.
Note that $H_{\boldsymbol x}$ couples each $\ket j_\mathrm{s}$ to $\ket{j-1}_\mathrm{s}$ and $\ket{j+1}_\mathrm{s}$, as $H_0$ in Eq.\ \eqref{eq:H0}, and flips the state of the write register conditioned on the value of the $j$-th bit $x_j$ of $\boldsymbol{x}$ (see Fig.\ \ref{fig:graph}). As a result, $\mathbb H_{\mathcal{S}_\mathrm{s}}\otimes\mathbb H_{\mathcal{W}_\mathrm{p}}$ is divided into two halves not coupled to one another. Importantly, $\ket0_\mathrm{s}\ket0_\mathrm{p}$ is always coupled to $\ket N_\mathrm{s}\ket{{\rm par}(\boldsymbol{x})}_\mathrm{p}$ and never to the wrong-parity state $\ket N_\mathrm{s}\ket{{\rm par}(\boldsymbol{x})\oplus1}_\mathrm{p}$. In addition, $H_{\boldsymbol x}$ is 2-sparse and $\|H_{\boldsymbol x}\|\leq1$ for all $\boldsymbol{x}\in\{0,1\}^N$.

%%%%%%%%%%%%%%%%%%%%%%%%%
%%%%%%%%%%%%%%%%%%%%%%%%%
\begin{figure}[tb]
\centering
\includegraphics[width=0.7\columnwidth]{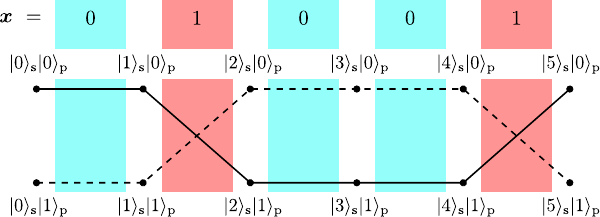}
\caption{Graph of the couplings present in Hamiltonian $H_{\boldsymbol{x}}$ for $\boldsymbol x=01001$. It encodes the string $\boldsymbol{x}$ in that a flipped-qubit coupling occurs conditioned on $x_j$, e.g.\ $\ket j_\mathrm{s}\ket1_\mathrm{p}$ is coupled to $\ket{j+1}_\mathrm{s}\ket1_\mathrm{p}$ or $\ket{j+1}_\mathrm{s}\ket0_\mathrm{p}$ if $x_j=0$ or $x_j=1$, respectively. Notice that initial state $\ket0_\mathrm{s}\ket0_\mathrm{p}$ is indirectly coupled to $\ket N_\mathrm{s}\ket{{\rm par}(\boldsymbol{x})}_\mathrm{p}=\ket 5_\mathrm{s}\ket{0}_\mathrm{p}$ (solid lines), not to $\ket N_\mathrm{s}\ket{{\rm par}(\boldsymbol{x})\oplus1}_\mathrm{p}=\ket 5_\mathrm{s}\ket{1}_\mathrm{p}$ (dashed lines).}%
\label{fig:graph}% 
\end{figure}
%%%%%%%%%%%%%%%%%%%%%%%%%
%%%%%%%%%%%%%%%%%%%%%%%%%
The following algorithm estimates the parity with a single application of a QITE-primitive on the fixed computational state $\ket0_\mathrm{s}\ket0_\mathrm{p}\ket0_{\mathcal A}$ and computational-basis measurements.
\begin{center}
\setlength{\fboxsep}{1pt}
%\framebox{
\begin{minipage}[t]{0.9\columnwidth}
\centering
\begin{algorithm}[H]\label{alg:prim2parity}\renewcommand{\thealgocf}{}
\SetAlgoLined
%\SetAlgorithmName{}{}{}
\caption{Bit-string parity from QITE}
\SetKwInOut{Input}{input}\SetKwInOut{Output}{output}
\SetKwData{Even}{even}
\Input{$\beta\neq0$, $\varepsilon'\geq0$, $\alpha\in(0,1]$, an $(\beta,\varepsilon',\alpha)$-QITE-primitive $P$ querying a block-encoding oracle for $H_{\boldsymbol x}$, an input state $\ket0_\mathrm{s}\ket0_\mathrm{p}\ket0_{\mathcal A}$}
\Output{$\mathrm{par}(\boldsymbol x)$ with probability $>1/2$} 
\BlankLine
apply $P$ on $\ket0_\mathrm{s}\ket0_\mathrm{p}\ket0_{\mathcal A}$\;
measure $\mathcal A$ in the computational basis\;
\lIf{outcome is $\ket0_{\mathcal A}$}{proceed}
\lElse{output $\mathrm{par}(\boldsymbol x)=0$ or 1 with probability $1/2$; end algorithm}
measure $\mathcal S_\mathrm{s}$ on basis $\{\ket j_\mathrm{s}\}_{j}$\;
\lIf{outcome is $\ket N_\mathrm{s}$}{proceed}
\lElse{output $\mathrm{par}(\boldsymbol x)=0$ or 1 with probability $1/2$; end algorithm}
measure $\mathcal W_\mathrm{p}$ in the computational basis and output the outcome\;
\end{algorithm}
\end{minipage}

\end{center}
The measurement on ${\mathcal A}$ is the usual QITE post-selection heralding whether $F_{\beta}(H_{\boldsymbol x})$ has been applied on $\ket0_\mathrm{s}\ket0_\mathrm{p}$. The measurement on $\mathcal S_\mathrm{s}$ heralds whether the fully-excited state $\ket{N}_\mathrm{s}$ has been obtained and, hence, whether the desired state $\ket{\mathrm{par}(\boldsymbol x)}_\mathrm{p}$ has been prepared in the write register $\mathcal W_\mathrm{p}$. 
%Although these measurements can fail, the algorithm is not probabilistic: in case of such failures, one makes an equal-probability guess of $0$ or $1$  (coin-toss guess), and the final success probability takes all runs into account.
Let us analyze the overall success probability.

Consider first the restricted, simplified case of $\varepsilon'=0$. Applying a $(\beta,0,\alpha)$-QITE-primitive on $\ket0_\mathrm{s}\ket0_\mathrm{p}\ket0_{\mathcal A}$
and then measuring $\ket0_{\mathcal A}$ on $\mathcal A$ prepares a state whose overlap with $\ket N_\mathrm{s}\ket{{\rm par}(\boldsymbol{x})}_\mathrm{p}\ket0_{\mathcal A}$ is given by Eq.\ \eqref{eq:overlapfinal}, as a straightforward calculation shows, which is strictly greater than 0. In turn, its overlap with $\ket N_\mathrm{s}\ket{{\rm par}(\boldsymbol{x})\oplus1}_\mathrm{p}\ket0_{\mathcal A}$ is zero. Hence, conditioned on a correct post-selection on $\mathcal A$ and on measuring $\ket N_\mathrm{s}$ too, the success probability of obtaining the correct parity on $\mathcal W_\mathrm{p}$ is unit. If any other measurement outcome is obtained, the algorithm simply makes a coin toss that guesses the correct parity with probability $1/2$. Therefore, for any non-null post-selection probability, the overall success probability of getting  ${\rm par}(\boldsymbol{x})$ is strictly greater than $1/2$.

For $\varepsilon'>0$, not any non-null overlap with $\ket N_\mathrm{s}\ket{{\rm par}(\boldsymbol{x})}_\mathrm{p}\ket0_{\mathcal A_P}$ suffices any longer, since the overlap with $\ket N_\mathrm{s}\ket{{\rm par}(\boldsymbol{x})\oplus1}_\mathrm{p}\ket0_{\mathcal A_P}$ is now also nonzero. More precisely, we need to demand that, conditioned on measuring $\ket0_{\mathcal A_P}$ and $\ket N_\mathrm{s}$, the correct-parity output has greater probability than the incorrect one:
\begin{multline}
\left|\bra N_\mathrm{s}\bra{{\rm par}(\boldsymbol{x})}_\mathrm{p}\bra0_{\mathcal A_P} \ket{\tilde\Psi_\beta(\boldsymbol x)}\right| > \\
\left|\bra N_\mathrm{s}\bra{{\rm par}(\boldsymbol{x})\oplus1}_\mathrm{p}\bra0_{\mathcal A_P} \ket{\tilde\Psi_\beta(\boldsymbol x)} \right|
\ ,
%c_\mathrm{ov}[\tilde\varrho_{\beta}(\boldsymbol x),\ket{{\rm par}(\boldsymbol{x})}_\mathrm{p}] > c_\mathrm{ov}[\tilde\varrho_{\beta}(\boldsymbol x),\ket{{\rm par}(\boldsymbol{x})\oplus1}_\mathrm{p}]
%\ ,
\label{eq:higheroverlap}
\end{multline}
where $\ket{\tilde\Psi_\beta(\boldsymbol x)}$ is the state after the QITE primitive. Since the primitive generates an $(\alpha,\varepsilon')$-block-encoding of the QITE propagator $F_\beta(H_{\boldsymbol x})$, we following bounds must hold
\begin{subequations}
\begin{align}
\big|\bra N_\mathrm{s}\bra{{\rm par}(\boldsymbol{x})}_\mathrm{p}\bra0_{\mathcal A_P} \ket{\tilde\Psi_\beta(\boldsymbol x)}\big|&		\geq \label{eq:bound1}\\
\alpha|\bra N_\mathrm{s}\bra{{\rm par}(\boldsymbol{x})}_\mathrm{p}\ F_\beta(H_{\boldsymbol x})& \ket0_\mathrm{s}\ket0_\mathrm{p}|-\varepsilon'\ ,
%c_\mathrm{ov} [\tilde\varrho_{\beta}(\boldsymbol x),\ket{{\rm par}(\boldsymbol{x})}_\mathrm{p}] \geq c_\mathrm{ov} [\varrho_{\beta}(\boldsymbol x),\ket{{\rm par}(\boldsymbol{x})}_\mathrm{p}]-\varepsilon'\ ,
\nonumber
\end{align}
and
\begin{equation}
\left|\bra N_\mathrm{s}\bra{{\rm par}(\boldsymbol{x})\oplus1}_\mathrm{p}\bra0_{\mathcal A_P}\ \ket{\tilde\Psi_\beta(H_{\boldsymbol x})}\right|\leq\,\varepsilon'\ .
%c_\mathrm{ov} [\tilde\varrho_{\beta}(\boldsymbol x),\ket{{\rm par}(\boldsymbol{x})\oplus1}_\mathrm{p}]\leq\,\varepsilon'\ .
\label{eq:bound2}
\end{equation}
\label{eq:bounds}\end{subequations}
From these, we see that Eq.\ \eqref{eq:higheroverlap} is fulfilled if we impose
\begin{equation}
\alpha|\bra N_\mathrm{s}\bra{{\rm par}(\boldsymbol{x})}_\mathrm{p}\ F_\beta(H_{\boldsymbol x}) \ket0_\mathrm{s}\ket0_\mathrm{p}|-\varepsilon' > \varepsilon' \ .
%c_\mathrm{ov} [\varrho_{\beta}(\boldsymbol x),\ket{{\rm par}(\boldsymbol{x})}_\mathrm{p}]-\varepsilon'>\varepsilon' \ .
\label{eq:psi_id}
\end{equation}
%In other words, if this equation is obeyed, QITE constitutes an algorithm for parity with success probability $>1/2$. 
Using again the fact that $|\bra N_\mathrm{s}\bra{{\rm par}(\boldsymbol{x})}_\mathrm{p}\ F_\beta(H_{\boldsymbol x}) \ket0_\mathrm{s}\ket0_\mathrm{p}|$ equals the overlap in Eq.\ \eqref{eq:overlapfinal}, one can straightforwardly see that Eq.\ \eqref{eq:psi_id} is equivalent to Eq.\ \eqref{eq:Paritycond}.
\end{proof}

Finally, we notice that the factor $4$ in the denominator of Eq.\ \eqref{eq:H_x} is not necessary for the above proof argument. However, it guarantees that $\|H_{\boldsymbol x}\|\leq1$, which is necessary for $H_{\boldsymbol x}$ to be block-encodable, as required for Lemma \ref{lem:fromU_xtoH_x}.

%%%%%%%%%%%%%%%%%%%%%%%%%%%%%%%%%%%%%%%%%%%%%%%%%%%%%%%%%%%%%%%%%%%%%%%%%%%%%%%%%%
%%%%%%%%%%%%%%%%%%%%%%%%%%%%%%%%%%%%%%%%%%%%%%%%%%%%%%%%%%%%%%%%%%%%%%%%%%%%%%%%%%
\section{Proof of Lemma \ref{lem:fromU_xtoH_x}}
\label{app:prooffromU_xtoH_x}
The proof is constructive. We design a quantum circuit (see Fig.\ \ref{fig:circuitHfromU}) that generates a block-encoding $U_{H_{\boldsymbol x}}$ of $H_{\boldsymbol x}$ from a single application of $U_{\boldsymbol x}$.
%%%%%%%%%%%%%%%%%%%%%%%%%%%%%%%%

\begin{figure}[tb!]%
\centering
\includegraphics[width=0.9\columnwidth]{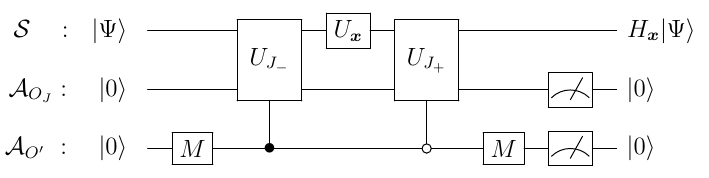}%
\caption{Generation of a block-encoding $U_{H_{\boldsymbol x}}$ of $H_{\boldsymbol x}$ from a single query to $U_{\boldsymbol x}$. $\mathcal A_{O_J}$ is the ancillary register required to block-encode $J_\pm$ and $\mathcal A_{O'}$ is the single-qubit ancilla required to create the linear combination in Eq.~\eqref{eq:HasU}.}%
\label{fig:circuitHfromU}%
\end{figure}
%%%%%%%%%%%%%%%%%%%%%%%%%%%%%%%%

\begin{proof}[Proof of Lemma \ref{lem:fromU_xtoH_x}]
We first need to define the total angular-momentum ladder operators $J_\pm$. In our notation, $J_+\ket j_s\coloneqq\sqrt{(N-j)(j+1)}\ket{j+1}_s$ and $J_-\ket j_s\coloneqq\sqrt{(N-j+1)j}\ket{j-1}_s$. Now, using Eqs.\ \eqref{eq:def_U_x} and \eqref{eq:H_x}, note that
\begin{equation}
H_{\boldsymbol x}	= \frac{J_+U_{\boldsymbol x} + U_{\boldsymbol x}J_-}{4\,N} \ .
\label{eq:HasU}
\end{equation}
The ladder operators can also be written in terms of Pauli operators as $J_\pm=\sum_{i=0}^{N-1} X_i\pm iY_i$. Since they are a linear combination of $2\, N$ unitaries, standard methods \cite{Low2019hamiltonian} can be used to produce a block encoding $U_{J_\pm}$ of $J_\pm/(2N)$ using an ancillary register $\mathcal A_{O_J}$ with $|\mathcal A_{O_J}|=\mathcal{O}\big(\log(N)\big)$ qubits. This requires no query to $U_{\boldsymbol x}$.

Then, the circuit in Fig.~\ref{fig:circuitHfromU} shows how to generate a perfect block-encoding $U_{H_{\boldsymbol x}}$ of $H_{\boldsymbol x}$ from one query to $U_{J_+}$, one query to $U_{J_-}$, and one query to $U_{\boldsymbol x}$, and using a single-qubit ancila $\mathcal A_{O'}$. 
Finally, a block-encoding oracle controlled-$U_{H_{\boldsymbol x}}$ for $H_{\boldsymbol x}$ can be obtained by replacing every gate in the figure with its version controlled by an extra single-qubit ancilla $\mathcal A_{O''}$ (not represented in Fig.~\ref{fig:circuitHfromU}). In the notation of Def. \ref{def:block_enc_or}, it reads $\mathcal A_{O_1}=\{\mathcal A_{O_J},\mathcal A_{O'},\mathcal A_{O''}\}$. This extra control does not alter the number of queries to $U_{\boldsymbol x}$. 
\end{proof}
%\bibliographystyle{plain}
%\bibliography{notas}

\printbibliography

\end{document}